\documentclass[12pt]{article}

\usepackage{subfigure}

 \usepackage[utf8]{inputenc}
  \usepackage[noTeX]{mmap}
  \usepackage[T1]{fontenc}
\usepackage{authblk}

\usepackage{latexsym}
\usepackage{graphics}
\usepackage{amsmath}
\usepackage{xspace}
\usepackage{amssymb}
\usepackage{psfrag}
\usepackage{pst-all}
\usepackage{color}

\definecolor{Red}{rgb}{1,0,0}
\definecolor{Blue}{rgb}{0,0,1}
\definecolor{Olive}{rgb}{0.41,0.55,0.13}
\definecolor{Yarok}{rgb}{0,0.5,0}
\definecolor{Green}{rgb}{0,1,0}
\definecolor{MGreen}{rgb}{0,0.8,0}
\definecolor{DGreen}{rgb}{0,0.55,0}
\definecolor{Yellow}{rgb}{1,1,0}
\definecolor{Cyan}{rgb}{0,1,1}
\definecolor{Magenta}{rgb}{1,0,1}
\definecolor{Orange}{rgb}{1,.5,0}
\definecolor{Violet}{rgb}{.5,0,.5}
\definecolor{Purple}{rgb}{.75,0,.25}
\definecolor{Brown}{rgb}{.75,.5,.25}
\definecolor{Grey}{rgb}{.5,.5,.5}

\newcommand{\ignore}[1]{\relax}

\definecolor{Red}{rgb}{1,0,0}
\definecolor{Blue}{rgb}{0,0,1}
\definecolor{Olive}{rgb}{0.41,0.55,0.13}
\definecolor{Green}{rgb}{0,1,0}
\definecolor{MGreen}{rgb}{0,0.8,0}
\definecolor{DGreen}{rgb}{0,0.55,0}
\definecolor{Yellow}{rgb}{1,1,0}
\definecolor{Cyan}{rgb}{0,1,1}
\definecolor{Magenta}{rgb}{1,0,1}
\definecolor{Orange}{rgb}{1,.5,0}
\definecolor{Violet}{rgb}{.5,0,.5}
\definecolor{Purple}{rgb}{.75,0,.25}
\definecolor{Brown}{rgb}{.75,.5,.25}
\definecolor{Grey}{rgb}{.5,.5,.5}
\definecolor{Pink}{rgb}{1,0,1}
\definecolor{DBrown}{rgb}{.5,.34,.16}
\definecolor{Black}{rgb}{0,0,0}

\usepackage{bbm}
\usepackage{mathrsfs}
\usepackage{epsfig}
\usepackage{soul,color}
\usepackage{graphicx}
\usepackage{amsfonts}
\usepackage{amsthm}
\usepackage[figuresright]{rotating}
\usepackage{dcolumn}
\usepackage{physics}
\usepackage{float}
\usepackage{bm}
\usepackage{verbatim}
\usepackage[normalem]{ulem}
\usepackage[colorlinks,linkcolor=blue,anchorcolor=blue,citecolor=blue,urlcolor=blue]{hyperref}
\usepackage{mathtools}

\usepackage{algorithm}
\usepackage{algorithmic}

\usepackage{xcolor}
\newtheorem{theorem}{Theorem}
\newtheorem{lemma}[theorem]{Lemma}
\newtheorem{conjecture}[theorem]{Conjecture}
\newtheorem{proposition}[theorem]{Proposition}
\newtheorem{corollary}[theorem]{Corollary}
\newtheorem{definition}{Definition}[section]

\newcommand{\E}{\mathbb{E}}

\newcommand{\C}{\mathbb{C}}

\newcommand{\R}{\mathbb{R}}
\newcommand{\Z}{\mathbb{Z}}

\definecolor{ao}{rgb}{0.0, 0.0, 1.0}

\renewcommand{\P}{\mathbb{P}}
\newcommand{\B}{\mathbb{B}}

\newcommand{\indicator}{{\bf{1}}}

\newcommand{\mcA}{\mathcal{A}}
\newcommand{\mcB}{\mathcal{B}}

\newcommand{\mcE}{\mathcal{E}}
\newcommand{\mcF}{\mathcal{F}}

\newcommand{\mcH}{\mathcal{H}}
\newcommand{\mcI}{\mathcal{I}}
\newcommand{\mcJ}{\mathcal{J}}

\newcommand{\mcS}{\mathcal{S}}

\newcommand{\tvar}{\operatorname{Var}}

\begin{document}

\title{Circuit complexity lower bounds for quantum spin glasses}

\author{
Omar Al-Ghattas\thanks{Eric and Wendy Schmidt Center, Broad Institute of MIT and Harvard, \href{mailto:oag@mit.edu}{oag@mit.edu}}
\qquad
David Gamarnik\thanks{Operations Research Center and Center for Statistics and Data Sciences, IDSS, Massachusetts Institute of Technology,
\href{mailto:gamarnik@mit.edu}{gamarnik@mit.edu}
}}

\maketitle

\begin{abstract}
A central question in quantum information theory concerns the circuit-complexity of states arising from standard models of quantum many-body systems. We study this question for quantum \(p\)-spin glasses, random Hamiltonians whose interaction terms act on \(p\)-tuples of qubits through Pauli strings. Recent work~\cite{anschuetz2025bounds} showed that, for these models, the optimum energy is separated from the best energy achievable by unentangled, or product, states. This leaves unresolved whether even depth-one circuits could prepare near-ground states, since such circuits can already generate entanglement.

We rule out this possibility by showing that the entanglement needed to close the product-state gap cannot be generated at shallow depth. When the average interaction degree grows with $n$, we prove a logarithmic-depth lower bound: for all sufficiently large fixed $p$, any circuit preparing an $n$-qubit state whose normalized energy is within a fixed positive constant of the optimum must have depth at least $\Omega_p(\log n)$. In the bounded-average-degree regime, we prove the corresponding fixed-depth obstruction: for every fixed depth $D$, a sufficiently large degree prefactor rules out depth-$D$ preparation of near-ground states. Both results hold uniformly over circuits using an arbitrary number of ancilla qubits.

Conceptually, our results give an obstruction in the spirit of the No Low-Energy Trivial States (NLTS) problem of Freedman and Hastings~\cite{freedman2014quantum}, but for random quantum spin glasses rather than code-based Hamiltonians such as~\cite{anshu2023nlts}, for which ground states can be prepared by polynomial-size circuits. Specifically, in our context,
the Hamiltonian is a random, generally noncommuting $p$-spin glass, and there is no known efficient procedure for preparing its near-ground states. This change in setting opens a probabilistic route to NLTS-like questions: we recast state-preparation lower bounds for random quantum Hamiltonians as uniform control of Gaussian processes indexed by shallow circuits.

\vspace{.1in}

\end{abstract}

\tableofcontents

\section{Introduction}
Quantum state complexity has emerged as a central topic at the interface of theoretical computer science, quantum computing, and several areas of physics, including the physics of
black holes. This prominence is due in part to remarkable connections between quantum many-body systems, conformal field theories, and gravity theory through the AdS/CFT holographic duality~\cite{baiguera2026quantum}. The notion of quantum state complexity is conceptually simple: given a quantum state $|\psi\rangle$, it is the size of the smallest quantum circuit, measured for instance, by depth or gate count, that prepares $|\psi\rangle$ either exactly or to a prescribed accuracy.

Standard counting arguments show that, if an $n$-qubit state is chosen at random, then with overwhelming probability it requires an exponentially large circuit to prepare even approximately.

A key question is whether such complexity lower bounds can be established for physically or computationally natural families of states, such as ground states of succinctly described Hamiltonians, meaning states of minimal energy, or states obtained by Schrödinger evolution under such Hamiltonians.

For example, it is conjectured that certain quantum states arising in the AdS/CFT theory have high state complexity, and that the growth of this complexity is related to the growth of gravitational wormholes~\cite{susskind2016computational, baiguera2026quantum}. 

This notion is inherently quantum. State complexity—or more precisely state \emph{preparation} complexity, which we abbreviate as SPC—has no meaningful analogue for deterministic classical states: any fixed classical configuration can be prepared from any other by a depth-one circuit of single-qubit gates. SPC should also not be confused with \emph{algorithmic} complexity, where the task is to design a circuit or algorithm that works uniformly for an entire class of problems whose descriptions are supplied as inputs.

In theoretical computer science, much of the interest in SPC is motivated by the quantum PCP (q-PCP) conjecture, the quantum analogue of the classical PCP theorem and a major open problem in quantum complexity theory. Although we do not state the conjecture formally here, one of its consequences, if it were established, would be the existence of quantum Hamiltonians consisting of linearly many local terms for which every near-ground state has superconstant SPC.
Freedman and Hastings isolated this implication as the No Low-Energy Trivial State (NLTS) conjecture~\cite{freedman2014quantum}, which was recently proved in breakthrough work by Anshu, Breuckmann and Nirkhe~\cite{anshu2023nlts}. Their construction uses an intricate Hamiltonian associated with recent constructions of quantum LDPC codes~\cite{PanteleevKalachev2022GoodQLDPC, LeverrierZemor2022QuantumTanner}, as discussed further below. Notably, an exact ground state of this Hamiltonian can be prepared by a polynomial-size Clifford circuit.

This leaves open the question of SPC for ground states that arise more naturally, such as those associated with random optimization problems. Quantum spin glasses, which will be defined formally in the next section, and related models such as the Sachdev-Ye-Kitaev (SYK) model, provide a prominent class of examples. In this paper, we ask: what is the SPC of near-ground states of quantum $p$-spin glasses? 

In contrast to classical \(p\)-spin models, which have been studied extensively
in both physics~\cite{MezardParisiVirasoro} and
mathematics~\cite{TalagrandBook, panchenko2013Sherrington}, quantum
\(p\)-spin models are much less understood. Nevertheless, several useful
bounds on their ground-state energies are known and will play an important
role in this paper. For the mean-field quantum \(p\)-spin Hamiltonian, the
natural normalization of the maximum energy is by \(\sqrt n\). With this
normalization, known results imply that, for fixed \(p\) and large \(n\), the
ground-state energy is at least of order \(3^{p/2}/p\), up to universal
constants~\cite{anschuetz2025strongly}. In contrast, the best energy achievable
by product, or unentangled, states is at most $(1+o_p(1))\sqrt{2\log p}$ as \(p\to\infty\)~\cite{anschuetz2025bounds}. Thus the product-state scale is
much smaller than the ground-state scale for large \(p\). In particular,
near-ground states must exhibit nontrivial entanglement. The natural question
is then how much entanglement, or, more precisely, how much SPC, is required.

\subsection{Summary of models and main results}

Schematically, the quantum \(p\)-spin Hamiltonian is a random \(p\)-local Hamiltonian formed by summing independent Gaussian-weighted interaction terms over choices of \(p\) qubits. Each interaction term is a Pauli string, with one of the three nonidentity Pauli operators acting on each of the selected qubits. Our main results establish circuit-depth lower bounds in three regimes, distinguished by the number of $p$-spin interactions present in the random Hamiltonian.

Before describing the regimes, we fix a convention. We formulate the spin-glass optimization problem in maximum-energy form: throughout this paper, ``near-ground'' means close to the maximum energy of the Hamiltonian. This is only a sign convention since replacing the Hamiltonian by its negative gives the usual minimum-energy formulation. In each of these regimes, the same separation between the maximum energy and the best product-state energy described above applies. 

The first regime is the mean-field setting. Here all
$\binom{n}{p}$ possible choices of $p$ spins are present; for each such choice, the Hamiltonian includes all $3^p$ Pauli local interaction terms supported on those spins, with independent Gaussian coefficients.

The second is the diluted, or sparse-hypergraph, version of the same model.
Starting from the mean-field model, each of the $\binom{n}{p}$ possible
choices of $p$ spins is retained independently with probability $\kappa_n$,
and deleted otherwise. When such a choice of spins is retained, all local
interaction terms supported on those spins are included. Thus the expected
number of retained $p$-spin interaction supports is
$\kappa_n\binom{n}{p}$. We take $\kappa_n=n^{-\beta}$ with $\beta<p-1$.
With this choice of parameters, every spin participates in $\omega(1)$
retained interactions in expectation. Equivalently, the underlying support
hypergraph has growing average degree.

In the third regime, the same dilution mechanism is used at the
bounded-average-degree scale. Namely, we set
$\kappa_{n,\zeta}=\zeta n^{-\beta}$ with $\beta=p-1$, for a large constant $\zeta$.
With this scaling, each spin participates in a bounded number of retained
interactions in expectation, and the expected total number of retained
$p$-spin interaction supports is linear in $n$.

We now summarize our main results. In the mean-field and growing-average-degree regimes, we establish a logarithmic-in-\(n\) lower bound on circuit depth. More precisely, for all sufficiently large fixed \(p\), there
exists \(c_p>0\) such that the following holds with high probability over the
randomness defining the Hamiltonian. No circuit of depth at most a sufficiently
small constant times \(\log n\), even with arbitrary ancillas, can prepare a
state whose normalized energy is within \(c_p\) of the ground-state optimum.

At the bounded-average-degree scale, our result takes a fixed-depth form.
For every sufficiently large \(p\) and every prescribed depth \(D\), no
depth-\(D\) circuit, even with arbitrary ancillas, can prepare a state whose
normalized energy is within \(c_p\) of the ground-state optimum, with
probability tending to one in the iterated limit \(n\to\infty\) followed by
\(\zeta\to\infty\).

The common mechanism behind these statements is that shallow circuits cannot
asymptotically improve on the product-state optimum. Since the product-state
optimum is of order \(\sqrt{\log p}\), whereas the ground-state optimum is
separated from it for all sufficiently large \(p\), this yields the stated
lower bounds for near-ground-state preparation.

\subsection{Proof ideas}

We now describe the main ideas behind our proofs. It is instructive to first contrast our approach with a related result of~\cite{anschuetz2025strongly}. They show that near-ground states of the $p$-body SYK model have SPC polynomially large in the number of qubits. In particular, circuits preparing such states must have polynomially large depth. Their argument relies crucially on a special feature of the SYK model: for every fixed quantum state, its energy is a centered Gaussian random variable with variance at most $1/n^{\Theta(p/2)}$, where the implicit constant is universal. A union bound over all states preparable by circuits of size $s$ then shows that the largest energy attainable by such circuits is $o(1)$ unless $s$ is polynomially large in $n.$ 

This argument is effective only when the number of ancilla qubits is controlled. With $r$ ancillas, the circuit-counting term scales with the total number $n+r$ of qubits. Thus the resulting size lower bound yields a meaningful depth lower bound only when $r$ is not too large relative to $n$. In particular, the argument does not give an ancilla-uniform depth lower bound.

The approach of~\cite{anschuetz2025strongly} also does not extend directly to the quantum spin-glass model considered in this work. The reason is that the worst-case variance of the $\sqrt{n}$-normalized energy of a fixed state is  $\Theta(1/n)$, whereas, as we recall, it is $1/n^{\Theta(p/2)}$ for the SYK model. On the other hand, even for depth-one circuits, the relevant
discretized class of preparable states has size $n^{\Theta(n)}$. Since
$\log n^{\Theta(n)}=\Theta(n\log n)$, a union bound over this class controls
the maximum only at scale
$\sqrt{n^{-1}\log n^{\Theta(n)}}=\Theta(\sqrt{\log n})=\omega(1)$.
For fixed $p$, this is already above the $O_p(1)$ ground-state scale, so the union-bound estimate is vacuous for our purposes.

This failure identifies the main impediment to a direct counting argument: even shallow circuits form a large enough class so that the resulting union-bound estimate already lies above the ground-state scale. The main novelty of our approach is to change the object being bounded. We do not try to control the raw energy of each circuit state directly. Instead, we subtract the part of the energy determined by its one-qubit marginals and study the resulting residual process.

More concretely, consider a depth-$D$ circuit acting on $n+r$ qubits prepared initially at a trivial starting state.
The Hamiltonian is assumed to act only on the first $n$ physical qubits while the remaining $r$ qubits are ancillas. We denote by $\rho$ the reduced state of the circuit output on the physical qubits,
obtained by tracing out the ancillas. Let $m(\rho)$ be the product density operator formed from its one-qubit marginals. We write
\[
    \operatorname{Tr}(H\rho)
    =
    \operatorname{Tr}(Hm(\rho))
    +
    R(\rho),
    \qquad
    R(\rho):=\operatorname{Tr}\bigl(H(\rho-m(\rho))\bigr).
\]
The first term is bounded above by the product-state optimum, independently of the circuit depth. Thus any possible improvement
over product states must come from the residual term $R(\rho)$, which isolates
the non-product $p$-point correlations of the circuit state.

The key point is that this residual has much smaller variance than the raw
energy. This variance reduction comes from a locality principle. The backward
light-cone, or backward support, of an output qubit consists of the input
qubits that can influence it after tracing the circuit backwards. Our crucial observation is that if the backward supports of the physical qubits appearing in a $p$-spin interaction
are disjoint, then the corresponding $p$-point expectation factorizes into
one-point expectations, and the residual coefficient vanishes. Thus only
interactions involving overlapping backward supports can contribute to the
residual.

Quantitatively, while the raw normalized energy of a fixed circuit state has variance of order $1/n$, the normalized residual $R(\rho)/\sqrt n$ has variance of order $2^D/n^2$. The same backward-support structure also controls the entropy of the circuit class, uniformly over the number of ancillas. A direct cover of all depth-$D$ circuits on $n+r$ qubits would depend on $r$. Instead,
for the purpose of computing the physical energy, we replace each circuit by an equivalent pruned circuit obtained by keeping only the gates that can influence the physical output qubits. Gates outside this pruned circuit cannot affect the physical marginal, and hence cannot affect the energy. The remaining active ancillas can then be relabeled into a canonical block, so the resulting covering bound is independent of the ambient number $r$.

After this pruning and relabeling, the effective entropy of depth-$D$ circuits is of order $Dn2^D(\log n+D)$. Thus standard Gaussian supremum bounds suggest the estimate
\[
    \frac{1}{\sqrt n}
    \sup_{\rho \text{ arising from a depth-}D\text{ circuit}}
    |R(\rho)|
    \lesssim
    \sqrt{\frac{4^D D(\log n+D)}{n}}.
\]
This tends to zero for logarithmic depths with a sufficiently small prefactor. Consequently, the residual is negligible, and depth-$D$ circuits cannot asymptotically improve over the product-state energy in this regime.

The diluted models require additional work to make this residual argument compatible with random sparsification. In the growing-average-degree regime, the retained hypergraph is dense enough that the residual variance and entropy bounds can be transferred from the mean-field setting after paying a sampling cost. At bounded average degree, this global comparison is no longer available, and the proof instead localizes the variance and entropy estimates
to the part of the retained hypergraph visible to the backward-support dependency graph. The same pruning and relabeling idea keeps these bounds uniform over the number of ancillas.

\subsection{NLTS connection and open problems}\label{ssec:NLTS-connection}
Our results make a conceptual connection to the NLTS problem, which we now recall. The NLTS property asks for local Hamiltonians with the feature that all sufficiently near-ground states have nontrivial circuit complexity: for some small constant $\epsilon>0$, no state within energy $\epsilon n$ of the optimum can be prepared by a constant-depth circuit. In its standard form, the Hamiltonian is required to be a sum of $O(n)$ local terms on $n$ qubits. 

The existing NLTS landscape is quite different from the spin-glass setting considered here. The known full NLTS construction is code-based: the obstruction to shallow preparation is enforced by the structure of a quantum LDPC code~\cite{anshu2023nlts}. We give a brief self-contained overview of this construction and its contrast with our setting in
Appendix~\ref{ssec:ABN}. Earlier results of Eldar and Harrow~\cite{eldar2017local} and Anshu et al.~\cite{anshu2022construction}
establish the weaker combinatorial form of NLTS, again through code-theoretic
constructions. To our knowledge, the only non-code construction is~\cite{anschuetz2024combinatorial}, which proves only the combinatorial
version. Apart from this exception, the Hamiltonians in these works consist of
commuting local terms and are frustration-free, with ground energy zero. Their ground spaces are deliberately engineered, and exact ground states can be prepared by polynomial-size circuits.  In contrast, the Hamiltonians studied here arise from the quantum $p$-spin glass model. As such they  are random,
consist of  noncommuting local Hamiltonians, and have no known polynomial-size circuit preparation.

In our first two regimes, where the average degree grows with $n$, our results give logarithmic-depth lower bounds for near-ground states. In this sense, they provide a full, rather than merely combinatorial, NLTS-type obstruction for quantum spin-glass Hamiltonians. The caveat is that the corresponding Hamiltonians have a superlinear number of local terms, and
therefore fall outside the standard linear-size formulation of NLTS. 

In the linear-size, bounded-average-degree case, our third regime, we prove a prefactor-dependent fixed-depth obstruction: for every prescribed depth $D$, the probability that a depth-$D$ circuit prepares a near-ground state can be made arbitrarily small by first taking $n$ large and then taking the average-degree prefactor $\zeta$ large. In light of the NLTS problem, we conjecture that a stronger fixed-$\zeta$ result holds, namely that for all sufficiently large constant $\zeta$, the minimum circuit depth is in fact $\Omega(\log n)$. Such a result would yield the full NLTS property
in the linear-size regime, but appears to be out of reach of our current understanding and proof techniques.  More broadly, we conjecture that near-ground states of quantum $p$-spin glasses cannot be prepared by polynomial-size circuits, even with arbitrary ancillas. Proving such a statement, however, would go far beyond the techniques used here: it would amount to a polynomial-size state-preparation lower bound for natural low-energy quantum witnesses. Such a result is closely tied to the difficulty of separating $\mathrm{BQP}$ from $\mathrm{QMA}$; see, for example,~\cite{nielsen2010quantum}.

The proof strategy used here is also quite different from the approaches used in prior NLTS constructions. The known constructions of combinatorial or full NLTS Hamiltonians, including~\cite{eldar2017local,anshu2022construction,anshu2023nlts,anschuetz2024combinatorial}, are based on the isoperimetric framework introduced by Eldar and Harrow~\cite{eldar2017local}. In broad terms, this framework exploits the fact that probability distributions arising from shallow quantum circuits satisfy an isoperimetric property: every set must have boundary flow comparable to its interior mass. The Hamiltonians in these constructions are then designed so that near-ground states induce probability measures which cluster into well-separated regions. Such clustering produces sets with very small boundary flow, contradicting the isoperimetric behavior forced by shallow circuit preparation. This rules out shallow circuits as generators of near-ground states.

This strategy is difficult to apply directly to quantum spin glasses, because we do not currently know whether their near-ground states exhibit an analogous clustering structure. A step in this direction was obtained in~\cite{anschuetz2026quantum}, but such clustering information is not presently available in a form that would yield the lower bounds proved here. Our approach therefore opens a probabilistic route to NLTS-type questions, recasting state-preparation lower bounds for random quantum Hamiltonians in the language of spin glasses and Gaussian processes rather than code structure and isoperimetry.

\subsection{Notation and paper outline}

We finish the introduction with some notation and a brief outline of the paper.
Throughout, \(\R\) and \(\C\) denote the real and complex numbers, respectively. We write
\(\Z\) for the integers and \(\Z_+\) for the nonnegative integers. For a
positive integer \(m\), let \([m]:=\{1,\ldots,m\}\). Let $\B_2^d:=\{x\in\R^d:\|x\|_2\le 1\}$ and $\mcS^{d-1}:=\{x\in\R^d:\|x\|_2=1\}$ denote the Euclidean unit ball and unit sphere in $\mathbb{R}^d$, respectively. If \(B\) is a finite set, \(|B|\) denotes its cardinality. Given a graph \(\mathsf G\), we write \(E(\mathsf G)\) for its edge set and
\(\Delta(\mathsf G)\) for its maximum degree. For \(x>0\), we write \(\log_+ x:=\max\{\log x,0\}\).

We write \(o_n(1)\) for a quantity that vanishes as \(n\to\infty\), with all
remaining fixed parameters held fixed, and \(o_p(1)\) for a quantity that
vanishes as \(p\to\infty\). More generally, for nonnegative quantities \(a_n\)
and \(b_n\), we write \(a_n=O(b_n)\), \(a_n=\Omega(b_n)\), and
\(a_n=\Theta(b_n)\) in the usual sense. Subscripts indicate allowed
dependencies of the implicit constants, so that \(O_p(\cdot)\),
\(\Omega_p(\cdot)\), and \(\Theta_p(\cdot)\) allow constants depending on
\(p\), but not on \(n\). We also write \(a_n=o(b_n)\) if \(a_n/b_n\to0\), and
\(a_n=\omega(b_n)\) if \(a_n/b_n\to\infty\). Constants denoted by \(c,C\) may change from line to line.
Subscripts indicate allowed dependencies, so that, for instance, \(C_p\) may depend on \(p\), but not on \(n\), \(D\), \(r\), or \(\zeta\), unless explicitly stated otherwise.

The remainder of the paper is organized as follows. 
Section~\ref{sec:setup-main-results} defines the quantum \(p\)-spin models, the depth-\(D\) circuit classes with ancillas, the relevant energy benchmarks, and states the main results, Theorems~\ref{thm:main-meanfield-gap} and~\ref{thm:main-critical-gap}. Section~\ref{sec:proofs} contains all technical proofs. The Appendix collects background on classical and quantum \(p\)-spin models, the comparison with code-based NLTS constructions, and auxiliary technical estimates used in the paper.

\section{Setup and main results}
\label{sec:setup-main-results}

We now give the formal setup for the paper. We first define the mean-field quantum \(p\)-spin Hamiltonian and its two diluted variants, corresponding respectively to growing and bounded average degree. We then introduce the class of depth-\(D\) quantum circuits considered in our lower bounds.

After setting up these models, we define the normalized maximum energies
associated with unrestricted states, product states, and depth-\(D\) circuit states. These quantities provide the benchmarks for the main results: the product-state optimum is the energy achievable without entanglement, while the ground-state optimum is the target energy scale. We then state our main
theorems, Theorems~\ref{thm:main-meanfield-gap} and
\ref{thm:main-critical-gap}, which show that shallow circuits cannot generate sufficient entanglement to move beyond the product-state scale. We close the section by formulating an open conjecture: an NLTS-type strengthening in the bounded-average-degree regime.

\subsection{Quantum $p$-spin glasses}
We now define the quantum \(p\)-spin Hamiltonians studied in this paper. For readers less familiar with spin glasses, Appendix~\ref{ssec:classical-quantum-p-spin-background}
recalls the classical \(p\)-spin model and explains how the quantum model below can be viewed as its noncommutative analogue.
Consider an \(n\)-qubit physical system with Hilbert space 
$\mcH_n=(\C^2)^{\otimes n}.$ Let
\[
    \sigma^1
    =
    \begin{pmatrix}
        0&1\\
        1&0
    \end{pmatrix},
    \qquad
    \sigma^2
    =
    \begin{pmatrix}
        0&-i\\
        i&0
    \end{pmatrix},
    \qquad
    \sigma^3
    =
    \begin{pmatrix}
        1&0\\
        0&-1
    \end{pmatrix}
\]
denote the Pauli matrices. Equivalently, \(\sigma^1,\sigma^2,\sigma^3\) are the usual \(X,Y,Z\) Pauli matrices, respectively. For \(i\in[n]\) and \(a\in\{1,2,3\}\), 
\(\sigma_i^a\) denotes the operator acting as \(\sigma^a\) on the \(i\)-th
qubit and as the identity on all other qubits. For fixed \(p\ge2\), let
\[
    \mcI_p^n
    :=
    \{(i_1,\dots,i_p)\in[n]^p:1\le i_1<\cdots<i_p\le n\}
\]
be the set of increasing $p$-tuples in \([n]\). For
\(I=(i_1,\dots,i_p)\in\mcI_p^n\) and
\(a=(a_1,\dots,a_p)\in\{1,2,3\}^p\), define the Pauli string $P_I^a
    :=
    \prod_{s=1}^p \sigma_{i_s}^{a_s}$.
Since the factors act on distinct qubits, the order of the factors in $I$ 
is immaterial. The mean-field quantum $p$-spin Hamiltonian is a 
random self-adjoint operator
\(H_{n,p}:\mcH_n\to\mcH_n\) defined by
\begin{equation}
\label{eq:mean-field-Hamiltonian}
H_{n,p}
=
\binom np^{-1/2}
\sum_{I\in\mcI_p^n}
\sum_{a\in\{1,2,3\}^p}
g_{I,a}P_I^a ,
\end{equation}
where  \(\{g_{I,a}\}\) are independent standard Gaussian
random variables.

Now we introduce a diluted quantum spin glass model which is obtained from
(\ref{eq:mean-field-Hamiltonian}) by
retaining a random subset of the $p$-tuples \(I\in\mcI_p^n\), while keeping all Pauli strings supported on each retained $p$-tuple. Define 
\begin{equation}
\label{eq:sparseHamiltonian}
    H_{n,p}^{(\beta)}
    :=
    m_n^{-1/2}
    \sum_{I\in\mcI_p^n}
    \xi_I
    \sum_{a\in\{1,2,3\}^p}
    g_{I,a}P_I^a .
\end{equation}
Here, the variables \(\{\xi_I:I\in\mcI_p^n\}\) are i.i.d. Bernoulli random
variables with \(\xi_I\sim\operatorname{Bernoulli}(\kappa_n)\), independent of
the Gaussian coefficients. Set
\[
    \kappa_n:=n^{-\beta},
    \qquad
    m_n:=\kappa_n\binom np.
\]
Thus \(\kappa_n\) is the retention probability and $m_n=\Theta_p(n^{p-\beta})$ is the expected number of retained $p$-tuples.
In the growing average degree setting we take
 \(0< \beta<p-1\). Then $m_n=\omega(n)$, and the expected degree of each vertex in the hypergraph diverges.
 
 The third and final regime we consider, namely the 
 bounded-average-degree regime, 
 is obtained by introducing a constant \(\zeta>0\) and setting
\[
    \kappa_{n,\zeta}:=\zeta n^{-(p-1)},
    \qquad
    m_{n,\zeta}:=\kappa_{n,\zeta}\binom np .
\]
Then $m_{n,\zeta} =\frac{\zeta}{p!}\,n(1+o_n(1))$.
Let \(\{\xi_I:I\in\mcI_p^n\}\) be independent Bernoulli random
variables with
\(\xi_I\sim\operatorname{Bernoulli}(\kappa_{n,\zeta})\), again independent of
the Gaussian coefficients. The bounded-average-degree diluted Hamiltonian is then defined analogously by
\begin{equation}
\label{eq:zetaHamiltonian}
    H_{n,p}^{(\zeta)}
    :=
    m_{n,\zeta}^{-1/2}
    \sum_{I\in\mcI_p^n}
    \xi_I
    \sum_{a\in\{1,2,3\}^p}
    g_{I,a}P_I^a .
\end{equation}
Thus in this regime the number of local
Hamiltonians is linear, and the average number of Hamiltonians per qubit is $O(1)$.

\subsection{Quantum circuit classes}
Let
\[
    \mathsf{S}_{n,\mathrm{all}}
    :=
    \{|\psi\rangle\in(\C^2)^{\otimes n}:\|\psi\|=1\}
\]
denote the set of all pure states in $\mcH_n$. Given \(|\psi\rangle \in\mcH_n\) and any Hamiltonian $H$,
the energy of this state with respect to $H$ is defined as 
$X(\psi):= \langle \psi|H|\psi\rangle$.  
The set of all product (unentangled) 
states $\mathsf{S}_{n,\mathrm{prod}}\subset \mathsf{S}_{n,\mathrm{all}}$ is defined as
\[
    \mathsf{S}_{n,\mathrm{prod}}
    :=
    \bigg\{
    |\psi\rangle=\bigotimes_{i=1}^n|\psi_i\rangle:
    |\psi_i\rangle\in\C^2,\ \|\psi_i\|=1
    \bigg\}.
\]
Next we introduce the set of states preparable by  quantum circuits with a fixed depth and  with
ancilla qubits.
Let \(r=r_n\ge 0\) denote the number of ancilla qubits. For notational
simplicity, we assume that \(n+r\) is even, so that each layer may be described
by a perfect matching. This entails no loss of generality, since if \(n+r\) is odd,
one may add a single unused ancilla qubit and pair any otherwise unmatched
qubits using identity gates. Since all of our bounds are uniform over \(r\),
this changes none of the estimates.

We work on the full computational Hilbert space $\mcH_{n,r} := (\C^2)^{\otimes n}\otimes(\C^2)^{\otimes r}$, where the first \(n\) tensor factors correspond to the physical qubits and the
remaining $r$ tensor factors correspond to ancillas. Circuits are initialized
at $|0^n\rangle\otimes |0^r\rangle \in \mcH_{n,r}$. The Hamiltonians
$H_{n,p}, H_{n,p}^{(\beta)}, H_{n,p}^{(\zeta)}$ introduced earlier, are now assumed
to act on the first $n$ qubits. That is, for each of these Hamiltonians $H$
defined on  $n$ qubits, we identify $H$ with $H\otimes I_{\mathrm{anc}}$, 
where \(I_{\mathrm{anc}}\) is the identity operator on the $r$ ancilla qubits.

We now define quantum circuits built from two-qubit unitary gates. This restriction is without loss of generality, since allowing gates acting on any fixed number $c$ of qubits would only change the constants in our final bounds.

Let \(U(4)\) be the unitary group on \(\C^2\otimes \C^2\simeq \C^4\).
Let \(\mathrm{PM}(n+r)\) denote the set of perfect matchings on
\([n+r]\). Given a matching \(M\in\mathrm{PM}(n+r)\), a layer
supported on \(M\) is defined as any  unitary of the form
\[
    V_M
    =
    \bigotimes_{(u,v)\in M} U_{uv},
    \qquad
    U_{uv}\in U(4).
\]
 Thus \(U_{uv}\) is an arbitrary two-qubit gate acting on the qubits indexed by $u$ and $v$, and as the identity on all other qubits. Since the edges of a matching are disjoint, the gates within a layer commute. Given any positive integer $D$ and 
 any sequence of $D$ matchings
 \(M_1,\dots,M_D\in\mathrm{PM}(n+r)\), define
\[
\mathsf{U}^{(D)}_{n,r}(M_1,\dots,M_D)
:=
\bigg\{
U=V_D\cdots V_2V_1:
V_\ell=\bigotimes_{e\in M_\ell}U_e^{(\ell)},
\quad
U_e^{(\ell)}\in U(4)
\bigg\}.
\]
The set of  states preparable by such circuits is denoted by
\[
\mathsf{S}_{n,r}^{(D)}(M_1,\dots,M_D)
:=
\left\{
U(|0^n\rangle\otimes |0^r\rangle):
U\in\mathsf{U}^{(D)}_{n,r}(M_1,\dots,M_D)
\right\}.
\]

The set of all  states preparable by depth-\(D\) circuits is
defined as
\[
\mathsf{S}_{n,r}^{(D)}
:=
\bigcup_{M_1,\dots,M_D\in\mathrm{PM}(n+r)}
\mathsf{S}_{n,r}^{(D)}(M_1,\dots,M_D).
\]
Thus \(\mathsf{S}_{n,r}^{(D)}\) is the class of pure states on the joint
physical-ancilla system obtained from
\(|0^n\rangle\otimes |0^r\rangle\) by applying \(D\) layers of disjoint
two-qubit gates. When evaluating an observable \(A\) acting on the first
\(n\) physical qubits, we identify it with \(A\otimes I_{\mathrm{anc}}\) on
the full \(n+r\) qubit system. In this sense, the joint physical-ancilla output 
\(|\phi\rangle\in\mathsf{S}_{n,r}^{(D)}\)
of a circuit
 is accessed through the expectations $\langle \phi|A\otimes I_{\mathrm{anc}}|\phi\rangle $. We say that
 $|\phi\rangle$
  has state preparation complexity (SPC) at most \(D\) with \(r\) ancillas.

It is convenient to parametrize states by their circuit architecture and gate choices. Define
\[
\mathsf{T}_{n,r}^{(D)}
:=
\left\{
(M_1,\dots,M_D,U):
M_1,\dots,M_D\in\mathrm{PM}(n+r),\
U\in\mathsf{U}^{(D)}_{n,r}(M_1,\dots,M_D)
\right\}.
\]
For \(\theta=(M_1,\dots,M_D,U)\in\mathsf{T}_{n,r}^{(D)}\), we write $\phi_\theta := U(|0^n\rangle\otimes |0^r\rangle)$.

\subsection{Normalized energies}
We introduce the ground state energy as the (normalized) maximum energy achieved by states,
and maximum energies achieved by product states and states with preparation complexity
 $D$:
\[
    E_{n,\mathrm{gs}}(p)
    :=
    \frac1{\sqrt n}
    \sup_{\psi\in\mathsf{S}_{n,\mathrm{all}}}
    \langle\psi|H_{n,p}|\psi\rangle,
    \qquad
    E_{n,\mathrm{prod}}(p)
    :=
    \frac1{\sqrt n}
    \sup_{\psi\in\mathsf{S}_{n,\mathrm{prod}}}
    \langle\psi|H_{n,p}|\psi\rangle.
\]
The normalized energy is of constant order. For depth-\(D\) circuits with $r$ ancillas, let
\[
    E_{n,r}^{(D)}(p)
    :=
    \frac1{\sqrt n}
    \sup_{\phi\in\mathsf{S}_{n,r}^{(D)}}
    \langle\phi|
    H_{n,p}\otimes I_{\mathrm{anc}}
    |\phi\rangle .
\]
We note that \(E_{n,\mathrm{gs}}(p)\) and \(E_{n,\mathrm{prod}}(p)\) are
defined on the original \(n\)-qubit Hilbert space, without ancillas. These
quantities serve as benchmarks for comparison with \(E_{n,r}^{(D)}(p)\):
\(E_{n,\mathrm{gs}}(p)\) is the unrestricted optimum, while
\(E_{n,\mathrm{prod}}(p)\) is the product-state optimum. The relevant bounds
for these benchmark quantities were obtained in~\cite{anschuetz2025bounds}
and~\cite{anschuetz2025strongly}, and are summarized in Theorem~\ref{thm:benchmarks}.

We define $E_{n,\mathrm{gs}}^{(\beta)}(p)$, $E_{n,\mathrm{prod}}^{(\beta)}(p)$, 
$E_{n,r}^{(\beta,D)}(p)$, 
$E_{n,\mathrm{gs}}^{(\zeta)}(p)$, $E_{n,\mathrm{prod}}^{(\zeta)}(p)$, 
$E_{n,r}^{(\zeta,D)}(p)$ similarly when the Hamiltonian $H_{n,p}$
is replaced by $H_{n,p}^{(\beta)}$ and $H_{n,p}^{(\zeta)}$, respectively.

\subsection{Circuit complexity lower bounds}

Before stating our main results, we record the energy estimates that will
serve as benchmarks. The mean-field estimates are due to
\cite[Theorems~2 and 6]{anschuetz2025bounds} and
\cite[Corollary~D.2]{anschuetz2025strongly}. The corresponding estimates
for the two diluted models follow from
Corollary~\ref{cor:sparse-mean-field-transfer-hp}, proved in
Appendix~\ref{app:universality-transfer}.

\begin{theorem}
[\cite{anschuetz2025bounds},\cite{anschuetz2025strongly}]
\label{thm:benchmarks}
There exists a universal constant \(C>0\) such that the following hold.

First, for all  $p \ge 3$, \(0<\beta<p-1\), with high
probability as \(n\to\infty\),
\begin{align*}
\frac{3^{p/2}}{Cp}
\le
E_{n,\mathrm{gs}}(p),\,
E_{n,\mathrm{gs}}^{(\beta)}(p)
&\le
3^{p/2}\sqrt{2\log p}, \\
E_{n,\mathrm{prod}}(p),\,
E_{n,\mathrm{prod}}^{(\beta)}(p)
&\le
(1+o_p(1))\sqrt{2\log p}.
\end{align*}
Second, for every \(p\ge3\), there exists
\(\zeta_0=\zeta_0(p)\) such that, for every fixed
\(\zeta\ge\zeta_0\), with high probability as \(n\to\infty\),
\begin{align*}
\frac{3^{p/2}}{Cp}
\le
E_{n,\mathrm{gs}}^{(\zeta)}(p)
&\le
3^{p/2}\sqrt{2\log p}, \\
E_{n,\mathrm{prod}}^{(\zeta)}(p)
&\le
(1+o_p(1))\sqrt{2\log p}.
\end{align*}
Here  \(o_p(1)\) denotes a function of $p$ which converges to zero as $p\to\infty$.
\end{theorem}

The asymptotic separation between the product-state scale \(\sqrt{\log p}\) and the ground-state scale \(3^{p/2}/p\), as \(p\to\infty\), is the fact that allows us to prove our circuit-depth lower bounds. Our first main result treats the mean-field and growing-average-degree regimes.

\begin{theorem}\label{thm:main-meanfield-gap}
There exists a universal constant \(\delta_0>0\) such that the following holds. For all sufficiently large \(p\), there exists \(c_p>0\) such that, for every sequence \(D_n\le \delta_0\log n\),
\[
\lim_{n\to\infty}
\sup_{r\in\mathbb Z_+}
\P\left(
    E_{n,\mathrm{gs}}(p)-E_{n,r}^{(D_n)}(p)<c_p
\right)
=0.
\]
Moreover, for every \(0<\beta<p-1\), there exists
\(\delta_{p,\beta}>0\) such that, for every sequence
\(D_n\le \delta_{p,\beta}\log n\),
\[
\lim_{n\to\infty}
\sup_{r\in\mathbb Z_+}
\P\left(
    E_{n,\mathrm{gs}}^{(\beta)}(p)
    -
    E_{n,r}^{(\beta,D_n)}(p)
    < c_p
\right)
=0.
\]
\end{theorem}

Our second theorem concerns the bounded-average-degree setting.
\begin{theorem}
\label{thm:main-critical-gap}
For all sufficiently large $p$, there exists \(c_p>0\) such that, for every
 \(D\ge1\),
\[
\lim_{\zeta\to\infty}
\limsup_{n\to\infty}
\sup_{r\in\mathbb Z_+}
\P\left(
    E_{n,\mathrm{gs}}^{(\zeta)}(p)
    -
    E_{n,r}^{(\zeta,D)}(p)
    < c_p
\right)
=0.
\]
\end{theorem}
Theorem~\ref{thm:main-critical-gap} gives a fixed-depth obstruction in the bounded-average-degree regime: for every prescribed depth \(D\), choosing the average-degree parameter \(\zeta\) sufficiently large rules out depth-\(D\) preparation of near-ground states. In light of the NLTS discussion in
Section~\ref{ssec:NLTS-connection}, it is natural to ask whether \(\zeta\) can instead be fixed once and for all, while still obtaining a depth lower bound growing with \(n\). We conjecture that this is the case.
\begin{conjecture}
For all sufficiently large $p$, there exists \(c_p>0\), $\zeta>0$ and $\delta>0$ such that if
 \(D_n\le \delta \log n\), then
\[
\limsup_{n\to\infty}
\sup_{r\in\mathbb Z_+}
\P\left(
    E_{n,\mathrm{gs}}^{(\zeta)}(p)
    -
    E_{n,r}^{(\zeta,D_n)}(p)
    < c_p
\right)
=0.
\]
\end{conjecture}

The proofs of Theorems~\ref{thm:main-meanfield-gap} and
\ref{thm:main-critical-gap} are found in Section~\ref{sec:proofs}.

\section{Proofs}\label{sec:proofs}
In this section we prove Theorems~\ref{thm:main-meanfield-gap} and
\ref{thm:main-critical-gap}. The proofs in the three regimes share the same
basic reduction. We decompose the energy of a circuit state into a product-like
term, determined only by its one-site physical marginals, and a residual term
which captures the entangled \(p\)-point contribution. The
product-like term is controlled by the product-state optimum. Thus the main
technical task is to show that, uniformly over shallow circuits and over the
number of ancillas, the residual term is negligible on the \(\sqrt n\) energy
scale.

We first introduce this decomposition in the mean-field model. For any normalized
state \(|\psi\rangle\in\mcH_{n,r}\) and any physical site \(i\in[n]\), let
\(m_i^\psi\in \B_2^3\) be the Bloch vector of the one-qubit reduced state at
site \(i\), defined by
\[
    m_i^\psi(b)
    :=
    \langle\psi|\sigma_i^b\otimes I_{\mathrm{anc}}|\psi\rangle,
    \qquad b\in\{1,2,3\}.
\]
For \(I=(i_1,\dots,i_p)\in\mcI_p^n\) and
\(a=(a_1,\dots,a_p)\in\{1,2,3\}^p\), define
\[
    \chi_{I,a}(\psi)
    :=
    \langle\psi|P_I^a\otimes I_{\mathrm{anc}}|\psi\rangle
    -
    \prod_{s=1}^p m_{i_s}^{\psi}(a_s).
\]
This coefficient is the part of the physical \(p\)-point Pauli expectation not
explained by the one-site physical marginals. When
\(\psi=\phi_\theta\) for some \(\theta\in\mathsf{T}_{n,r}^{(D)}\), we write
\(\chi_{I,a}(\theta):=\chi_{I,a}(\phi_\theta)\).

For \(\theta\in\mathsf{T}_{n,r}^{(D)}\), set $X(\theta)
    :=
    \langle\phi_\theta|
    H_{n,p}\otimes I_{\mathrm{anc}}
    |\phi_\theta\rangle $. Then $X(\theta)=X_0(\theta)+R(\theta)$, where
\[
    X_0(\theta)
    :=
    \binom np^{-1/2}
    \sum_{I\in\mcI_p^n}
    \sum_{a\in\{1,2,3\}^p}
    g_{I,a}
    \prod_{s=1}^p m_{i_s}^{\phi_\theta}(a_s),
\]
and
\[
    R(\theta)
    :=
    \binom np^{-1/2}
    \sum_{I\in\mcI_p^n}
    \sum_{a\in\{1,2,3\}^p}
    g_{I,a}\chi_{I,a}(\theta).
\]
The term \(X_0\) is bounded above by the product-state optimum, since it is the energy obtained by replacing the physical state by the product state with the same one-site marginals. Therefore any improvement of depth-\(D\) circuit states over product states must come from the residual process \(R\). The same decomposition applies to the diluted Hamiltonians by inserting the
sparsification factors \(\xi_I\) and using the corresponding normalization.
We will repeatedly use the canonical metric associated with a centered
Gaussian process. If \(\{Z_t:t\in T\}\) is a centered Gaussian process, its
canonical metric is $\mathsf{d}_Z(s,t) := \sqrt{\E[(Z_s-Z_t)^2]},$ for $s,t\in T$. In particular, for the mean-field residual process \(R\), we write $\mathsf{d}_R(\theta,\theta')
    :=
    \sqrt{\operatorname{Var}(R(\theta)-R(\theta')) }$.
For diluted models, we use the analogous conditional canonical metric, obtained
by conditioning on the sparsification variables \(\xi\) and taking variance
only over the Gaussian coefficients. For a metric space \((T,\mathsf d)\), we write \(N(T,\mathsf d,\eta)\) for the minimal number of \(\mathsf d\)-balls of radius \(\eta\) needed to cover \(T\). The corresponding metric entropy is $\log N(T,\mathsf d,\eta)$. Thus, for example, $\log N(\mathsf{T}_{n,r}^{(D)},\mathsf d_R,\eta)$ is the metric entropy of the depth-\(D\) circuit parameter space at scale
\(\eta\), measured in the canonical metric of the residual process.

The rest of the section is devoted to controlling the residual process. In Subsection~\ref{ssec:backwardsupports}, we prove the backward-support and dependency-graph estimates that encode the locality of depth-\(D\) circuits. We then use these estimates to control the residual in the mean-field regime in Subsection~\ref{ssec:mean-field-regime}, in the growing-average-degree sparse regime in Subsection~\ref{ssec:growing-average-degree-diluted-regime}, and in the bounded-average-degree diluted regime in Subsection~\ref{ssec:bounded-average-degree-diluted-model}.

\subsection{Backward supports and the dependency graph} \label{ssec:backwardsupports}
The purpose of this section is to make precise the locality constraints imposed
by a depth-$D$ circuit. These constraints are the main input in our analysis
of the residual process. Although a depth-\(D\) circuit may create nontrivial
correlations among the output qubits, such correlations are limited by the
backward geometry of the circuit. An output observable can only depend on the
input qubits reached by tracing its support backward through the \(D\) layers.

This motivates the definition of a backward support set, which records which
input sites can influence a given set of output sites. This is the
circuit-theoretic analogue of a \textit{past light-cone}. We will use
these sets to construct the dependency graph \(\mathsf G_D^{\mathrm{dep}}\).
The key consequence is that if the vertices of a $p$-tuple have pairwise
disjoint backward supports, then the corresponding $p$-point Pauli
expectation factorizes into one-site expectations. Equivalently, \(\chi_{I,a}\) vanishes whenever $I$ is an independent set of the dependency graph.

\begin{definition}[Backward support sets]
\label{def:backward-support}
Fix matchings \(M_1,\dots,M_D \in \mathrm{PM}(n+r)\). For \(S\subset[n+r]\), define $L_0(S):=S$, and recursively, for \(t=0,\dots,D-1\),
\[
L_{t+1}(S)
:=
L_t(S)\cup M_{D-t}(L_t(S)),
\]
where
\[
M_\ell(A):=\{j\in[n+r]: \text{ there exists }i\in A
\text{ such that }(i,j)\in M_\ell\}.
\]
For a singleton \(S=\{i\}\), we write \(L_D(i):=L_D(\{i\})\).
\end{definition}

Equivalently, \(L_D(S)\) is the set obtained by starting from \(S\) and, for
\(\ell=D,D-1,\dots,1\), adjoining every vertex paired by \(M_\ell\) with one of
the vertices already present. 

The set \(L_D(i)\) should be interpreted as the set of input qubits which may
influence the output qubit \(i\). 

\begin{lemma}[Backward support growth]
\label{lem:backward-support-growth}
Let $U \in \mathsf{U}^{(D)}_{n,r}(M_1,\dots, M_D)$. If an observable \(A\) is supported on \(S\subset[n+r]\), then \(U^*AU\) is
supported on \(L_D(S)\). Moreover, $|L_D(S)|\le \min\{n+r,2^D|S|\}$.
\end{lemma}

\begin{proof}
For \(\ell=1,\dots,D\), write
\[
V_\ell := \bigotimes_{(u,v)\in M_\ell} U_{uv}^{(\ell)},
\]
so that $U=V_D\cdots V_1$. Define intermediate observables by
\[
A_0:=A,
\qquad
A_{t+1}:=V_{D-t}^* A_t V_{D-t},
\qquad t=0,\dots,D-1.
\]
Thus \(A_D=U^*AU\). We claim by induction that \(A_t\) is supported on \(L_t(S)\). This is true for
\(t=0\), since \(A_0=A\) is supported on \(S=L_0(S)\).

Assume that \(A_t\) is supported on \(L_t(S)\). We pull \(A_t\) back
through the layer \(V_{D-t}\), namely we form $A_{t+1}=V_{D-t}^*A_tV_{D-t}$. Since \(M_{D-t}\) is a matching, each gate in this layer acts
on a disjoint pair of qubits. A gate acting on a pair disjoint from \(L_t(S)\)
commutes with \(A_t\) and therefore does not change its support. A gate acting
on a pair \((u,v)\) with at least one endpoint in \(L_t(S)\) can only enlarge
the support by adding the other endpoint of that pair. Hence
\[
\operatorname{supp}(A_{t+1})
\subseteq
L_t(S)\cup M_{D-t}(L_t(S))
=
L_{t+1}(S).
\]
This proves the induction, and therefore \(U^*AU=A_D\) is supported on
\(L_D(S)\).

It remains to prove the size bound. Since \(M_\ell\) is a matching, every
element of \(L_t(S)\) has at most one partner under \(M_{D-t}\). Therefore $|M_{D-t}(L_t(S))|\le |L_t(S)|$. Consequently,
\[
|L_{t+1}(S)|
=
|L_t(S)\cup M_{D-t}(L_t(S))|
\le
|L_t(S)|+|M_{D-t}(L_t(S))|
\le
2|L_t(S)|.
\]
Iterating gives $|L_D(S)|\le 2^D |S|$. Since \(L_D(S)\subseteq[n+r]\), we also have \(|L_D(S)|\le n+r\). Combining the two
bounds yields the result.
\end{proof}

\begin{lemma}[Backward support multiplicity]
\label{lem:backward-support-multiplicity}
Fix matchings \(M_1,\dots,M_D \in \mathrm{PM}(n+r)\). For each input site \(v\in[n+r]\),
\[
|\{j\in[n]:v\in L_D(j)\} |
\le 2^D.
\]
\end{lemma}

\begin{proof}
This is the same growth estimate as Lemma~\ref{lem:backward-support-growth},
but applied in the forward direction. Starting from \(v\), define
\[
F_0(v):=\{v\},
\qquad
F_{\ell+1}(v):=F_\ell(v)\cup M_{\ell+1}(F_\ell(v)),
\qquad
\ell=0,\dots,D-1.
\]
Since each \(M_\ell\) is a matching, each step can at most double the size of
the set. Hence $|F_D(v)|\le 2^D$.

Moreover, \(v\in L_D(j)\) if and only if \(j\in F_D(v)\). Therefore
\[
|\{j\in[n]:v\in L_D(j)\} |
\le
|F_D(v)|
\le 2^D.
\]
\end{proof}

\begin{definition}[Dependency graph]
\label{def:dependency-graph}
Given \(M_1,\dots,M_D \in \mathrm{PM}(n+r)\), define the dependency graph $\mathsf{G}^{\mathrm{dep}}_D=\mathsf{G}^{\mathrm{dep}}_D(M_1,\dots,M_D)$ on vertex set \([n]\) by declaring distinct vertices $i,j \in [n]$ adjacent if $L_D(i)\cap L_D(j)\neq\emptyset$.
\end{definition}

In what follows, we denote by $\Delta(\mathsf{G})$ the maximum degree of the graph $\mathsf{G}.$ Specifically, 
\begin{align*}
    \Delta(\mathsf{G}) &= \max_{i \le n} \mathrm{deg}_{\mathsf{G}}(i),
\end{align*}
where $\mathrm{deg}_{\mathsf{G}}(i)$ is the number of vertices in $\mathsf{G}$ that are adjacent to $i$.

\begin{lemma}[Bounded degree of the dependency graph]
\label{lem:dependency-degree}
For every $D \ge 0$, the dependency graph $\mathsf{G}^{\mathrm{dep}}_D$ satisfies $\Delta(\mathsf{G}^{\mathrm{dep}}_D)\le 4^D$. 
\end{lemma}
\begin{proof}
Fix \(i\in[n]\). Every neighbor \(j\) of \(i\) in
\(\mathsf G_D^{\mathrm{dep}}\) satisfies
\(L_D(i)\cap L_D(j)\neq\emptyset\). Hence, by a union bound,
\[
\deg_{\mathsf G_D^{\mathrm{dep}}}(i)
\le
\sum_{v\in L_D(i)}
|\{j\in[n]:v\in L_D(j)\} |.
\]
By Lemmas~\ref{lem:backward-support-growth} and
\ref{lem:backward-support-multiplicity}, the right-hand side is at most
\[
|L_D(i)|2^D\le 2^D\cdot 2^D=4^D.
\]
Taking the maximum over \(i\in[n]\) proves the claim.
\end{proof}

We will use the following simple consequence of Lemma~\ref{lem:backward-support-growth}. 
Let \(|\phi\rangle=U|0^{n+r}\rangle\), where
\(|0^{n+r}\rangle:=|0^n\rangle\otimes |0^r\rangle\), and
\(U\in\mathsf{U}^{(D)}_{n,r}(M_1,\dots,M_D)\). If observables \(A\) and \(B\)
are supported on sets \(S_A,S_B\subset[n+r]\) such that
\(L_D(S_A)\cap L_D(S_B)=\emptyset\), then \(U^*AU\) and \(U^*BU\) are
supported on disjoint sets of input qubits. Since the input state
\(|0^{n+r}\rangle\) is a product state, their expectations factor:
\[
\langle 0^{n+r}|(U^*AU)(U^*BU)|0^{n+r}\rangle
=
\langle 0^{n+r}|U^*AU|0^{n+r}\rangle
\langle 0^{n+r}|U^*BU|0^{n+r}\rangle .
\]
Thus disjoint backward supports force factorization of the corresponding
output observables.

\begin{lemma}[Residual support through the dependency graph]
\label{lem:residual-support-dependency-graph}
Fix matchings \(M_1,\dots,M_D \in \mathrm{PM}(n+r)\), and let $\mathsf G_D^{\mathrm{dep}}$ be the associated dependency graph. Let $\theta \in \mathsf{T}_{n,r}^{(D)}$ and $\phi:= \phi_\theta = U (| 0^n\rangle \otimes | 0^r\rangle).$ If the sites appearing in \(I=(i_1,\dots,i_p)\in\mcI_p^n\) form an independent
set in \(\mathsf G_D^{\mathrm{dep}}\), then
\[
    \chi_{I,a}(\theta)=0
    \qquad
    \text{for every }a\in\{1,2,3\}^p.
\]
\end{lemma}

\begin{proof}
For each \(s=1,\dots,p\), define the pulled-back one-site observable $O_s:=U^*(\sigma_{i_s}^{a_s}\otimes I_{\mathrm{anc}})U.$ By Lemma~\ref{lem:backward-support-growth}, \(O_s\) is supported on
\(L_D(i_s)\). Since the sites in $I$ form an independent set in $\mathsf{G}_D^{\mathrm{dep}}$, we have
\[
    L_D(i_s)\cap L_D(i_t)=\emptyset
    \qquad
    \text{for all }s\neq t .
\]
Since the observables \(O_1,\dots,O_p\) are supported on pairwise disjoint
input tensor factors and \(|0^{n+r}\rangle\) is a product state,
\begin{align*}
\langle\phi|P_I^a\otimes I_{\mathrm{anc}}|\phi\rangle
&=
\langle 0^{n+r}|U^*(P_I^a\otimes I_{\mathrm{anc}})U|0^{n+r}\rangle  
=
\langle 0^{n+r}|O_1\cdots O_p|0^{n+r}\rangle  \\
&=
\prod_{s=1}^p
\langle 0^{n+r}|O_s|0^{n+r}\rangle  
=
\prod_{s=1}^p
\langle\phi|\sigma_{i_s}^{a_s}\otimes I_{\mathrm{anc}}|\phi\rangle
=
\prod_{s=1}^p m_{i_s}^\phi(a_s).
\end{align*}
Hence \(\chi_{I,a}(\theta)=0\).
\end{proof}

\subsection{Mean-field regime} \label{ssec:mean-field-regime}

In this section we prove the mean-field ground-state depth separation, giving the first part of
Theorem~\ref{thm:main-meanfield-gap}. The main step is to show that the residual
process $R$ is uniformly small over the depth-\(D\) circuit class with an
arbitrary number of ancillas. We prove this using a standard Gaussian-process
strategy. We first control the variance and metric entropy of $R$, then combine these estimates using Dudley's entropy integral and the Borell-TIS inequality.

The variance bound, Proposition~\ref{prop:full-residual-variance}, uses the locality of pulled-back physical observables.
Although the circuit acts on \(n+r\) qubits, the Hamiltonian contains only
physical Pauli strings. For each physical site \(i\), the pulled-back one-site observable is supported on the backward support \(L_D(i)\), whose size is at most \(2^D\). Together with the backward-support multiplicity estimate, this gives the uniform variance bound
\[
    \sup_{\theta\in\mathsf{T}_{n,r}^{(D)}}
    \operatorname{Var}(R(\theta))
    \le
    C_p\frac{2^D}{n},
\]
with a constant independent of $r$.

The entropy bound, Proposition~\ref{prop:full-residual-entropy}, is the only place where the arbitrary number of ancillas
requires additional work relative to the no-ancilla setting. A direct net over all \(n+r\) qubits would depend on $r$. Instead, we prune the circuit to the joint backward support of the
physical outputs. Since each layer is a matching, this active set has size at
most \(n2^D\). Gates outside this active set are invisible to all physical
observables, and the remaining active ancillas can be relabeled into a
canonical block. This gives an $r$-uniform entropy bound for the residual
metric, at the price of an additional factor \(2^D\). 

Combining the variance and entropy estimates gives the expected supremum bound in Lemma~\ref{lem:expected-sup-residual}. We then apply Borell-TIS in
Proposition~\ref{prop:circuit-residual-high-prob} to obtain the high-probability
residual bound needed for the mean-field depth-to-product comparison. The final
ground-state gap follows by comparing the product-state and ground-state
energies.

\begin{proposition}[Variance bound, mean-field regime]
\label{prop:full-residual-variance}
Fix \(p\ge2\). For every \(r=r_n\ge0\),
\[
    \sup_{\theta\in\mathsf{T}_{n,r}^{(D)}}
    \operatorname{Var}\bigl(R(\theta)\bigr)
    \le
    C_p\frac{2^D}{n}.
\]
The constant \(C_p\) is independent of $r$.
\end{proposition}

\begin{proof}
Fix \(\theta=(M_1,\dots,M_D,U)\in\mathsf{T}_{n,r}^{(D)}\), and write $\phi=\phi_\theta=U|0^{n+r}\rangle$. Since the \(g_{I,a}\)'s are independent standard Gaussians,
\[
\operatorname{Var}(R(\theta))
=
\binom np^{-1}
\sum_{I\in\mcI_p^n}
\sum_{a\in\{1,2,3\}^p}
|\chi_{I,a}(\theta)|^2 .
\]
Thus it suffices to prove the deterministic bound
\[
\sum_{I\in\mcI_p^n}
\sum_{a\in\{1,2,3\}^p}
|\chi_{I,a}(\theta)|^2
\le
C_p n^{p-1}2^D.
\]

For each physical site \(i\in[n]\) and label \(b\in\{1,2,3\}\), define the
pulled-back Pauli observable $O_i^b
    :=
    U^*(\sigma_i^b\otimes I_{\mathrm{anc}})U$. Also define its mean, centered version, and fluctuation vector respectively by
\[
    m_i^b
    :=
    \langle0^{n+r}|O_i^b|0^{n+r}\rangle,
    \qquad
    A_i^b
    :=
    O_i^b-m_i^b \mathrm{Id}_{n+r},
    \qquad
    v_i^b
    :=
    A_i^b|0^{n+r}\rangle,
\]
where $\mathrm{Id}_{n+r}$ denotes the identity operator on $\mcH_{n,r}$.
We refer to \(v_i^b\) as the fluctuation vector because it is the component of
\(O_i^b|0^{n+r}\rangle\) left after subtracting its mean component
\(m_i^b|0^{n+r}\rangle\). 

By Lemma~\ref{lem:backward-support-growth}, \(O_i^b\), and hence \(A_i^b\), is supported on \(L_D(i)\subset[n+r]\). Therefore, we have
\[
    v_i^b
    \in
    (\mathbb C^2)^{\otimes L_D(i)}
    \otimes
    |0\rangle^{\otimes([n+r]\setminus L_D(i))}.
\]
Equivalently, if \(e_T\) denotes the computational basis vector with ones on
\(T\subset[n+r]\) and zeros elsewhere, then \(v_i^b\) is a linear combination
only of basis vectors \(e_T\) with \(T\subseteq L_D(i)\).

Note also that,
\[
\langle 0^{n+r},v_i^b\rangle
=
\langle0^{n+r}|A_i^b|0^{n+r}\rangle
= 0
.
\]
Moreover, since \(A_i^b\) is self-adjoint,
\[
\|v_i^b\|^2
=
\langle0^{n+r}|(A_i^b)^2|0^{n+r}\rangle
=
\langle0^{n+r}|(O_i^b-m_i^b \mathrm{Id}_{n+r})^2|0^{n+r}\rangle .
\]
Using \(\phi=U |0^{n+r} \rangle\), this equals
\[
\langle
\phi
|
(
\sigma_i^b\otimes I_{\mathrm{anc}}
-
\langle\phi|\sigma_i^b\otimes I_{\mathrm{anc}}|\phi\rangle \mathrm{Id}_{n+r}
)^2
|
\phi
\rangle .
\]
Since \((\sigma_i^b\otimes I_{\mathrm{anc}})^2=\mathrm{Id}_{n+r}\), we obtain
\[
\|v_i^b\|^2
=
1-
\langle\phi|\sigma_i^b\otimes I_{\mathrm{anc}}|\phi\rangle^2
\le 1.
\]

By Lemma~\ref{lem:backward-support-multiplicity}, the physical backward
supports satisfy
\[
\sup_{v\in[n+r]}
|\{i\in[n]:v\in L_D(i)\} |
\le
2^D .
\]
We now apply Lemma~\ref{lem:bessel-type-ineq} in Appendix~\ref{ssec:AuxTechnical} with 
\[
    \Omega=[n+r],
    \qquad
    J=[n],
    \qquad
    J_0=\{1,2,3\},
    \qquad
    L_i=L_D(i),
    \qquad
    w_i^b=v_i^b,
    \qquad
    \gamma=2^D .
\]
It follows that for every
\(x\in(\mathbb C^2)^{\otimes(n+r)}\),
\begin{align}\label{eq:fluctuation-vectors-bound}
\sum_{i=1}^n\sum_{b=1}^3
|\langle x,v_i^b\rangle|^2
\le
3\cdot 2^D\|x\|^2.
\end{align}

Now fix \(I=(i_1,\dots,i_p)\in\mcI_p^n\) and
\(a=(a_1,\dots,a_p)\in\{1,2,3\}^p\). Write
\[
    O_s:=O_{i_s}^{a_s},
    \qquad
    m_s:=m_{i_s}^{a_s},
    \qquad
    A_s:=O_s-m_s \mathrm{Id}_{n+r}.
\]
By the definition of \(P_I^a\) and \(O_s\), we have
\[
P_I^a\otimes I_{\mathrm{anc}}
=
\prod_{s=1}^p(\sigma_{i_s}^{a_s}\otimes I_{\mathrm{anc}}),
\]
and hence
\[
U^*(P_I^a\otimes I_{\mathrm{anc}})U
=
\prod_{s=1}^p
U^*(\sigma_{i_s}^{a_s}\otimes I_{\mathrm{anc}})U
=
O_1\cdots O_p .
\]
Therefore
\[
\chi_{I,a}(\theta)
=
\langle0^{n+r}|O_1\cdots O_p|0^{n+r}\rangle
-
\prod_{s=1}^p m_s .
\]

We now bound \(\chi_{I,a}(\theta)\). Since \(A_s=O_s-m_s \mathrm{Id}_{n+r}\), we have the
telescoping identity
\[
O_1\cdots O_p-\prod_{s=1}^p m_s \mathrm{Id}_{n+r}
=
\sum_{s=1}^p
\big(
\prod_{r<s}O_r
\big)
A_s
\big(
\prod_{r>s}m_r
\big),
\]
and so
\[
\chi_{I,a}(\theta)
=
\sum_{s=1}^p
\big(
\prod_{r>s}m_r
\big)
\big\langle
0^{n+r}
\big |
\big (\prod_{r<s}O_r \big)
A_s
\big|
0^{n+r}
\big\rangle .
\]
For each \(s\), define
\[
    x_s
    :=
    \big(\prod_{r<s}O_r\big)^*|0^{n+r}\rangle .
\]
Each \(O_r\) is unitary, so \(\|x_s\|=1\). Moreover, since
\(v_{i_s}^{a_s}=A_s|0^{n+r}\rangle\),
\[
\big\langle
0^{n+r}
\big|
\big(\prod_{r<s}O_r\big)
A_s
\big|
0^{n+r}
\big\rangle
=
\langle x_s,v_{i_s}^{a_s}\rangle .
\]
Also \(|m_r|\le1\). Hence
\[
|\chi_{I,a}(\theta)|
\le
\sum_{s=1}^p
|\langle x_s,v_{i_s}^{a_s}\rangle|.
\]

By Cauchy--Schwarz,
\[
|\chi_{I,a}(\theta)|^2
\le
p
\sum_{s=1}^p
|\langle x_s,v_{i_s}^{a_s}\rangle|^2 .
\]

We now sum this estimate over \(I\in\mcI_p^n\) and
\(a\in\{1,2,3\}^p\). It is convenient to enlarge the sum from increasing
$p$-tuples to ordered $p$-tuples of distinct physical sites,
\[
    [n]^p_{\neq}
    :=
    \{(i_1,\dots,i_p)\in[n]^p:
    i_1,\dots,i_p \text{ are distinct}\}.
\]
This changes the estimate only by a constant depending on $p$. Hence
\[
\begin{aligned}
\sum_{I\in\mcI_p^n}
\sum_{a\in\{1,2,3\}^p}
|\chi_{I,a}(\theta)|^2
&\le
C_p
\sum_{(i_1,\dots,i_p)\in[n]^p_{\neq}}
\sum_{a\in\{1,2,3\}^p}
\sum_{s=1}^p
|\langle x_s,v_{i_s}^{a_s}\rangle|^2  \\
&=
C_p
\sum_{s=1}^p
\sum_{(i_1,\dots,i_p)\in[n]^p_{\neq}}
\sum_{a\in\{1,2,3\}^p}
|\langle x_s,v_{i_s}^{a_s}\rangle|^2 .
\end{aligned}
\]
We bound the inner sum for a fixed \(s\). Write
    $i_{-s}:=(i_r)_{r\neq s}$ and $
    a_{-s}:=(a_r)_{r\neq s}$. Then
\[
\sum_{(i_1,\dots,i_p)\in[n]^p_{\neq}}
\sum_{a\in\{1,2,3\}^p}
|\langle x_s,v_{i_s}^{a_s}\rangle|^2  =
\sum_{i_{-s}\in[n]^{p-1}_{\neq}}
\sum_{a_{-s}\in\{1,2,3\}^{p-1}}
\sum_{i_s\in[n]\setminus\{i_r:r\neq s\}}
\sum_{a_s=1}^3
|\langle x_s,v_{i_s}^{a_s}\rangle|^2 .
\]
For each fixed choice of \(i_{-s},a_{-s}\), the vector $x_s$ is fixed and has norm one. Dropping the restriction
\(i_s\notin\{i_r:r\neq s\}\) can only increase the inner sum, so
\[
\sum_{i_s\in[n]\setminus\{i_r:r\neq s\}}
\sum_{a_s=1}^3
|\langle x_s,v_{i_s}^{a_s}\rangle|^2
\le
\sum_{i_s=1}^n
\sum_{a_s=1}^3
|\langle x_s,v_{i_s}^{a_s}\rangle|^2 .
\]
Applying \eqref{eq:fluctuation-vectors-bound} with \(x=x_s\), we get
\[
\sum_{i_s=1}^n
\sum_{a_s=1}^3
|\langle x_s,v_{i_s}^{a_s}\rangle|^2
\le
3\cdot 2^D\|x_s\|^2
=
3\cdot 2^D .
\]
There are at most \(n^{p-1}\) choices for \(i_{-s}\) and \(3^{p-1}\) choices
for \(a_{-s}\). Therefore, for each fixed \(s\in[p]\),
\[
\sum_{(i_1,\dots,i_p)\in[n]^p_{\neq}}
\sum_{a\in\{1,2,3\}^p}
|\langle x_s,v_{i_s}^{a_s}\rangle|^2
\le
C_p n^{p-1}2^D .
\]
Substituting this bound and summing over \(s=1,\dots,p\), we obtain
\[
\sum_{I\in\mcI_p^n}
\sum_{a\in\{1,2,3\}^p}
|\chi_{I,a}(\theta)|^2
\le
C_p n^{p-1}2^D .
\]
This is the desired deterministic bound. Taking the supremum over
\(\theta\in\mathsf{T}_{n,r}^{(D)}\) completes the proof.

\end{proof}

\begin{proposition}[Entropy bound, mean-field regime]
\label{prop:full-residual-entropy}
Fix \(p\ge2\). Let \(r=r_n\ge0\). Let \(\mathsf{d}_{R}\) be the canonical metric of the residual process
$R$ on \(\mathsf{T}_{n,r}^{(D)}\). Then, for every \(\eta>0\),
\[
\log N(\mathsf{T}_{n,r}^{(D)},\mathsf{d}_{R},\eta)
\le
C_pDn2^D(\log n+D)
+
C_pDn2^D
\log_+\left(\frac{C_pDn2^D}{\eta}\right),
\]
uniformly over \(r\ge0\).
\end{proposition}

\begin{proof}
We identify a parameter
\(\theta\in\mathsf{T}_{n,r}^{(D)}\) with its corresponding circuit unitary \(U\) when no
confusion can arise. Thus $U=V_DV_{D-1}\cdots V_1$, where each layer has the form
\[
    V_t
    =
    \bigotimes_{e\in M_t}U_e^{(t)},
    \qquad
    U_e^{(t)}\in U(4),
\]
for a perfect matching \(M_t\) on \([n+r]\). Throughout, whenever we
write a tensor product over a partial matching, the resulting operator is
understood to act as the identity on all qubits not covered by that matching.

We prove the entropy bound by replacing each circuit, in the metric
\(\mathsf{d}_{R}\), by an equivalent pruned and relabeled circuit acting on at most \(n2^D\) qubits.

We now construct the pruned circuit. It is enough to track how physical
observables pull back through the circuit. Here and below, a physical observable means an observable supported on the first \(n\) qubits.

Write $U=V_DV_{D-1}\cdots V_1$, where \(V_1\) acts first and \(V_D\) acts last. Given a physical observable \(X\), define the pulled-back observables
\[
    X_D:=X,
    \qquad
    X_{t-1}:=V_t^*X_tV_t,
    \qquad t=D,D-1,\dots,1.
\]
Thus \(X_t\) is the observable obtained after pulling \(X\) backward through
the later layers \(D,D-1,\dots,t+1\). In particular, $X_0=U^*XU.$
Since \(X_D=X\) is physical, \(\operatorname{supp}(X_D)\subseteq[n]\). We will
use a joint support bound that works for all physical observables at once, and therefore initialize $S_D:=[n].$

Now define the backward active sets recursively. For \(t=D,D-1,\dots,1\), set
\[
    R_t:=\{e\in M_t:e\cap S_t\neq\emptyset\},
    \qquad
    S_{t-1}:=S_t\cup\bigcup_{e\in R_t}e.
\]
Here \(S_t\) is a set which contains the support of $X_t$, uniformly over all physical observables \(X\) considered in the residual process. The
set \(R_t\) consists of precisely those gates in layer \(t\) that touch this
current backward support. These are the only gates in layer \(t\) that can
affect the pulled-back observable when passing from time \(t\) to time
\(t-1\).

Since \(M_t\) is a matching, each qubit in \(S_t\) belongs to at most one edge
of \(M_t\). Therefore each active qubit can introduce at most one new qubit
when we pull backward through layer \(t\). Hence $|S_{t-1}|\le 2|S_t|,$ and iterating gives
\[
    |S_t|\le n2^{D-t},
    \qquad t=0,1,\dots,D.
\]
In particular, \(|S_0|\le n2^D\). Moreover, since every edge of \(R_t\)
touches \(S_t\), we have \(|R_t|\le |S_t|\), and hence
\[
    \sum_{t=1}^D |R_t|
    \le
    \sum_{t=1}^D |S_t|
    \le
    n(2^D-1).
\]

Define the pruned layer and pruned circuit by
\[
    V_t^{\mathrm{pr}}
    :=
    \bigotimes_{e\in R_t}U_e^{(t)},
    \qquad 
    U^{\mathrm{pr}}
    :=
    V_D^{\mathrm{pr}}\cdots V_1^{\mathrm{pr}},
\]
where \(V_t^{\mathrm{pr}}\) acts as the identity on all qubits not covered by
\(R_t\). We claim that, for every physical observable \(X\), $U^*XU=(U^{\mathrm{pr}})^*XU^{\mathrm{pr}}$.

To prove the claim, use the pulled-back observables \(X_t\) defined above. We prove by induction that \(\operatorname{supp}(X_t)\subseteq S_t\) for every
\(t=D,D-1,\dots,0\), and that
\[
    X_{t-1}=(V_t^{\mathrm{pr}})^*X_tV_t^{\mathrm{pr}}
    \qquad\text{for }t=D,D-1,\dots,1.
\]
The base case holds because \(X_D=X\) is physical, so
\(\operatorname{supp}(X_D)\subseteq[n]=S_D\).

Assume that \(\operatorname{supp}(X_t)\subseteq S_t\). Let
\[
    V_t^{\mathrm{bad}}
    :=
    \bigotimes_{e\in M_t\setminus R_t}U_e^{(t)}.
\]
Since \(M_t\) is a matching, the gates in \(R_t\) and \(M_t\setminus R_t\)
act on disjoint pairs, and hence $V_t=V_t^{\mathrm{bad}}V_t^{\mathrm{pr}}$. By definition of \(R_t\), every edge in \(M_t\setminus R_t\) is disjoint from
\(S_t\). Hence \(V_t^{\mathrm{bad}}\) is supported outside
\(\operatorname{supp}(X_t)\), and therefore $(V_t^{\mathrm{bad}})^*X_tV_t^{\mathrm{bad}}=X_t$. It follows that
\[
    X_{t-1}
    =
    V_t^*X_tV_t
    =
    (V_t^{\mathrm{pr}})^*
    (V_t^{\mathrm{bad}})^*
    X_t
    V_t^{\mathrm{bad}}
    V_t^{\mathrm{pr}}
    =
    (V_t^{\mathrm{pr}})^*X_tV_t^{\mathrm{pr}} .
\]
Finally, conjugation by \(V_t^{\mathrm{pr}}\) can only spread the support along
the edges in \(R_t\). Thus
\[
    \operatorname{supp}(X_{t-1})
    \subseteq
    S_t\cup\bigcup_{e\in R_t}e
    =
    S_{t-1}.
\]
This proves the induction. Since \(X_0=U^*XU\), the displayed identity for
\(X_{t-1}\) also gives $U^*XU=(U^{\mathrm{pr}})^*XU^{\mathrm{pr}}$, as claimed.

Let \(\theta^{\mathrm{pr}}\) denote the circuit parameter obtained from
\(\theta\) by replacing \(U\) with \(U^{\mathrm{pr}}\), equivalently by setting
all deleted gates to the identity. Applying this identity to the physical
Pauli products \(P_I^a\otimes I_{\mathrm{anc}}\) and to the one-site physical
observables \(\sigma_i^b\otimes I_{\mathrm{anc}}\), we obtain $\chi_{I,a}(\theta)=\chi_{I,a}(\theta^{\mathrm{pr}})$ for every \(I\in\mcI_p^n\) and \(a\in\{1,2,3\}^p\). Consequently, $\mathsf d_R(\theta,\theta^{\mathrm{pr}})=0$.

After pruning, each circuit parameter \(\theta\) is \(\mathsf d_R\)-equivalent
to the pruned parameter \(\theta^{\mathrm{pr}}\), whose nonidentity gates
involve only the qubits in the backward active set \(S_0\). In particular, at
most \(n(2^D-1)\) ancillas can remain active. However, these active ancillas
may have arbitrary labels among the \(r\) available ancillas. If we counted
these labels directly, we would reintroduce an \(r\)-dependent combinatorial
factor. Since all ancillas are initialized in the same state and the observables
in the residual process act only on the physical qubits, we may permute ancilla
labels, while keeping the physical labels fixed, without changing any physical
expectation.

Concretely, let $A(\theta^{\mathrm{pr}}):=S_0\setminus[n]$ be the set of ancilla labels appearing in the joint backward support of the
physical outputs. Since \(|S_0|\le n2^D\), we have $|A(\theta^{\mathrm{pr}})|\le n(2^D-1)$.

Set
\[
    q_*:=\min\{r,n(2^D-1)\},
    \qquad
    Q:=n+q_*.
\]
Then \(Q\le n2^D\). Let
\[
    K:=\{1,\dots,n\}\cup\{n+1,\dots,n+q_*\}
\]
be the set consisting of the physical qubits and the first \(q_*\) ancillas.

Since \(|A(\theta^{\mathrm{pr}})|\le q_*\), there exists a permutation \(\pi\)
of \([n+r]\) such that \(\pi(i)=i\) for every physical label \(i\in[n]\), and
\[
    \pi(A(\theta^{\mathrm{pr}}))\subseteq\{n+1,\dots,n+q_*\}.
\]
Let \(P_\pi\) denote the corresponding qubit-permutation unitary, and define $\widetilde U:=P_\pi U^{\mathrm{pr}}P_\pi^*$.

Let \(\widetilde\theta\) denote the corresponding relabeled circuit parameter.
Then \(\widetilde\theta\) is a depth-\(D\) matching-circuit parameter whose
nonidentity gates are all supported on \(K\).

We now check that this relabeling preserves all residual coefficients. Since
\(\pi\) fixes the physical labels, for every physical observable \(X\), $P_\pi^*XP_\pi=X$. Also \(P_\pi|0^{n+r}\rangle=|0^{n+r}\rangle\). Therefore
\[
\langle0^{n+r}|\widetilde U^*X\widetilde U|0^{n+r}\rangle
=
\langle0^{n+r}|
U^{\mathrm{pr}*}XU^{\mathrm{pr}}
|0^{n+r}\rangle.
\]
Applying this identity to all physical Pauli products and one-site physical
Pauli observables gives
\[
    \chi_{I,a}(\widetilde\theta)
    =
    \chi_{I,a}(\theta^{\mathrm{pr}})
    =
    \chi_{I,a}(\theta)
\]
for every \(I\in\mcI_p^n\) and \(a\in\{1,2,3\}^p\). Thus $\mathsf d_R(\theta,\widetilde\theta)=0$.

Consequently, every circuit parameter in \(\mathsf{T}_{n,r}^{(D)}\) is
\(\mathsf{d}_{R}\)-equivalent to a circuit parameter whose nonidentity gates act only
on a canonical set of size $Q\le n2^D$.

It remains to cover depth-\(D\) matching circuits on the set \(K\). Let \(N_Q\) be the number of partial matchings on a \(Q\)-element set. A partial matching with exactly \(m\) edges is
obtained by choosing \(2m\) vertices and partitioning them into \(m\) unordered
pairs. Hence
\[
    N_Q
    =
    \sum_{m=0}^{\lfloor Q/2\rfloor}
    \binom{Q}{2m} \frac{(2m)!}{2^m m!}
    =
    \sum_{m=0}^{\lfloor Q/2\rfloor}
    \frac{Q!}{(Q-2m)!2^m m!}.
\]
For each \(m\),
\[
    \frac{Q!}{(Q-2m)!2^m m!}
    \le Q^{2m}\le Q^Q.
\]
Thus
\[
    N_Q
    \le
    (Q+1)Q^Q
    \le
    \exp (CQ\log(eQ)).
\]
A depth-\(D\) architecture is a sequence of \(D\) partial matchings on \(K\). Therefore the number of possible architectures is at most \(N_Q^D\). Hence
\[
    \log(N_Q^D)
    =
    D\log N_Q
    \le
    C DQ\log(eQ).
\]
Since \(Q\le n2^D\), we have
\[
    DQ\log(eQ)
    \le
    Dn2^D\log(en2^D)
    \le
    C Dn2^D(\log n+D).
\]
Thus the logarithm of the number of possible depth-\(D\) architectures is at most $C Dn2^D(\log n+D)$.

We now cover the remaining unitary degrees of freedom on the active set \(K\). The counting is done in two stages. We first fix the pruned, relabeled architecture, namely the sequence of partial matchings \(\mathcal M=(M_1,\dots,M_D)\) on \(K\). Then we discretize the two-qubit unitaries assigned to the edges of these matchings. The nets obtained for the different architectures will later be combined by taking a union over the possible choices of \(\mathcal M\).

Fix a depth-\(D\) architecture $\mathcal M=(M_1,\dots,M_D)$, where each \(M_t\) is a partial matching on the active set $K.$ Let $J(\mathcal M):=\sum_{t=1}^D |M_t|$ be the number of two-qubit gate locations. Since each \(M_t\) is a partial
matching, 
\[J(\mathcal M)\le \frac{DQ}{2}\le Dn2^D.\]

Set $J_0:=Dn2^D.$ The unitary group \(U(4)\), equipped with the operator norm, has covering numbers
\[
    N(U(4),\|\cdot\|_{\mathrm{op}},\delta)
    \le
    \left(\frac{C}{\delta}\right)^C,
    \qquad 0<\delta\le1.
\]
Taking the product over all gate locations in the fixed architecture gives a
gate-net of cardinality at most
\[
    \left(\frac{C}{\delta}\right)^{C J(\mathcal M)}
    \le
    \left(\frac{C}{\delta}\right)^{C J_0}.
\]

We now relate the gate accuracy \(\delta\) to the residual metric. Let \(\theta_U\) and \(\theta_V\) be two circuit parameters with the same
architecture. Enumerate their corresponding gate locations in circuit order and write
\[
    U=A_J A_{J-1}\cdots A_1,
    \qquad
    V=B_JB_{J-1}\cdots B_1,
\]
where \(J=J(\mathcal M)\), and suppose that $\|A_j-B_j\|_{\mathrm{op}}\le\delta$ for $j=1,\dots,J$. The telescoping expansion gives
\[
    U-V
    =
    \sum_{j=1}^J
    A_J\cdots A_{j+1}
    (A_j-B_j)
    B_{j-1}\cdots B_1,
\]
and therefore $\|U-V\|_{\mathrm{op}} \le J\delta \le J_0\delta$. For any physical Pauli observable \(X=\sigma_i^b\otimes I_{\mathrm{anc}}\), we have
\[
    \|U^*XU-V^*XV\|_{\mathrm{op}}
    \le
    2\|U-V\|_{\mathrm{op}}
    \le
    2J_0\delta.
\]
Since all pulled-back Pauli observables are unitary, for
\(I=\{i_1,\dots,i_p\}\) and \(a=(a_1,\dots,a_p)\),
\[
\big\|
\prod_{s=1}^p O_{i_s}^{a_s}(U)
-
\prod_{s=1}^p O_{i_s}^{a_s}(V)
\big\|_{\mathrm{op}}
\le
\sum_{s=1}^p
\|O_{i_s}^{a_s}(U)-O_{i_s}^{a_s}(V)\|_{\mathrm{op}}
\le
2pJ_0\delta.
\]
Similarly, for every physical site \(i\in[n]\) and every \(b\in\{1,2,3\}\),
\[
\left|
\langle 0^{n+r}|O_i^b(U)|0^{n+r}\rangle
-
\langle 0^{n+r}|O_i^b(V)|0^{n+r}\rangle
\right|
\le
\|O_i^b(U)-O_i^b(V)\|_{\mathrm{op}}
\le
2J_0\delta .
\]
Since these expectations have absolute value at most \(1\), it follows that
\[
\big|
\prod_{s=1}^p
\langle 0^{n+r}|O_{i_s}^{a_s}(U)|0^{n+r}\rangle
-
\prod_{s=1}^p
\langle 0^{n+r}|O_{i_s}^{a_s}(V)|0^{n+r}\rangle
\big|
\le
2pJ_0\delta .
\]
Hence
\[
    |\chi_{I,a}(\theta_U)-\chi_{I,a}(\theta_V)|
    \le C_pJ_0\delta,
\]
and so
\[
\mathsf d_R(\theta_U,\theta_V)^2
=
\binom np^{-1}
\sum_{I\in\mcI_p^n}
\sum_{a\in\{1,2,3\}^p}
|\chi_{I,a}(\theta_U)-\chi_{I,a}(\theta_V)|^2
\le C_pJ_0^2\delta^2.
\]
Thus $\mathsf{d}_{R}(\theta_U,\theta_V)\le C_pJ_0\delta$. If \(0<\eta\le C_p\), choose $\delta=\frac{\eta}{C_pJ_0}.$ Then the product gate-net is an \(\eta\)-net in \(\mathsf{d}_{R}\) for a
fixed architecture, and its logarithmic size is at most $C_pJ_0\log (C_pJ_0/\eta)$.

If \(\eta>C_p\), then one representative per fixed architecture is enough
after increasing \(C_p\), because \(|\chi_{I,a}|\le2\) implies that the
\(\mathsf{d}_{R}\)-diameter is bounded by a constant depending only on $p$.
Thus, for all \(\eta>0\), the gate contribution for a fixed architecture is
bounded by
\[
    C_pJ_0\log_+\left(\frac{C_pJ_0}{\eta}\right)
    =
    C_pDn2^D
    \log_+\left(\frac{C_pDn2^D}{\eta}\right).
\]

It remains to combine the architecture count with the fixed-architecture gate nets. By the pruning and relabeling construction above, every circuit parameter
\(\theta\in\mathsf{T}_{n,r}^{(D)}\) is \(\mathsf{d}_{R}\)-equivalent to a
pruned, relabeled circuit parameter whose nonidentity gates act only on the
canonical active set \(K\), with \(|K|\le n2^D\). Thus it suffices to cover the collection of such pruned, relabeled circuits.

We showed that the number of possible pruned architectures on \(K\) has logarithm at most $C Dn2^D(\log n+D)$. For each fixed pruned architecture, the preceding gate-discretization argument
gives an \(\eta\)-net in \(\mathsf{d}_{R}\) of logarithmic size at most $C_pDn2^D
    \log_+ (C_pDn2^D/\eta)$. Taking the union of these fixed-architecture nets over all pruned
architectures gives a net for the full parameter space \(\mathsf{T}_{n,r}^{(D)}\).
Equivalently, the covering number is bounded by the number of pruned
architectures times the largest fixed-architecture covering number. Therefore,
after absorbing constants into \(C_p\),
\[
\log N(\mathsf{T}_{n,r}^{(D)},\mathsf{d}_{R},\eta)
\le
C_pDn2^D(\log n+D)
+
C_pDn2^D
\log_+\left(\frac{C_pDn2^D}{\eta}\right).
\]
The right-hand side is independent of $r$, as claimed.

\end{proof}

\begin{lemma}[Expected supremum, mean-field regime]
\label{lem:expected-sup-residual}
Fix \(p\ge 2\) and \(D\ge 1\). Then
\[
\sup_{r\in\mathbb Z_+}
\E
\sup_{\theta \in \mathsf{T}_{n,r}^{(D)}}
|R(\theta)|
\le
C_p
2^{D}
\sqrt{D(\log n+D)}.
\]
\end{lemma}

\begin{proof}
By Proposition~\ref{prop:full-residual-variance}, we have
\begin{align}\label{eq:varBoundExpectedSupResidual}
\sigma_\infty^2
:=
\sup_{\theta \in \mathsf{T}_{n,r}^{(D)}}
\tvar
(
R(\theta)
) 
\le
C_p\frac{2^D}{n}.
\end{align}
By Dudley's entropy integral,
\[
\E\sup_{\theta \in \mathsf{T}_{n,r}^{(D)}} |R (\theta)|
\le
C \int_0^{\operatorname{diam}(\mathsf{d}_R)}
\sqrt{
\log N(\mathsf{T}_{n,r}^{(D)},\mathsf{d}_R,\eta)
}
d\eta
\le
C
\int_0^{2\sigma_\infty}
\sqrt{
\log N(\mathsf{T}_{n,r}^{(D)},\mathsf{d}_R,\eta)
}
d\eta ,
\]
since \(\operatorname{diam}(\mathsf{d}_R)\le 2\sigma_\infty\). The entropy estimate in Proposition~\ref{prop:full-residual-entropy} gives
\[
\E\sup_{\theta \in \mathsf{T}_{n,r}^{(D)}} |R (\theta)|
\le
C
\int_0^{2\sigma_\infty}
\sqrt{
C_pDn2^D(\log n+D)
+
C_pDn2^D
\log_+\left(\frac{C_pDn2^D}{\eta}\right)
}
d\eta.
\]
Using the elementary bound
\[
\int_0^{2\sigma}
\sqrt{
A+\log(2\sigma/\eta)
}
d\eta
\le
C\sigma\sqrt{A+1},
\qquad A\ge 0,
\]
together with \eqref{eq:varBoundExpectedSupResidual}, we obtain
\[
\E\sup_{\theta \in \mathsf{T}_{n,r}^{(D)}} |R(\theta)|
\le
C_p
\sigma_\infty
\sqrt{Dn2^D(\log n+D)}.
\]
The claim follows by applying once more \eqref{eq:varBoundExpectedSupResidual}:
\[
\sigma_\infty
\sqrt{Dn2^D(\log n+D)}
\le
C_p
\left(\frac{2^D}{n}\right)^{1/2}
\sqrt{Dn2^D(\log n+D)}
=
C_p
2^D
\sqrt{D(\log n+D)}.
\]
\end{proof}

\begin{proposition}[High-probability bound, mean-field regime]
\label{prop:circuit-residual-high-prob}
Fix \(p\ge 2\). There exist constants \(C_p,c_p>0\), depending only on
\(p\), such that for every \(D\ge1\) and every \(r\in\mathbb Z_+\),
\[
\P\bigg(
\sup_{\theta \in \mathsf{T}_{n,r}^{(D)}}
|R(\theta)|
\le
C_p 2^{D}\sqrt{D(\log n+D)}
\bigg)
\ge
1-2\exp\left(-c_p 2^D Dn(\log n+D)\right).
\]
\end{proposition}

\begin{proof}
 $R$ is a centered Gaussian process indexed by $\mathsf{T}_{n,r}^{(D)}$. Since the number of \(D\)-tuples of perfect matchings is finite, and since for
fixed matchings the circuit class is a compact subset of a finite-dimensional
parameter space, the process is separable. In what follows, the absolute value
in the supremum is handled by the standard symmetrization of the index set,
\(\mathsf{T}_{n,r}^{(D)}\mapsto \mathsf{T}_{n,r}^{(D)}\times\{-1,1\}\), which replaces
\(R(\theta)\) by \(\pm R(\theta)\). This only changes covering
numbers by a factor of \(2\), and is therefore absorbed into the constants
below. We will henceforth abuse notation and continue to write \(\mathsf{T}_{n,r}^{(D)}\)
for the resulting symmetrized index set. We now apply the Borell-TIS inequality, which states that if
\(\{Z_t:t\in T\}\) is a centered separable Gaussian process with
\(\sigma^2_T:=\sup_{t\in T}\tvar(Z_t)<\infty\), then, for every \(u>0\),
\[
\P(
\sup_{t\in T} Z_t
\ge
\E\sup_{t\in T}Z_t+u
)
\le
2\exp\left(
-\frac{u^2}{2\sigma_T^2}
\right).
\]
Using Proposition~\ref{prop:full-residual-variance}, there exists $A_p < \infty$ such that $\sup_{\theta\in\mathsf{T}_{n,r}^{(D)}}\tvar(R(\theta)) \le A_p2^D/n$. Therefore, for every \(u>0\),
\[
\P(
\sup_{\theta\in\mathsf{T}_{n,r}^{(D)}}
R(\theta)
\ge
\E
\sup_{\theta\in\mathsf{T}_{n,r}^{(D)}}
R(\theta)
+u
)
\le
2\exp\left(
-\frac{u^2}{
2A_p2^D/n
}
\right).
\]

Choosing $u
=
B_p
2^{D}
\sqrt{D(\log n+D)}$ with $B_p > 0$, we have
\[
\frac{u^2}{2A_p2^D/n}
=
\frac{B_p^2}{2A_p}
2^D Dn(\log n+D).
\]
Setting \(c_p := B_p^2/(2A_p)\), we get that, with probability at least $1-2\exp(-c_p 2^D Dn(\log n+D))$,
\[
\sup_{\theta\in\mathsf{T}_{n,r}^{(D)}}
R(\theta)
\le
\E
\sup_{\theta\in\mathsf{T}_{n,r}^{(D)}}
R(\theta)
+
B_p
2^{D}
\sqrt{D(\log n+D)}.
\]
Combining this with Lemma~\ref{lem:expected-sup-residual} and increasing the constant \(C_p\)  gives, with the same probability,
\[
\sup_{\theta\in\mathsf{T}_{n,r}^{(D)}}
|R(\theta)|
\le
C_p
2^{D}
\sqrt{D(\log n+D)}.
\]
\end{proof}

\begin{theorem}[Product-depth $D$ gap, mean-field regime]
\label{thm:log-depth-product-gap-residual}
Fix \(p\ge 2\). Let \(D=D_n\ge 1\) satisfy
\[
4^{D_n}D_n(\log n+D_n)=o(n).
\]
Then, for every \(\varepsilon>0\),
\[
\lim_{n\to\infty}
\sup_{r\in\mathbb Z_+}
\P\left(
    E_{n,r}^{(D_n)}(p)-E_{n,\mathrm{prod}}(p)>\varepsilon
\right)
=0.
\]
Moreover, $E_{n,r}^{(D_n)}(p)\ge E_{n,\mathrm{prod}}(p)$ for every \(n\) and \(r\).
\end{theorem}

\begin{proof}
The lower bound \(E_{n,r}^{(D_n)}(p)\ge E_{n,\mathrm{prod}}(p)\) follows because
the depth-\(D_n\) class contains all physical product states. Indeed, let $|\psi\rangle=\bigotimes_{i=1}^n|\psi^{(i)}\rangle$ be any product state on the physical qubits. Choose a first-layer matching on the \(n+r\) qubits. For each edge containing one or two physical qubits, choose the corresponding two-qubit unitary so that, starting from \(|00\rangle\), the
physical output factor on each physical endpoint \(i\) is
\(|\psi^{(i)}\rangle\). On edges containing no physical qubits, choose the
identity. Choose all gates in layers \(\ell=2,\dots,D_n\) to be the identity.
Then the reduced state on the first \(n\) physical qubits is $\bigotimes_{i=1}^n|\psi^{(i)}\rangle\langle\psi^{(i)}|$. Hence every physical product-state energy is attainable inside
\(\mathsf{S}_{n,r}^{(D_n)}\).

For the reverse bound, write $\langle \phi_\theta|H_{n,p} \otimes I_{\mathrm{anc}}|\phi_\theta \rangle = X_0(\theta) + R(\theta)$. We first bound the supremum of the zeroth-order term. For \(x_1,\dots,x_n\in \B_2^3\), define
the multilinear functional
\[
X_0(x_1,\dots,x_n)
:=
\binom np^{-1/2}
\sum_{I=\{i_1,\dots,i_p\}\in\mcI_p^n}
\sum_{a\in\{1,2,3\}^p}
g_{I,a}
\prod_{s=1}^p x_{i_s}(a_s).
\]
For every depth-\(D_n\) state \(\phi\), its one-site Bloch vector $m_i(\phi)
:=(m_i^\phi(1),m_i^\phi(2),m_i^\phi(3))$ belongs to \(\B_2^3\), and by definition $X_0(\phi)
=
X_0(m_1(\phi),\dots,m_n(\phi))$. Therefore
\[
\sup_{\theta \in \mathsf{T}_{n,r}^{(D_n)}} X_0(\theta)
\le
\sup_{x_1,\dots,x_n\in \B_2^3}
X_0(x_1,\dots,x_n).
\]

We claim that the last supremum is equal to the product-state optimum. Since
\(X_0\) is multilinear and \((\B_2^3)^n\) is compact, the supremum is attained.
Fixing all variables except \(x_i\), the map $x_i\mapsto X_0(x_1,\dots,x_n)$
is linear on \(\B_2^3\). Hence, unless this linear functional is identically
zero, its maximum over \(\B_2^3\) is attained on the sphere \(\mcS^2\). If it is
identically zero, any choice of \(x_i\in \mcS^2\) is also optimal. Iterating this
coordinate by coordinate, there is a maximizer with $x_1,\dots,x_n\in \mcS^2.$ The points of \(\mcS^2\) are exactly the Bloch vectors of pure one-qubit states.
Thus
\[
\frac1{\sqrt n}
\sup_{x_1,\dots,x_n\in \B_2^3}
X_0(x_1,\dots,x_n)
=
\frac1{\sqrt n}
\sup_{\psi\in \mathsf{S}_{n,\mathrm{prod}}}
\langle \psi|H_{n,p}|\psi\rangle
=
E_{n,\mathrm{prod}}(p).
\]
Consequently, $ \frac{1}{\sqrt{n}} \sup_{\theta \in \mathsf{T}_{n,r}^{(D_n)}} X_0(\theta) \le E_{n,\mathrm{prod}}(p)$.

It remains to control the residual. By
Proposition~\ref{prop:circuit-residual-high-prob}, with probability tending to
one, uniformly over $r$
\[
\frac1{\sqrt n}
\sup_{\theta \in \mathsf{T}_{n,r}^{(D_n)}}
|
R(\theta)
|
\le
C_p
\sqrt{
\frac{4^{D_n}D_n(\log n+D_n)}{n}
}.
\]
By assumption, $4^{D_n}D_n(\log n+D_n)=o(n)$, and so the right-hand side is \(o_n(1)\) in probability. Hence $E_{n,r}^{(D_n)}(p)
\le
E_{n,\mathrm{prod}}(p)+o_n(1)$ in probability, uniformly over $r$. Combining this with the lower bound proves the first claim.
\end{proof}

\begin{proof}[Proof of Theorem~\ref{thm:main-meanfield-gap}, mean-field case]
By Theorem~\ref{thm:benchmarks}, there exists a universal constant
\(c>0\) such that, for every fixed \(p\ge3\), $E_{n,\mathrm{gs}}(p) \ge c3^{p/2}/p$ with probability tending to one. On the other hand, by Theorem~\ref{thm:benchmarks}, for every \(\varepsilon>0\) and all
sufficiently large $p$, $E_{n,\mathrm{prod}}(p)
    \le
    (1+\varepsilon)\sqrt{2\log p}$ with probability tending to one. Therefore, with probability tending to one,
\[
    E_{n,\mathrm{gs}}(p)-E_{n,\mathrm{prod}}(p)
    \ge
    c\frac{3^{p/2}}{p}
    -
    (1+\varepsilon)\sqrt{2\log p}.
\]
Since \(3^{p/2}/p\gg\sqrt{\log p}\) as \(p\to\infty\), there exists
\(p_0\ge3\) such that, for every fixed \(p\ge p_0\), the right-hand side is
bounded below by \(2c_p\) for some \(c_p>0\). Thus, $E_{n,\mathrm{gs}}(p)-E_{n,\mathrm{prod}}(p)
    \ge
    2c_p$ with probability tending to one.

Choose \(\delta>0\) small enough that \(\delta<1/\log 4\). If \(D_n\le \delta\log n\), then
\[
    4^{D_n}D_n(\log n+D_n)
    \le
    n^{\delta\log 4}\,\delta\log n\, (1+\delta)\log n
    =
    O_\delta\bigl(n^{\delta\log 4}(\log n)^2\bigr)
    =
    o(n).
\]
Therefore the hypothesis of
Theorem~\ref{thm:log-depth-product-gap-residual} is satisfied. Hence $E_{n,r}^{(D_n)}(p)-E_{n,\mathrm{prod}}(p)=o_n(1)$ in probability, uniformly over $r$. Equivalently, for every \(\eta>0\),
\[
\lim_{n\to\infty}
\sup_{r\in\mathbb Z_+}
\P\left(
    E_{n,r}^{(D_n)}(p)-E_{n,\mathrm{prod}}(p)>\eta
\right)
=0.
\]
Taking $\eta= c_p,$ we have, with probability tending to one uniformly over $r,$ $E_{n,r}^{(D_n)}(p)-E_{n,\mathrm{prod}}(p) \le c_p .$

Combining the two estimates gives
\begin{align*}
E_{n,\mathrm{gs}}(p)-E_{n,r}^{(D_n)}(p)
&=
\left(
E_{n,\mathrm{gs}}(p)-E_{n,\mathrm{prod}}(p)
\right)
-
\left(
E_{n,r}^{(D_n)}(p)-E_{n,\mathrm{prod}}(p)
\right)\ge 2c_p - c_p = c_p,
\end{align*}
with probability tending to one uniformly over $r$. Equivalently,
\[
\lim_{n\to\infty}
\sup_{r\in\mathbb Z_+}
\P\left(
    E_{n,\mathrm{gs}}(p)-E_{n,r}^{(D_n)}(p)<c_p
\right)
=0.
\]
\end{proof}

\subsection{Growing-average-degree diluted regime} \label{ssec:growing-average-degree-diluted-regime}

In this section we prove the growing-average-degree diluted ground-state depth separation, which is the second statement of Theorem~\ref{thm:main-meanfield-gap}. The proof follows the same general
strategy as in the mean-field regime, but the residual process now contains the
Bernoulli mask \(\xi_I\) and is normalized by \(m_n^{-1/2}\). Thus, in addition
to controlling the Gaussian fluctuations of the residual, we must also control
the random sparsification uniformly over the depth-\(D\) circuit class.

The first input is the sparse residual variance bound,
Lemma~\ref{lem:sparse-full-residual-variance}. Conditionally on the
sparsification variables, the variance is a sparse average of the residual
coefficient sizes:
\[
    \operatorname{Var}_g(R^{(\beta)}(\theta)\mid \xi)
    =
    \frac1{m_n}
    \sum_{I\in\mcI_p^n}
    \xi_I
    \sum_{a\in\{1,2,3\}^p}
    |\chi_{I,a}(\theta)|^2 .
\]
The mean-field variance bound controls the mean of this average. To make the bound
uniform in \(\theta\), we apply Bernstein's inequality on a finite
pruned and relabeled net of the circuit class. This produces the additional
sparse sampling cost $Dn2^D(\log n+D)/m_n$, which is absent in the mean-field regime.

The second input is the conditional entropy bound,
Lemma~\ref{lem:sparse-residual-entropy}. On the event
\(|\mcE_p|\le 2m_n\), the sparse canonical metric is controlled by the same
operator-norm metric used in the mean-field entropy argument. Therefore the same
pruning and relabeling mechanism applies. Namely, the circuit can be restricted to the joint backward support of the physical outputs, involving at most \(n2^D\) qubits, and the active ancillas can be relabeled into a canonical block. This gives an entropy bound independent of $r$, with the same \(2^D\) cost as in the mean-field ancilla setting.

Combining the conditional variance and entropy estimates gives the expected
supremum bound in Lemma~\ref{lem:sparse-residual-expected-sup}. A Markov
argument then yields the sparse residual negligibility result,
Lemma~\ref{lem:sparse-residual-negligibility}, under the growing-average-degree diluted
depth condition. This gives the growing-average-degree diluted depth-to-product comparison,
Theorem~\ref{thm:sparse-log-depth-product-gap}.

Finally, to convert the depth-to-product comparison into a ground-state
depth separation, we use
Corollary~\ref{cor:sparse-mean-field-transfer-hp} from
Appendix~\ref{app:universality-transfer}. This transfers the mean-field
ground-state lower bound and product-state upper bound to the
growing-average-degree diluted model. Combining these transferred bounds
with the depth-to-product comparison proves the growing-average-degree
statement of Theorem~\ref{thm:main-meanfield-gap}.

We first briefly recall the growing-average-degree diluted setup from
\eqref{eq:sparseHamiltonian}. In this regime, $\kappa_n=n^{-\beta}$ and the sparsification variables
\(\{\xi_I:I\in\mcI_p^n\}\) are independent
\(\operatorname{Bernoulli}(\kappa_n)\) random variables, independent of the
Gaussian disorder. The retained $p$-tuples form the random set $\mcE_p:=\{I\in\mcI_p^n:\xi_I=1\}$,
so that
\[
    |\mcE_p|=\sum_{I\in\mcI_p^n}\xi_I,
    \qquad
    \E |\mcE_p|=
    m_n = \Theta_p(n^{p-\beta}).
\]
The deterministic normalization \(m_n^{-1/2}\) in
\eqref{eq:sparseHamiltonian} is natural because \(|\mcE_p|/m_n\to1\) in
probability whenever \(m_n\to\infty\).

The sparsification variables define a random $p$-uniform hypergraph $\mcF_p=([n],\mcE_p)$. Its vertices are the physical qubits, and each retained
\(I\in\mcE_p\) is a hyperedge corresponding to one retained $p$-body
interaction. The degree of a vertex \(i\) is $\deg_{\mcF_p}(i)
    :=
    |\{I\in\mcE_p:i\in I\}|$. Since there are \(\binom{n-1}{p-1}\) possible $p$-tuples containing \(i\),
\[
    \E\deg_{\mcF_p}(i)
    =
    \kappa_n\binom{n-1}{p-1}
    = \Theta_p(n^{p-1-\beta}).
\]
We refer to \(\beta<p-1\) as the growing-average-degree diluted regime, since the expected
degree of each physical vertex diverges.

We now introduce the sparse analogue of the residual process. The residual
coefficients \(\chi_{I,a}(\theta)\) are the same as in the mean-field setting. The
sparse residual is obtained by inserting the Bernoulli mask \(\xi_I\) and using
the normalization \(m_n^{-1/2}\):
\[
R^{(\beta)}(\theta)
:=
m_n^{-1/2}
\sum_{I\in\mcI_p^n}
\xi_I
\sum_{a\in\{1,2,3\}^p}
g_{I,a}\chi_{I,a}(\theta).
\]
Equivalently, for \(\theta\in\mathsf{T}_{n,r}^{(D)}\), $\langle \phi_\theta|
H_{n,p}^{(\beta)}\otimes I_{\mathrm{anc}}
|\phi_\theta\rangle
=
X_0^{(\beta)}(\theta)
+
R^{(\beta)}(\theta)$, where
\[
X_0^{(\beta)}(\theta)
:=
m_n^{-1/2}
\sum_{I\in\mcI_p^n}
\xi_I
\sum_{a\in\{1,2,3\}^p}
g_{I,a}
\prod_{s=1}^p m_{i_s}^{\phi_\theta}(a_s).
\]
Thus the only difference from the mean-field decomposition is the sparsification
factor \(\xi_I\) and the normalization by \(m_n^{1/2}\) rather than
\(\binom np^{1/2}\).

\begin{lemma}[Variance bound, growing-average-degree diluted regime]
\label{lem:sparse-full-residual-variance}
Let \(D=D_n\ge1\), and assume
\[
\frac{2^{D_n}}{n}
+
\frac{D_n2^{D_n}n(\log n+D_n)}{m_n}
=o(1).
\]
Then there exists \(C_p<\infty\) such that
\[
\lim_{n\to\infty}
\sup_{r\in\mathbb Z_+}
\mathbb P_\xi\left(
\sup_{\theta\in\mathsf{T}_{n,r}^{(D_n)}}
\operatorname{Var}_g\!\left(R^{(\beta)}(\theta)\mid \xi\right)
>
C_p
\left[
\frac{2^{D_n}}{n}
+
\frac{D_n2^{D_n}n(\log n+D_n)}{m_n}
\right]
\right)
=0.
\]
\end{lemma}

\begin{proof}
Conditioned on the sparsification variables, the variance of the sparse
residual process is
\[
\tvar_g
(R^{(\beta)}(\theta)\mid \xi)
=
\frac1{m_n}
\sum_{I\in\mcI_p^n}\xi_I F_\theta(I),
\qquad 
F_\theta(I):=\sum_{a\in\{1,2,3\}^p}|\chi_{I,a}(\theta)|^2.
\]
Thus we need to control the sparse average $m_n^{-1}
\sum_{I\in\mcI_p^n}\xi_I F_\theta(I)$ uniformly over \(\theta\in\mathsf{T}_{n,r}^{(D)}\).

We first record the mean and boundedness properties of \(F_\theta\). For the
mean-field residual process, since the coefficients are independent Gaussians, we
get
\[
\tvar_g\bigl(R(\theta)\bigr)
=
\frac1{\binom np}
\sum_{I\in\mcI_p^n}
\sum_{a\in\{1,2,3\}^p}
|\chi_{I,a}(\theta)|^2
=
\frac1{\binom np}
\sum_{I\in\mcI_p^n}
F_\theta(I).
\]
Therefore, by Proposition~\ref{prop:full-residual-variance},
\[
\frac1{\binom np}
\sum_{I\in\mcI_p^n}
F_\theta(I)
\le
C_p\frac{2^D}{n}
\]
uniformly in \(\theta\in\mathsf{T}_{n,r}^{(D)}\). Also,
\(0\le F_\theta(I)\le C_p\). To see this, note that for each \(a\),
\[
|\chi_{I,a}(\theta)|
\le
|\langle\phi_{\theta}|P_I^a\otimes I_{\mathrm{anc}}|\phi_{\theta}\rangle|
+
\big|\prod_{s=1}^p m_{i_s}^{\phi_\theta}(a_s)\big|
\le 2,
\]
and there are \(3^p\) choices of \(a\).

The remaining issue is uniformity in \(\theta\). For each fixed \(\theta\), the
quantity $m_n^{-1}\sum_{I\in\mcI_p^n}\xi_I F_\theta(I)$ is a bounded Bernoulli average. We first apply Bernstein's inequality, on a finite net of the circuit class. We then pass
from the net to all \(\theta\in\mathsf{T}_{n,r}^{(D)}\) by proving a Lipschitz
estimate for the coefficients \(F_\theta(I)\) with respect to the underlying
circuit unitary.

Set $\mcJ_{n,D}:=C_pDn2^D(\log n+D)$, where \(C_p\) is chosen sufficiently large. The assumption is precisely $\frac{2^D}{n}+\frac{\mcJ_{n,D}}{m_n}=o(1),$ after changing the value of \(C_p\).  Let
\[
\Lambda_{n,D}
:=
\frac{2^D}{n}
+
\frac{\mcJ_{n,D}}{m_n}.
\]
By assumption, \(\Lambda_{n,D}=o(1)\). We now construct a finite net using the
same pruning and relabeling argument as in Proposition~\ref{prop:full-residual-entropy}.
Every circuit parameter $\theta=(M_1,\dots,M_D,U)\in\mathsf{T}_{n,r}^{(D)}$ is equivalent, for all residual coefficients, to a pruned and relabeled circuit $\widetilde\theta
    =
    (\widetilde M_1,\dots,\widetilde M_D,\widetilde U)$ whose nonidentity gates act only on an active set of size at most
\(n2^D\). In particular,
\[
    \chi_{I,a}(\theta)
    =
    \chi_{I,a}(\widetilde\theta)
    \qquad
    \text{for all } I\in\mcI_p^n,\ a\in\{1,2,3\}^p .
\]
It is therefore enough to build a net for these pruned, relabeled  circuits.

Let \(J_0:=Dn2^D\), which is an upper bound on the number of active two-qubit
gate locations in a pruned depth-\(D\) circuit. For each fixed pruned architecture, take a \(\delta\)-net of \(U(4)\) in operator norm for
each active gate location. By the usual telescoping bound for products of gates, if $\delta=\frac{\Lambda_{n,D}}{C_pJ_0}$, then every pruned circuit \(\widetilde U\) has a representative
\(U_0\) in the net such that $\|\widetilde U-U_0\|_{\mathrm{op}}\le \Lambda_{n,D}$. The logarithm of the gate-net size for a fixed architecture is at most $C_pJ_0\log(C_pJ_0/\Lambda_{n,D})$. The number of pruned  architectures has logarithm at most
\(C Dn2^D(\log n+D)\), again by the counting argument in
Proposition~\ref{prop:full-residual-entropy}. Since $\Lambda_{n,D}\ge 2^D/n$, we have
\[
\log\left(\frac{C_pJ_0}{\Lambda_{n,D}}\right)
\le
\log(C_pD n^2)
\le
C(\log n+D),
\]
after increasing constants. Hence, after taking the union over pruned architectures, we obtain a finite net \(\mathcal N\) such that $\log|\mathcal N|
    \le
    \mcJ_{n,D}$, and every \(\theta\in\mathsf T_{n,r}^{(D)}\) has a pruned 
representative \(\widetilde\theta\) and a net point \(\theta_0\in\mathcal N\)
satisfying $\|\widetilde U-U_0\|_{\mathrm{op}}\le \Lambda_{n,D}$. The construction is independent of \(r\), because pruning removes all gates
invisible to physical observables and the active ancillas are relabeled into a canonical block.

We now apply Bernstein's inequality on the net.
Fix \(\theta_0\in\mathcal N\). Since \(0\le F_{\theta_0}(I)\le C_p\) and
\[
\frac1{\binom np}\sum_{I\in\mcI_p^n}F_{\theta_0}(I)
\le
C_p\frac{2^D}{n},
\]
we have
\[
\E_\xi
\bigg[
\frac1{m_n}\sum_{I\in\mcI_p^n}\xi_I F_{\theta_0}(I)
\bigg]
=
\frac1{\binom np}\sum_{I\in\mcI_p^n}F_{\theta_0}(I)
\le
C_p\frac{2^D}{n}.
\]
Apply Bernstein's inequality to the centered random variables $Z_I
    :=
    \xi_I F_{\theta_0}(I)
    -
    \E_\xi[\xi_I F_{\theta_0}(I)]$.
Since \(0\le F_{\theta_0}(I)\le C_p\), we have \(|Z_I|\le C_p\). Moreover,
\[
\sum_{I\in\mcI_p^n}\E_\xi Z_I^2
\le
\sum_{I\in\mcI_p^n}
\E_\xi[\xi_I F_{\theta_0}(I)^2]
\le
C_p\kappa_n\sum_{I\in\mcI_p^n}F_{\theta_0}(I)
\le
C_p m_n\frac{2^D}{n}.
\]
Also,
\[
\E_\xi
\sum_{I\in\mcI_p^n}\xi_I F_{\theta_0}(I)
=
\kappa_n\sum_{I\in\mcI_p^n}F_{\theta_0}(I)
\le
C_p m_n\frac{2^D}{n}.
\]
Therefore, applying Bernstein's inequality with \(t=3\mcJ_{n,D}\), we obtain
\[
\P_\xi \bigg(
\frac1{m_n}\sum_{I\in\mcI_p^n}\xi_I F_{\theta_0}(I)
>
C_p\bigg[
\frac{2^D}{n}
+
\sqrt{
\frac{2^D}{n}\frac{\mcJ_{n,D}}{m_n}
}
+
\frac{\mcJ_{n,D}}{m_n}
\bigg]
\bigg)
\le
\exp(-3\mcJ_{n,D}).
\]
Since
\[
\sqrt{
\frac{2^D}{n}\frac{\mcJ_{n,D}}{m_n}
}
\le
\frac12\frac{2^D}{n}
+
\frac12\frac{\mcJ_{n,D}}{m_n},
\]
we get
\[
\P_\xi\bigg(
\frac1{m_n}\sum_{I\in\mcI_p^n}\xi_I F_{\theta_0}(I)
>
C_p\bigg[
\frac{2^D}{n}
+
\frac{\mcJ_{n,D}}{m_n}
\bigg]
\bigg)
\le
\exp(-3\mcJ_{n,D}).
\]
Taking a union bound over \(\mathcal N\), and using
\(\log|\mathcal N|\le \mcJ_{n,D}\), gives, with probability tending to one,
\[
\sup_{\theta_0\in\mathcal N}
\frac1{m_n}
\sum_{I\in\mcI_p^n}\xi_I F_{\theta_0}(I)
\le
C_p
\left[
\frac{2^D}{n}
+
\frac{\mcJ_{n,D}}{m_n}
\right].
\]

It remains to pass from the net to all of \(\mathsf{T}_{n,r}^{(D)}\). Since
\(\mcJ_{n,D}/m_n=o(1)\), we have \(m_n\to\infty\), and therefore $|\mcE_p|\le 2m_n$ with probability tending to one. We work on the intersection of this event and the event on which the estimate in the previous display holds.

Let \(\theta\in\mathsf{T}_{n,r}^{(D)}\), let \(\widetilde\theta\) be its
pruned and relabeled representative, and choose
\(\theta_0\in\mathcal N\) in the same pruned architecture such that
$\|\widetilde U-U_0\|_{\mathrm{op}}
    \le
    \Lambda_{n,D}$.
Since \(\chi_{I,a}(\theta)=\chi_{I,a}(\widetilde\theta)\), it suffices to
compare \(F_{\tilde\theta}(I)\) and \(F_{\theta_0}(I)\).

We claim that, for every \(I\in\mcI_p^n\), $|F_{\tilde\theta}(I)-F_{\theta_0}(I)|
    \le
    C_p\Lambda_{n,D}$.
To see this, note that for each label \(a\in\{1,2,3\}^p\),
\[
\big|
\langle 0^{n+r}|
\widetilde U^*(P_I^a\otimes I_{\mathrm{anc}})\widetilde U
|0^{n+r}\rangle
-
\langle 0^{n+r}|
U_0^*(P_I^a\otimes I_{\mathrm{anc}})U_0
|0^{n+r}\rangle
\big|
\le
2\|\widetilde U-U_0\|_{\mathrm{op}}
\le
2\Lambda_{n,D},
\]
since \(\|P_I^a\otimes I_{\mathrm{anc}}\|_{\mathrm{op}}=1\). Similarly, for every physical
one-site Pauli observable,
\[
|m_i^{\phi_{\tilde \theta}}(b)-m_i^{\phi_{\theta_0}}(b)|
\le
2\|\widetilde U-U_0\|_{\mathrm{op}}
\le
2\Lambda_{n,D}.
\]
Since all one-site means have absolute value at most \(1\), the product
difference satisfies
\[
\big|
\prod_{s=1}^p m_{i_s}^{\phi_{\tilde \theta}}(a_s)
-
\prod_{s=1}^p m_{i_s}^{\phi_{\theta_0}}(a_s)
\big|
\le
C_p\Lambda_{n,D}.
\]
Therefore $|\chi_{I,a}(\widetilde\theta)-\chi_{I,a}(\theta_0)|
\le
C_p\Lambda_{n,D}$. Since \(|\chi_{I,a}(\widetilde\theta)|\le2\) and
\(|\chi_{I,a}(\theta_0)|\le2\), we also have
\[
\big|
|\chi_{I,a}(\widetilde\theta)|^2
-
|\chi_{I,a}(\theta_0)|^2
\big |
\le
C_p\Lambda_{n,D}.
\]
Summing over the \(3^p\) labels \(a\) proves the claim.

Using \(\chi_{I,a}(\theta)=\chi_{I,a}(\widetilde\theta)\), we get
\[
\frac1{m_n}
\sum_{I\in\mcI_p^n}
\xi_I
|F_{\theta}(I)-F_{\theta_0}(I)|
=
\frac1{m_n}
\sum_{I\in\mcI_p^n}
\xi_I
|F_{\tilde\theta}(I)-F_{\theta_0}(I)|  
\le
C_p\Lambda_{n,D}\frac{|\mcE_p|}{m_n}
\le
C_p\Lambda_{n,D}.
\]
Combining this with the net bound yields
\[
\sup_{\theta\in\mathsf{T}_{n,r}^{(D)}}
\frac1{m_n}
\sum_{I\in\mcI_p^n}\xi_I F_\theta(I)
\le
C_p
\left[
\frac{2^D}{n}
+
\frac{\mcJ_{n,D}}{m_n}
\right].
\]
Finally, since
\[
\tvar_g
(
R^{(\beta)}(\theta)\mid \xi
)
=
\frac1{m_n}
\sum_{I\in\mcI_p^n}
\xi_I
F_\theta(I),
\]
the desired uniform variance bound follows.
\end{proof}

\begin{lemma}[Entropy bound, growing-average-degree diluted regime]
\label{lem:sparse-residual-entropy}
Let $\mathsf E_{n,1} := \{|\mcE_p|\le 2m_n\}$ and condition on a realization of the sparsification variables \(\xi\in\mathsf E_{n,1}\).
Let \(\mathsf{d}_\xi\) be the conditional canonical metric of the sparse residual
process \(R^{(\beta)}\), namely
\[
\mathsf{d}_\xi(\theta,\theta')^2
:=
\tvar_g  (
R^{(\beta)}(\theta)-R^{(\beta)}(\theta')
\mid \xi
),
\qquad
\theta,\theta'\in\mathsf{T}_{n,r}^{(D)} .
\]
Then, for every \(r\in\mathbb Z_+\) and every \(\eta>0\),
\[
\log N(\mathsf{T}_{n,r}^{(D)},\mathsf{d}_\xi,\eta)
\le
C_pDn2^D(\log n+D)
+
C_pDn2^D
\log_+\left(
\frac{C_pDn2^D}{\eta}
\right).
\]
\end{lemma}

\begin{proof}
The proof is the same pruning and relabeling argument as in
Proposition~\ref{prop:full-residual-entropy}. We only record the changes caused by
the sparsification.

Conditioned on the sparsification variables, the canonical metric is
\[
\mathsf{d}_\xi(\theta,\theta')^2
=
\frac1{m_n}
\sum_{I\in\mcI_p^n}
\xi_I
\sum_{a\in\{1,2,3\}^p}
|\chi_{I,a}(\theta)-\chi_{I,a}(\theta')|^2 .
\]
By the coefficient Lipschitz estimate proved in
Proposition~\ref{prop:full-residual-entropy}, for every \(I\in\mcI_p^n\) and
\(a\in\{1,2,3\}^p\),
\[
    |\chi_{I,a}(\theta)-\chi_{I,a}(\theta')|
    \le
    C_p\|U-U'\|_{\mathrm{op}} .
\]
Therefore, on \(\mathsf E_{n,1}\),
\[
\mathsf{d}_\xi(\theta,\theta')^2
\le
\frac1{m_n}
\sum_{I\in\mcE_p}
\sum_{a\in\{1,2,3\}^p}
C_p\|U-U'\|_{\mathrm{op}}^2  
\le
C_p\|U-U'\|_{\mathrm{op}}^2.
\]
Hence \(\mathsf{d}_\xi(\theta,\theta') \le C_p\|U-U'\|_{\mathrm{op}}\).

The remainder of the proof follows exactly the covering argument from
Proposition~\ref{prop:full-residual-entropy}: after pruning to the joint backward
support of the physical outputs and relabeling the active ancillas into a
canonical block, the relevant circuits act on at most \(n2^D\) qubits. 
\end{proof}

\begin{lemma}[Expected supremum, growing-average-degree diluted regime]
\label{lem:sparse-residual-expected-sup}
Define
\[
    \mathsf E_{n,1}:=\{|\mcE_p|\le 2m_n\}
\]
and, for each \(r\in\mathbb Z_+\), define
\[
\mathsf E_{n,r}^{\mathrm{var}}
:=
\left\{
\sup_{\theta\in\mathsf{T}_{n,r}^{(D)}}
\tvar_g
\left(
R^{(\beta)}(\theta)\mid \xi
\right)
\le
C_p
\left[
\frac{2^D}{n}
+
\frac{Dn2^D(\log n+D)}{m_n}
\right]
\right\}.
\]
Set $\mathsf E_{n,r}:=\mathsf E_{n,1}\cap \mathsf E_{n,r}^{\mathrm{var}}$. On the event \(\mathsf E_{n,r}\), we have
\[
\E_g
\bigg[
\frac1{\sqrt n}
\sup_{\theta\in\mathsf{T}_{n,r}^{(D)}}
|R^{(\beta)}(\theta)|
\,\big | \, \xi
\bigg ]
\le
C_p
\left\{
D2^D (\log n+D)
\left[
\frac{2^D}{n}
+
\frac{Dn2^D(\log n+D)}{m_n}
\right]
\right\}^{1/2}.
\]
\end{lemma}

\begin{proof}
As in the proof of Proposition~\ref{prop:circuit-residual-high-prob}, the
absolute value in the supremum is handled by symmetrizing the index set. This
only changes covering numbers by a factor of \(2\), which is absorbed into the
constants.

Fix \(r\in\mathbb Z_+\), and condition on a realization
\(\xi\in\mathsf E_{n,r}\). Let \(\mathsf d_\xi\) denote the conditional
canonical metric of the Gaussian process \(R^{(\beta)}\), conditioned on the
sparsification variables \(\xi\). That is, for
\(\theta,\theta'\in\mathsf{T}_{n,r}^{(D)}\),
\[
\mathsf{d}_\xi(\theta,\theta')^2
:=
\tvar_g
\left(
R^{(\beta)}(\theta)-R^{(\beta)}(\theta')
\mid \xi
\right).
\]

On \(\mathsf E_{n,1}\), Lemma~\ref{lem:sparse-residual-entropy} gives, for
every \(\eta>0\),
\[
\log N(\mathsf{T}_{n,r}^{(D)},\mathsf{d}_\xi,\eta)
\le
C_pDn2^D(\log n+D)
+
C_pDn2^D
\log_+\left(
\frac{C_pDn2^D}{\eta}
\right).
\]
Next, set
\[
\Lambda_{n,D}
:=
\frac{2^D}{n}
+
\frac{Dn2^D(\log n+D)}{m_n}.
\]
On \(\mathsf E_{n,r}^{\mathrm{var}}\),
\[
\sup_{\theta\in\mathsf{T}_{n,r}^{(D)}}
\tvar_g
\left(
R^{(\beta)}(\theta)\mid \xi
\right)
\le
C_p\Lambda_{n,D}.
\]
Therefore,
\[
\operatorname{diam}(\mathsf{T}_{n,r}^{(D)},\mathsf d_\xi)
\le
2\left(
\sup_{\theta\in\mathsf{T}_{n,r}^{(D)}}
\tvar_g(R^{(\beta)}(\theta)\mid\xi)
\right)^{1/2}
\le
C_p\sqrt{\Lambda_{n,D}}.
\]

Thus Dudley's entropy integral gives
\[
\E_g
\left[
\sup_{\theta\in\mathsf{T}_{n,r}^{(D)}}|R^{(\beta)}(\theta)|
\,\big|\,\xi
\right]
\le
C
\int_0^{C_p\sqrt{\Lambda_{n,D}}}
\sqrt{
C_pDn2^D(\log n+D)
+
C_pDn2^D
\log_+\left(
\frac{C_pDn2^D}{\eta}
\right)
}
\,d\eta .
\]
Using the elementary entropy-integral bound from
Lemma~\ref{lem:expected-sup-residual}, together with
\(\Lambda_{n,D}\ge 2^D/n\), we obtain
\[
\E_g
\left[
\sup_{\theta\in\mathsf{T}_{n,r}^{(D)}}|R^{(\beta)}(\theta)|
\,\big|\,\xi
\right]
\le
C_p
\sqrt{\Lambda_{n,D}}\,
\sqrt{Dn2^D(\log n+D)}.
\]
Dividing by \(\sqrt n\) gives
\[
\E_g
\bigg[
\frac1{\sqrt n}
\sup_{\theta\in\mathsf{T}_{n,r}^{(D)}}
|R^{(\beta)}(\theta)|
\,\big | \, \xi
\bigg]
\le
C_p
\left\{
D2^D(\log n+D)\Lambda_{n,D}
\right\}^{1/2}.
\]
Substituting the definition of \(\Lambda_{n,D}\) proves the claim.
\end{proof}

\begin{lemma}[High-probability bound, growing-average-degree diluted regime]
\label{lem:sparse-residual-negligibility}
Assume
\[
D_n2^{D_n}(\log n+D_n)
\left[
\frac{2^{D_n}}{n}
+
\frac{D_n n2^{D_n}(\log n+D_n)}{m_n}
\right]
=o(1).
\]
Then, for every \(\varepsilon>0\),
\[
\lim_{n\to\infty}
\sup_{r\in\mathbb Z_+}
\mathbb P\left(
\frac1{\sqrt n}
\sup_{\theta\in\mathsf{T}_{n,r}^{(D_n)}}
|R^{(\beta)}(\theta)|
>\varepsilon
\right)
=0.
\]
\end{lemma}

\begin{proof}
Let $\mathsf E_{n,1}:=\{|\mcE_p|\le 2m_n\}$. For each \(r\in\mathbb Z_+\), let \(\mathsf E_{n,r}^{\mathrm{var}}\) and
\(\mathsf E_{n,r}\) be the events defined in
Lemma~\ref{lem:sparse-residual-expected-sup}. That is, $\mathsf E_{n,r}
    :=
    \mathsf E_{n,1}\cap \mathsf E_{n,r}^{\mathrm{var}}$. We first claim that
\[
\lim_{n\to\infty}
\sup_{r\in\mathbb Z_+}
\P_\xi(\mathsf E_{n,r}^c)
=0.
\]
Indeed, \(|\mcE_p|=\sum_{I\in\mcI_p^n}\xi_I\),
\(\E_\xi|\mcE_p|=m_n\), and \(m_n\to\infty\). Bernstein's inequality gives
\[
\P_\xi(\mathsf E_{n,1}^c)
=
\P_\xi(|\mcE_p|>2m_n)
\le
\exp(-c m_n)
=
o(1).
\]
On the other hand, Lemma~\ref{lem:sparse-full-residual-variance} gives
$\lim_{n\to\infty}
\sup_{r\in\mathbb Z_+}
\P_\xi\left((\mathsf E_{n,r}^{\mathrm{var}})^c\right)
=0$. Therefore
\[
\sup_{r\in\mathbb Z_+}
\P_\xi(\mathsf E_{n,r}^c)
\le
\P_\xi(\mathsf E_{n,1}^c)
+
\sup_{r\in\mathbb Z_+}
\P_\xi\left((\mathsf E_{n,r}^{\mathrm{var}})^c\right)
\to0.
\]

On \(\mathsf E_{n,r}\), Lemma~\ref{lem:sparse-residual-expected-sup} gives
\[
\E_g
\bigg[
\frac1{\sqrt n}
\sup_{\theta\in\mathsf{T}_{n,r}^{(D_n)}}
|R^{(\beta)}(\theta)|
\mid \xi
\bigg]
\le
C_p
\bigg\{
D_n2^{D_n}(\log n+D_n)
\bigg[
\frac{2^{D_n}}{n}
+
\frac{D_n n 2^{D_n}(\log n+D_n)}{m_n}
\bigg]
\bigg\}^{1/2}.
\]
The right-hand side is \(o(1)\) by assumption.

Let \(\varepsilon>0\), and define
\[
A_{n,r}
:=
\left\{
\frac1{\sqrt n}
\sup_{\theta\in\mathsf{T}_{n,r}^{(D_n)}}
|R^{(\beta)}(\theta)|
>
\varepsilon
\right\}.
\]
For every fixed \(r\), Markov's inequality gives, on \(\mathsf E_{n,r}\),
\[
\P_g(A_{n,r}\mid \xi)
\le
\frac{C_p}{\varepsilon}
\bigg(
D_n2^{D_n}(\log n+D_n)
\bigg[
\frac{2^{D_n}}{n}
+
\frac{D_n n 2^{D_n}(\log n+D_n)}{m_n}
\bigg]
\bigg)^{1/2}.
\]

Consequently, for every \(r\in\mathbb Z_+\),
\[
\begin{aligned}
\P_{\xi,g}(A_{n,r})
&\le
\P_\xi(\mathsf E_{n,r}^c)
+
\E_\xi\left[
\indicator_{\mathsf E_{n,r}}
\P_g(A_{n,r}\mid \xi)
\right]  \\
&\le
\P_\xi(\mathsf E_{n,r}^c)
+
\frac{C_p}{\varepsilon}
\left(
D_n 2^{D_n}(\log n+D_n)
\left[
\frac{2^{D_n}}{n}
+
\frac{D_n n2^{D_n}(\log n+D_n)}{m_n}
\right]
\right)^{1/2}.
\end{aligned}
\]
Taking the supremum over \(r\in\mathbb Z_+\), we obtain
\[
\begin{aligned}
\sup_{r\in\mathbb Z_+}\P_{\xi,g}(A_{n,r})
&\le
\sup_{r\in\mathbb Z_+}\P_\xi(\mathsf E_{n,r}^c)
+
\frac{C_p}{\varepsilon}
\left(
D_n 2^{D_n}(\log n+D_n)
\left[
\frac{2^{D_n}}{n}
+
\frac{D_n n2^{D_n}(\log n+D_n)}{m_n}
\right]
\right)^{1/2}.
\end{aligned}
\]
The first term tends to zero by the preceding claim, and the second term tends
to zero by the assumed depth condition. Hence
\[
\lim_{n\to\infty}
\sup_{r\in\mathbb Z_+}
\P_{\xi,g}\left(
\frac1{\sqrt n}
\sup_{\theta\in\mathsf{T}_{n,r}^{(D_n)}}
|R^{(\beta)}(\theta)|
>
\varepsilon
\right)
=0,
\]
which proves the claim.
\end{proof}
\begin{theorem}[Product-depth $D$ gap, growing-average-degree diluted regime]
\label{thm:sparse-log-depth-product-gap}
Fix \(p\ge2\), and let \(D=D_n \ge 1\) satisfy
\[
D_n2^{D_n}(\log n+D_n)
\left[
\frac{2^{D_n}}{n}
+
\frac{D_n n2^{D_n}(\log n+D_n)}{m_n}
\right]
=o(1).
\]
Then, for every \(\varepsilon>0\),
\[
\lim_{n\to\infty}
\sup_{r\in\mathbb Z_+}
\P\left(
E_{n,r}^{(\beta,D_n)}(p)
-
E_{n,\mathrm{prod}}^{(\beta)}(p)
>
\varepsilon
\right)
=0.
\]
Moreover,
\[
E_{n,r}^{(\beta,D_n)}(p)\ge E_{n,\mathrm{prod}}^{(\beta)}(p)
\]
for every \(n,r\), and \(D_n\ge1\).
\end{theorem}

\begin{proof}
The proof is identical to the mean-field depth-to-product comparison,
Theorem~\ref{thm:log-depth-product-gap-residual}, with the mean-field Hamiltonian
\(H_{n,p}\) replaced by the sparse Hamiltonian \(H_{n,p}^{(\beta)}\). The
zeroth-order term \(X_0^{(\beta)}\) is bounded by the sparse product-state
optimum \(E_{n,\mathrm{prod}}^{(\beta)}(p)\) by the same multilinearity
argument, while Lemma~\ref{lem:sparse-residual-negligibility} gives
\[
\frac1{\sqrt n}
\sup_{\theta\in\mathsf{T}_{n,r}^{(D_n)}}
|R^{(\beta)}(\theta)|
=
o_n(1)
\]
in probability, uniformly over $r$, under the stated depth condition.
\end{proof}

\begin{proof}[Proof of Theorem~\ref{thm:main-meanfield-gap}, growing-average-degree diluted case]
By Corollary~\ref{cor:sparse-mean-field-transfer-hp}, the mean-field ground-state and product-state high-probability bounds transfer to the growing-average-degree diluted model. In particular, if \(c>0\) is the universal constant from Theorem~\ref{thm:benchmarks}, then for every \(\varepsilon>0\), $E_{n,\mathrm{gs}}^{(\beta)}(p) \ge c3^{p/2}/p-\varepsilon$ with probability tending to one. Also, by the same corollary, for every
\(\alpha>0\), for all sufficiently large fixed $p$, and every
\(\varepsilon>0\), $E_{n,\mathrm{prod}}^{(\beta)}(p)
\le
(1+\alpha)\sqrt{2\log p}+\varepsilon$ with probability tending to one. Therefore, with probability tending to one,
\[
E_{n,\mathrm{gs}}^{(\beta)}(p)
-
E_{n,\mathrm{prod}}^{(\beta)}(p)
\ge
c\frac{3^{p/2}}{p}
-
(1+\alpha)\sqrt{2\log p}
-
2\varepsilon .
\]
Since \(3^{p/2}/p\gg\sqrt{\log p}\) as \(p\to\infty\), there exists
\(p_0\ge3\) such that, for every fixed \(p\ge p_0\), we may choose
\(\alpha>0\) and \(\varepsilon>0\) sufficiently small so that the right-hand
side is bounded below by \(2c_p\) for some \(c_p>0\). Hence $E_{n,\mathrm{gs}}^{(\beta)}(p) - E_{n,\mathrm{prod}}^{(\beta)}(p) \ge 2c_p$ with probability tending to one.

Choose \(\delta_{p,\beta}>0\) small enough that $\delta_{p,\beta} \log 4<\min\{1,p-\beta-1\}$. Since \(m_n = \Theta_p(n^{p-\beta})\), if \(D_n\le \delta_{p,\beta}\log n\), then 
\begin{align*}
& D_n2^{D_n}(\log n+D_n)
\left[
    \frac{2^{D_n}}{n}
    +
    \frac{D_n n2^{D_n}(\log n+D_n)}{m_n}
\right] \\
&\qquad \le
C_{p,\beta}(\log n)^2 n^{\delta_{p,\beta}\log 4-1}
+
C_{p,\beta}(\log n)^4 n^{\delta_{p,\beta}\log 4-(p-\beta-1)}
=
o(1).
\end{align*}
Therefore the hypothesis of Theorem~\ref{thm:sparse-log-depth-product-gap} is satisfied. Hence $E_{n,r}^{(\beta, D_n)}(p)
-
E_{n,\mathrm{prod}}^{(\beta)}(p)
=
o_n(1)$ in probability, uniformly over $r$. Equivalently, for every \(\eta>0\),
\[
\lim_{n\to\infty}
\sup_{r\in\mathbb Z_+}
\P(
E_{n,r}^{(\beta, D_n)}(p)
-
E_{n,\mathrm{prod}}^{(\beta)}(p)
>
\eta )
=0.
\]

Combining the last two estimates yields
\begin{align*}
E_{n,\mathrm{gs}}^{(\beta)}(p)
-
E_{n,r}^{(\beta, D_n)}(p)
&=
[
E_{n,\mathrm{gs}}^{(\beta)}(p)
-
E_{n,\mathrm{prod}}^{(\beta)}(p)
]
-
[
E_{n,r}^{(\beta, D_n)}(p)
-
E_{n,\mathrm{prod}}^{(\beta)}(p)
]  \\
&\ge
2c_p-o_n(1)
\ge
c_p
\end{align*}
with probability tending to one, uniformly over $r$. Therefore
\[
\lim_{n\to\infty}
\sup_{r\in\mathbb Z_+}
\P(
    E_{n,\mathrm{gs}}^{(\beta)}(p)
    -
    E_{n,r}^{(\beta, D_n)}(p)
    < c_p
)
=0.
\]
This proves the theorem.
\end{proof}

\subsection{Bounded-average-degree diluted regime}
\label{ssec:bounded-average-degree-diluted-model}
We now turn to the bounded-average-degree diluted regime with a variable
prefactor, where the Hamiltonian \(H_{n,p}^{(\zeta)}\) is defined in
\eqref{eq:zetaHamiltonian}. For fixed \(\zeta\), $m_{n,\zeta}
=
\frac{\zeta}{p!}n(1+o_n(1))$, so the retained hypergraph contains only linearly many retained interaction terms, while the limit \(\zeta\to\infty\) corresponds to increasing the average degree within this regime.
Throughout this section, \(p\ge2\) and
\(D\ge1\) are fixed, and all statements are uniform over an arbitrary number
\(r=r_n\ge0\) of ancillas.

The proof differs substantially from the growing-average-degree regime. We therefore begin with a brief overview of the argument.
When \(m_{n,\zeta} \gg n\), the residual can be compared directly to its
mean-field analogue by a global sampling argument. At bounded average degree, this is no longer possible. Instead, we exploit locality: by
Lemma~\ref{lem:residual-support-dependency-graph}, only retained hyperedges
meeting the dependency graph can contribute to the residual, while
Lemma~\ref{lem:zeta-good-event} shows that every bounded-degree dependency graph
intersects only \(O_{p,D}(n)\) retained hyperedges. This yields the variance
bound in Lemma~\ref{lem:zeta-residual-variance-dependency}. Together with the
corresponding local entropy estimate of
Lemma~\ref{lem:zeta-local-residual-entropy}, Dudley's entropy integral and the
Borell--TIS inequality give the residual estimate of
Lemma~\ref{lem:zeta-critical-residual-high-prob}. From there, the depth-to-product comparison, together with
Corollary~\ref{cor:sparse-mean-field-transfer-hp}, yields
Theorem~\ref{thm:main-critical-gap}.

\begin{lemma}[Hypergraph estimates]
\label{lem:zeta-good-event}
Fix \(p\ge2\), \(D\ge1\), and \(\zeta\ge1\). There exist constants
\(C_p<\infty\) and \(C_{p,D}<\infty\), depending only on the displayed subscripts and in particular independent of \(n\) and \(\zeta\), such that the following holds.

Let \(\mcE_p\subset\mcI_p^n\) be the retained random set of $p$-tuples, where we recall that each \(I\in\mcI_p^n\) is retained independently with probability
\(\zeta n^{-(p-1)}\). Define a graph \(Q\) on vertex set \([n]\) by declaring
two distinct vertices \(i,j\in[n]\) adjacent if they appear together in at
least one retained hyperedge. That is,
\[
    E(Q)
    :=
    \left\{
    \{i,j\}: i\neq j,\ \exists I\in\mcE_p
    \text{ such that } \{i,j\}\subset I
    \right\}.
\]
For any graph \(\mathsf G\) on \([n]\), define
\[
    \mcB(\mathsf G)
    :=
    \left\{
    I\in\mcE_p:
    \text{ there exists }\{i,j\}\subset I
    \text{ with } \{i,j\}\in E(\mathsf G)
    \right\}.
\]
Equivalently, \(\mcB(\mathsf G)\) is the set of retained hyperedges that are
not independent sets of \(\mathsf G\).

Then, with probability tending to one as \(n\to\infty\), the following
estimates hold simultaneously:
\[
    |\mcE_p|\le C_p\zeta n,
\]
\[
    |\{\mathsf G\subset Q:\Delta(\mathsf G)\le4^D\}|
    \le
    \exp\{C_{p,D}n\log(e+\zeta)\},
\]
and
\[
    \sup_{\mathsf G:\Delta(\mathsf G)\le4^D}
    |\mcB(\mathsf G)|
    \le
    C_{p,D}n.
\]
Moreover, for each fixed \(\zeta\), there exist constants
\(c_{p,D,\zeta}>0\) and \(n_0(p,D,\zeta)<\infty\) such that, for all
\(n\ge n_0(p,D,\zeta)\), the probability that at least one of the three estimates fails is at most \(\exp(-c_{p,D,\zeta}n)\).
\end{lemma}
\begin{proof}
First,
\[
    |\mcE_p|=\sum_{I\in\mcI_p^n}\xi_I,
    \qquad
    \E_\xi|\mcE_p|
    =
    \zeta n^{-(p-1)}\binom np
    =
    \frac{\zeta}{p!}n(1+o_n(1)).
\]
Choosing \(C_p\) sufficiently large, Bernstein's inequality gives
constants \(c_p>0\) and \(n_0(p)<\infty\) such that, for all
\(n\ge n_0(p)\),
\[
    \P_\xi\left(|\mcE_p|>C_p\zeta n\right)
    \le
    \exp(-c_p\zeta n).
\]
Let \(\mathsf E_{n,0}\) denote the event on which
\(|\mcE_p|\le C_p\zeta n\). On \(\mathsf E_{n,0}\), we have
\[
    |E(Q)|
    \le
    \binom p2|\mcE_p|
    \le
    C_p\zeta n,
\]
after increasing \(C_p\).

Let \(K:=4^D\). Then, deterministically,
\[
    |\{\mathsf G\subset Q:\Delta(\mathsf G)\le K\}|
    \le
    \prod_{v=1}^n
    \sum_{\ell=0}^K\binom{\deg_Q(v)}{\ell}
    \le
    \exp\left\{
        K\sum_{v=1}^n\log(e+\deg_Q(v))
    \right\}.
\]
By concavity of \(x\mapsto\log(e+x)\),
\[
    \sum_{v=1}^n\log(e+\deg_Q(v))
    \le
    n\log\left(
        e+\frac{2|E(Q)|}{n}
    \right).
\]
Hence, on \(\mathsf E_{n,0}\),
\[
    |\{\mathsf G\subset Q:\Delta(\mathsf G)\le4^D\}|
    \le
    \exp\{C_{p,D}n\log(e+\zeta)\}.
\]

It remains to prove the uniform bound on \(|\mcB(\mathsf G)|\). Fix a
graph \(\mathsf G\) with \(\Delta(\mathsf G)\le4^D\). Then $|E(\mathsf G)|
    \le
    2^{2D-1}n$. Let
\[
    \mcI_p(\mathsf G)
    :=
    \{I\in\mcI_p^n:
      I\text{ contains at least one edge of }\mathsf G\}.
\]
By a union bound over edges of \(\mathsf G\),
\[
    |\mcI_p(\mathsf G)|
    \le
    |E(\mathsf G)|\binom{n-2}{p-2}
    \le
    C_{p,D}n^{p-1}.
\]
It follows that $|\mcB(\mathsf G)|
    =
    \sum_{I\in\mcI_p(\mathsf G)}\xi_I$ is stochastically dominated by a binomial random variable with
parameters \(C_{p,D}n^{p-1}\) and
\(\zeta n^{-(p-1)}\). In particular, $\E|\mcB(\mathsf G)|
    \le
    C_{p,D}\zeta$.

We now make the constants in the union bound explicit. Let
\[
    \mathfrak G_{n,D}
    :=
    \left\{
        ([n],E):
        E\subseteq \binom{[n]}{2},
        \ \Delta([n],E)\le 4^D
    \right\}.
\]
Since every graph in \(\mathfrak G_{n,D}\) has at most
\(2^{2D-1}n\) edges, there exists \(B_D<\infty\) such that
\[
    |\mathfrak G_{n,D}|
    \le
    \sum_{q=0}^{\lfloor 2^{2D-1}n\rfloor}
    \binom{\binom n2}{q}
    \le
    \exp(B_Dn\log n)
\]
for all sufficiently large \(n\).

Set $A_{p,D}:=2(B_D+2)$. For a fixed graph \(\mathsf G\), the standard binomial tail bound gives
\[
\begin{aligned}
    \P_\xi\left(
        |\mcB(\mathsf G)|>A_{p,D}n
    \right)
    &\le
    \left(
        \frac{eB_{p,D}\zeta}
             {A_{p,D}n}
    \right)^{A_{p,D}n}.
\end{aligned}
\]
For each fixed \(\zeta\), there exists
\(n_0(p,D,\zeta)<\infty\) such that, whenever
\(n\ge n_0(p,D,\zeta)\),
\[
    \frac{eB_{p,D}\zeta}{A_{p,D}n}
    \le
    n^{-1/2}.
\]
Consequently,
\[
    \P_\xi\left(
        |\mcB(\mathsf G)|>A_{p,D}n
    \right)
    \le
    \exp\left(
        -\frac{A_{p,D}}2n\log n
    \right).
\]
Taking a union bound over all graphs of maximum degree at most \(4^D\)
gives
\[
    \P_\xi\left(
        \sup_{\mathsf G:\Delta(\mathsf G)\le4^D}
        |\mcB(\mathsf G)|
        >
        A_{p,D}n
    \right)
    \le
    \exp\left(
        B_Dn\log n
        -
        \frac{A_{p,D}}2n\log n
    \right)
    \le
    \exp(-2n\log n),
\]
where the last inequality follows from
\(A_{p,D}/2=B_D+2\). In particular, the probability tends to zero, and is at most \(\exp(-2n)\) for all sufficiently large \(n\).

Finally, on the intersection of \(\mathsf E_{n,0}\) and the event $ \{
        \sup_{\mathsf G:\Delta(\mathsf G)\le4^D}
        |\mcB(\mathsf G)|
        \le
        A_{p,D}n
    \}$, all three desired estimates hold simultaneously. Renaming
\(A_{p,D}\) as \(C_{p,D}\), and increasing constants if necessary,
completes the proof.

Indeed, for each fixed \(\zeta\), the probability that one of the three
estimates fails is at most
\[
    \exp(-c_p\zeta n)+\exp(-2n\log n)
    \le
    \exp(-c_{p,D,\zeta}n)
\]
for all sufficiently large \(n\), for some
\(c_{p,D,\zeta}>0\).
\end{proof}

\begin{lemma}[Variance bound, bounded-average-degree diluted regime]
\label{lem:zeta-residual-variance-dependency}
Fix \(p\ge2\), \(D\ge1\), and \(\zeta\ge1\). Then, for every
\(\varepsilon>0\),
\[
\limsup_{n\to\infty}
\sup_{r\in\mathbb Z_+}
\mathbb P_\xi\left(
\sup_{\theta\in\mathsf{T}_{n,r}^{(D)}}
\operatorname{Var}_g\!\left(R^{(\zeta)}(\theta)\mid \xi\right)
>
\frac{C_{p,D}}{\zeta}+\varepsilon
\right)
=0.
\]
\end{lemma}

\begin{proof}
Conditioned on \(\xi=(\xi_I, I\in \mcI_p^n) \),
\[
\tvar_g
(R^{(\zeta)}(\theta)\mid\xi )
=
\frac1{m_{n,\zeta}}
\sum_{I\in\mcI_p^n}
\xi_I
\sum_{a\in\{1,2,3\}^p}
|\chi_{I,a}(\theta)|^2 .
\]
For every \(I,a,\theta\), we have
\[
    |\chi_{I,a}(\theta)|
    \le
    |\langle\phi_{\theta}|P_I^a \otimes I_{\mathrm{anc}}|\phi_{\theta} \rangle|
    +
    \big |\prod_{s=1}^p m^{\phi_\theta}_{i_s}(a_s)\big |
    \le 2.
\]
Therefore
\[
    \sum_{a\in\{1,2,3\}^p}
    |\chi_{I,a}(\theta)|^2
    \le
    C_p .
\]
By
Lemma~\ref{lem:residual-support-dependency-graph}, the coefficient
\(\chi_{I,a}(\theta)\) can be nonzero only if $I\in \mcB(\mathsf G_D^{\mathrm{dep}})$, that is, only if \(I\) contains at least one edge of
\(\mathsf G_D^{\mathrm{dep}}\). Hence
\[
\tvar_g
(R^{(\zeta)}(\theta)\mid\xi )
\le
C_p
\frac{
|\mcB(\mathsf G_D^{\mathrm{dep}})|
}{
m_{n,\zeta}
}.
\]
By Lemma~\ref{lem:dependency-degree}, $\Delta(\mathsf G_D^{\mathrm{dep}}) \le 4^D$. Therefore Lemma~\ref{lem:zeta-good-event} implies that, with
probability tending to one,
\[
\sup_{\theta\in\mathsf{T}_{n,r}^{(D)}}
\tvar_g
(R^{(\zeta)}(\theta)\mid\xi )
\le
C_p
\sup_{\mathsf G:\Delta(\mathsf G)\le4^D}
\frac{|\mcB(\mathsf G)|}{m_{n,\zeta}}
\le
\frac{C_{p,D}}{\zeta}
+
o_n(1).
\]
This proves the claim.
\end{proof}

\begin{lemma}[Entropy bound, bounded-average-degree diluted regime]
\label{lem:zeta-local-residual-entropy}
Fix \(p\ge2\), \(D\ge1\), and \(\zeta\ge1\). With probability tending to one
over the retained hypergraph, there is an event on which the following holds:
for every \(r\in\mathbb Z_+\) and every \(\eta>0\),
\[
\log N(\mathsf{T}_{n,r}^{(D)},\mathsf d_\xi,\eta)
\le
C_{p,D} n
\left[
    \log(e+\zeta)
    +
    \log_+\left(\frac{C_{p,D}}{\eta}\right)
\right].
\]
\end{lemma}

\begin{proof}
We work on a realization of the retained hypergraph for which the three estimates
in Lemma~\ref{lem:zeta-good-event} hold. This occurs with probability tending
to one. 

For an architecture \(M_1,\dots,M_D\), let
\[
    \mathsf G_D^{\mathrm{dep}}
    =
    \mathsf G_D^{\mathrm{dep}}(M_1,\dots,M_D)
\]
be the associated dependency graph. We define the \(Q\)-restricted dependency graph associated with this
architecture by $\mathsf G_Q
    :=
    \mathsf G_D^{\mathrm{dep}}\cap Q$, where the intersection is taken at the level of edge sets, so that $E(\mathsf G_Q)
    =
    E(\mathsf G_D^{\mathrm{dep}})\cap E(Q)$. The set of retained hyperedges on which the residual coefficients may be nonzero depends only on this \(Q\)-restricted dependency graph. Indeed, if a retained hyperedge contains a
pair \(\{i,j\}\in E(\mathsf G_D^{\mathrm{dep}})\), then, by the definition of
\(Q\), the same pair also belongs to \(E(Q)\). Hence
\[
    \mcB(\mathsf G_D^{\mathrm{dep}})
    =
    \mcB(\mathsf G_D^{\mathrm{dep}}\cap Q)
    =
    \mcB(\mathsf G_Q).
\]
Moreover, by Lemma~\ref{lem:dependency-degree},
\(\Delta(\mathsf G_D^{\mathrm{dep}})\le4^D\), and therefore $\Delta(\mathsf G_Q)\le4^D$. 

Thus the possible \(Q\)-restricted dependency graphs are contained in the
class $\{\mathsf G\subset Q:\Delta(\mathsf G)\le4^D\}$.
By Lemma~\ref{lem:zeta-good-event}, the number of possible \(Q\)-restricted dependency graphs is at most $\exp\{C_{p,D}n\log(e+\zeta)\}$.

Fix one graph \(\mathsf G\subset Q\) with \(\Delta(\mathsf G)\le4^D\). Again by Lemma~\ref{lem:zeta-good-event}, $|\mcB(\mathsf G)|\le C_{p,D}n$. For each circuit parameter \(\theta\) whose \(Q\)-restricted dependency graph is \(\mathsf G\), define its active residual coefficient vector by
\[
    r_{\mathsf G}(\theta)
    :=
    \bigl(\chi_{I,a}(\theta)\bigr)_
    {I\in\mcB(\mathsf G),a\in\{1,2,3\}^p}.
\]
By Lemma~\ref{lem:residual-support-dependency-graph}, all inactive coefficients vanish, that is, if \(I\notin\mcB(\mathsf G)\), then
\(\chi_{I,a}(\theta)=0\) for every \(a\). Also \(|\chi_{I,a}(\theta)|\le2\), so
\[
    r_{\mathsf G}(\theta)\in[-2,2]^{k_{\mathsf G}},
    \qquad
    k_{\mathsf G}:=3^p|\mcB(\mathsf G)|\le C_{p,D}n .
\]

We cover this cube in the scaled Euclidean metric
\[
    d_{\mathsf G}(x,y)^2
    :=
    \frac1{m_{n,\zeta}}
    \sum_{I\in\mcB(\mathsf G)}
    \sum_{a\in\{1,2,3\}^p}
    |x_{I,a}-y_{I,a}|^2 .
\]
Since \(m_{n,\zeta} =\Theta( \zeta n) \), \(\zeta\ge1\), and
\(k_{\mathsf G}\le C_{p,D}n\), the standard Euclidean cube covering bound gives a grid of cardinality at most
\[
    \left(\frac{C_{p,D}}{\eta}\right)^{k_{\mathsf G}}
    \le
    \exp\left\{
        C_{p,D}n
        \log_+\left(\frac{C_{p,D}}{\eta}\right)
    \right\}.
\]
Every active residual coefficient vector is within \(d_{\mathsf G}\)-distance \(\eta\) of some grid point. Keeping only the grid cells that intersect the set of realizable residual
vectors, and choosing one circuit parameter from each such cell, gives a
\(2\eta\)-net for the corresponding circuit parameters in the canonical metric.

Recall that the conditional canonical metric of the Gaussian process
\(R^{(\zeta)}\), conditioned on \(\xi\), is
\[
    \mathsf{d}_\xi(\theta,\theta')^2
    :=
    \operatorname{Var}_g
    \left(
        R^{(\zeta)}(\theta)
        -
        R^{(\zeta)}(\theta')
        \big| \xi
    \right).
\]
Since the Gaussian coefficients \(g_{I,a}\) are independent, this equals
\[
    \mathsf{d}_\xi(\theta,\theta')^2
    =
    \frac1{m_{n,\zeta}}
    \sum_{I\in\mcI_p^n}
    \xi_I
    \sum_{a\in\{1,2,3\}^p}
    |\chi_{I,a}(\theta)-\chi_{I,a}(\theta')|^2 .
\]
Now suppose that \(\theta\) and \(\theta'\) have the same \(Q\)-restricted dependency graph \(\mathsf G\). By Lemma~\ref{lem:residual-support-dependency-graph},
all residual coefficients outside \(\mcB(\mathsf G)\) vanish for both
parameters. Therefore the preceding display reduces to
\[
    \mathsf{d}_\xi(\theta,\theta')^2
    =
    \frac1{m_{n,\zeta}}
    \sum_{I\in\mcB(\mathsf G)}
    \sum_{a\in\{1,2,3\}^p}
    |\chi_{I,a}(\theta)-\chi_{I,a}(\theta')|^2 .
\]
Thus, on a fixed  \(Q\)-restricted dependency graph \(\mathsf G\), the conditional canonical metric is exactly the scaled Euclidean
distance between the active residual coefficient vectors. Consequently the
preceding cube net gives an \(O(\eta)\)-cover for all circuit parameters with
fixed \(Q\)-restricted dependency graph \(\mathsf G\).

Finally, taking the union over all possible  \(Q\)-restricted dependency graphs gives
\[
\log N(\mathsf{T}_{n,r}^{(D)},\mathsf{d}_\xi,\eta)
\le
C_{p,D} n
\left[
    \log(e+\zeta)
    +
    \log_+\left(\frac{C_{p,D}}{\eta}\right)
\right],
\]
after adjusting constants. For \(\eta>1\), the same bound follows by monotonicity of covering numbers
from the case \(\eta=1\), after increasing \(C_{p,D}\).
\end{proof}

\begin{lemma}[High-probability bound, bounded-average-degree diluted regime]
\label{lem:zeta-critical-residual-high-prob}
Fix \(p\ge2\), \(D\ge1\), and \(\zeta\ge1\). Then, for every
\(\varepsilon>0\),
\[
\limsup_{n\to\infty}
\sup_{r\in\mathbb Z_+}
\mathbb P\left(
\frac1{\sqrt n}
\sup_{\theta\in\mathsf{T}_{n,r}^{(D)}}
|R^{(\zeta)}(\theta)|
>
C_{p,D}\sqrt{\frac{\log(e+\zeta)}{\zeta}}
+\varepsilon
\right)
=0.
\]
\end{lemma}

\begin{proof}
Let \(A_{p,D}\) be the constant from
Lemma~\ref{lem:zeta-residual-variance-dependency}. Applying that lemma with
slack \(\zeta^{-1}\), we obtain
\[
\limsup_{n\to\infty}
\sup_{r\in\mathbb Z_+}
\P_\xi\left(
\sup_{\theta\in\mathsf T_{n,r}^{(D)}}
\operatorname{Var}_g\!\left(R^{(\zeta)}(\theta)\mid \xi\right)
>
\frac{A_{p,D}+1}{\zeta}
\right)
=0.
\]
Thus, after replacing \(A_{p,D}+1\) by \(C_{p,D}\), the event
\[
\mathsf E_{n,r}^{\mathrm{var}}
:=
\left\{
\sup_{\theta\in\mathsf T_{n,r}^{(D)}}
\operatorname{Var}_g\!\left(R^{(\zeta)}(\theta)\mid \xi\right)
\le
\frac{C_{p,D}}{\zeta}
\right\}
\]
satisfies $\limsup_{n\to\infty}
\sup_{r\in\mathbb Z_+}
\P_\xi\bigl((\mathsf E_{n,r}^{\mathrm{var}})^c\bigr)=0$.

Let \(\mathsf E_n^{\mathrm{ent}}\) be the event from
Lemma~\ref{lem:zeta-local-residual-entropy}. On this event, for every
\(r\in\mathbb Z_+\) and every \(\eta>0\),
\[
\log N(\mathsf T_{n,r}^{(D)},\mathsf d_\xi,\eta)
\le
C_{p,D} n
\left[
    \log(e+\zeta)
    +
    \log_+\left(\frac{C_{p,D}}{\eta}\right)
\right],
\]
and \(\P_\xi(\mathsf E_n^{\mathrm{ent}})\to1\). Set $\mathsf E_{n,r}:=\mathsf E_n^{\mathrm{ent}}
    \cap \mathsf E_{n,r}^{\mathrm{var}}$. Then $\limsup_{n\to\infty}
\sup_{r\in\mathbb Z_+}
\P_\xi(\mathsf E_{n,r}^c)=0$.

Fix \(r\in\mathbb Z_+\), and condition on a realization
\(\xi\in\mathsf E_{n,r}\). Let \(\mathsf d_\xi\) be the conditional canonical
metric of \(R^{(\zeta)}\). On \(\mathsf E_{n,r}\),
\[
\operatorname{diam}(\mathsf T_{n,r}^{(D)},\mathsf d_\xi)
\le
2
\left(
\sup_{\theta\in\mathsf T_{n,r}^{(D)}}
\operatorname{Var}_g(R^{(\zeta)}(\theta)\mid \xi)
\right)^{1/2}
\le
C_{p,D}\zeta^{-1/2}.
\]

We first bound the conditional mean. The absolute value in the supremum is
handled by symmetrizing the index set, which only changes the entropy by an
additive constant. Dudley's entropy integral gives
\[
\E_g\big[
\sup_{\theta\in\mathsf T_{n,r}^{(D)}}
|R^{(\zeta)}(\theta)|
\,\big|\,\xi
\big]
\le
C
\int_0^{C_{p,D}\zeta^{-1/2}}
\sqrt{
C_{p,D} n
\left[
    \log(e+\zeta)
    +
    \log_+\left(\frac{C_{p,D}}{\eta}\right)
\right]
}
\,d\eta .
\]
Using
\[
    \int_0^\sigma
    \sqrt{A+\log_+(C/\eta)}\,d\eta
    \le
    C'\sigma\sqrt{A+\log(e+C/\sigma)+1},
\]
with \(\sigma=C_{p,D}\zeta^{-1/2}\) and \(A=\log(e+\zeta)\), we obtain
\[
\E_g\big[
\frac1{\sqrt n}
\sup_{\theta\in\mathsf T_{n,r}^{(D)}}
|R^{(\zeta)}(\theta)|
\,\big|\,\xi
\big]
\le
C_{p,D}
\sqrt{\frac{\log(e+\zeta)}{\zeta}}.
\]

Now define $Z_\theta:=\frac1{\sqrt n}R^{(\zeta)}(\theta)$. On \(\mathsf E_{n,r}\),
\[
\sup_{\theta\in\mathsf T_{n,r}^{(D)}}
\operatorname{Var}_g(Z_\theta\mid \xi)
\le
\frac{C_{p,D}}{\zeta n}.
\]
By Borell--TIS, conditionally on \(\xi\in\mathsf E_{n,r}\), for every
\(t>0\),
\[
\P_g\left(
\frac1{\sqrt n}
\sup_{\theta\in\mathsf T_{n,r}^{(D)}}
|R^{(\zeta)}(\theta)|
>
C_{p,D}
\sqrt{\frac{\log(e+\zeta)}{\zeta}}
+
t
\,\big|\,\xi
\right)
\le
2\exp\left(-c_{p,D}\zeta n t^2\right).
\]
Taking \(t=\varepsilon\), we get
\[
\P_g\left(
\frac1{\sqrt n}
\sup_{\theta\in\mathsf T_{n,r}^{(D)}}
|R^{(\zeta)}(\theta)|
>
C_{p,D}
\sqrt{\frac{\log(e+\zeta)}{\zeta}}
+
\varepsilon
\,\big|\,\xi
\right)
\le
2\exp\left(-c_{p,D}\zeta n\varepsilon^2\right).
\]

Therefore, for every \(r\in\mathbb Z_+\),
\[
\P_{\xi,g}\left(
\frac1{\sqrt n}
\sup_{\theta\in\mathsf T_{n,r}^{(D)}}
|R^{(\zeta)}(\theta)|
>
C_{p,D}
\sqrt{\frac{\log(e+\zeta)}{\zeta}}
+
\varepsilon
\right) 
\le
\P_\xi(\mathsf E_{n,r}^c)
+
2\exp\left(-c_{p,D}\zeta n\varepsilon^2\right).
\]
Taking the supremum over \(r\in\mathbb Z_+\) and then \(\limsup_{n\to\infty}\)
gives
\[
\limsup_{n\to\infty}
\sup_{r\in\mathbb Z_+}
\P_{\xi,g}\left(
\frac1{\sqrt n}
\sup_{\theta\in\mathsf T_{n,r}^{(D)}}
|R^{(\zeta)}(\theta)|
>
C_{p,D}
\sqrt{\frac{\log(e+\zeta)}{\zeta}}
+
\varepsilon
\right)
=0.
\]
This proves the claim.
\end{proof}

\begin{theorem}[Product-depth $D$ gap, bounded-average-degree diluted regime]
\label{thm:iterated-zeta-depth-product-comparison-hp}
Fix \(p\ge2\) and \(D\ge1\). Then, for every \(\varepsilon>0\),
\[
\lim_{\zeta\to\infty}
\limsup_{n\to\infty}
\sup_{ r \in \Z_+}
\P (
E_{n,r}^{(\zeta,D)}(p)
-
E_{n,\mathrm{prod}}^{(\zeta)}(p)
>
\varepsilon
)
=0.
\]
\end{theorem}

\begin{proof}
The lower bound $E_{n,r}^{(\zeta,D)}(p) \ge E_{n,\mathrm{prod}}^{(\zeta)}(p)$ holds because the depth-\(D\) circuit class contains all product states.

For the reverse bound, let \(\theta\in\mathsf{T}_{n,r}^{(D)}\), with
\(\phi_\theta=U|0^{n+r}\rangle\). By the residual decomposition,
\[
\langle\phi_\theta|H_{n,p}^{(\zeta)} \otimes I_{\mathrm{anc}}|\phi_\theta\rangle
=
X_0^{(\zeta)}(\theta)
+
R^{(\zeta)}(\theta).
\]
Taking suprema over \(\theta\in\mathsf{T}_{n,r}^{(D)}\) gives
\[
E_{n,r}^{(\zeta,D)}(p)
\le
\frac1{\sqrt n}
\sup_{\theta\in\mathsf{T}_{n,r}^{(D)}}X_0^{(\zeta)}(\theta)
+
\frac1{\sqrt n}
\sup_{\theta\in\mathsf{T}_{n,r}^{(D)}}|R^{(\zeta)}(\theta)|.
\]
Arguing exactly as in the mean-field case proof in Theorem~\ref{thm:log-depth-product-gap-residual}, the zeroth-order term is bounded by the product-state optimum, 
\[
\frac1{\sqrt n}
\sup_{\theta\in\mathsf{T}_{n,r}^{(D)}}X_0^{(\zeta)}(\theta)
\le
E_{n,\mathrm{prod}}^{(\zeta)}(p).
\]
Consequently,
\[
0\le
E_{n,r}^{(\zeta,D)}(p)
-
E_{n,\mathrm{prod}}^{(\zeta)}(p)
\le
\frac1{\sqrt n}
\sup_{\theta\in\mathsf{T}_{n,r}^{(D)}}|R^{(\zeta)}(\theta)|.
\]
By Lemma~\ref{lem:zeta-critical-residual-high-prob}, for every \(\gamma>0\),
\[
\limsup_{n\to\infty}
\sup_{r\in\mathbb Z_+}
\P\left(
\frac1{\sqrt n}
\sup_{\theta\in\mathsf T_{n,r}^{(D)}}
|R^{(\zeta)}(\theta)|
>
C_{p,D}\sqrt{\frac{\log(e+\zeta)}{\zeta}}
+\gamma
\right)
=0.
\]
Using
\[
0\le
E_{n,r}^{(\zeta,D)}(p)
-
E_{n,\mathrm{prod}}^{(\zeta)}(p)
\le
\frac1{\sqrt n}
\sup_{\theta\in\mathsf T_{n,r}^{(D)}}
|R^{(\zeta)}(\theta)|,
\]
we obtain, for every \(\gamma>0\),
\[
\limsup_{n\to\infty}
\sup_{r\in\mathbb Z_+}
\P\left(
E_{n,r}^{(\zeta,D)}(p)
-
E_{n,\mathrm{prod}}^{(\zeta)}(p)
>
C_{p,D}\sqrt{\frac{\log(e+\zeta)}{\zeta}}
+\gamma
\right)
=0.
\]
Now fix \(\varepsilon>0\). Choose \(\zeta\) sufficiently large that $C_{p,D}\sqrt{\frac{\log(e+\zeta)}{\zeta}}
< \frac{\varepsilon}{2}$.
Taking \(\gamma=\varepsilon/2\), we get
\[
\limsup_{n\to\infty}
\sup_{r\in\mathbb Z_+}
\P\left(
E_{n,r}^{(\zeta,D)}(p)
-
E_{n,\mathrm{prod}}^{(\zeta)}(p)
>
\varepsilon
\right)
=0
\]
for all sufficiently large \(\zeta\). Therefore
\[
\lim_{\zeta\to\infty}
\limsup_{n\to\infty}
\sup_{r\in\mathbb Z_+}
\P\left(
E_{n,r}^{(\zeta,D)}(p)
-
E_{n,\mathrm{prod}}^{(\zeta)}(p)
>
\varepsilon
\right)
=0.
\]
\end{proof}

\begin{proof}[Proof of Theorem~\ref{thm:main-critical-gap}]
Let \(c>0\) be the universal constant from Theorem~\ref{thm:benchmarks}. By
Corollary~\ref{cor:sparse-mean-field-transfer-hp}, for every \(\varepsilon>0\),
\[
\lim_{\zeta\to\infty}
\limsup_{n\to\infty}
\P\left(
    E_{n,\mathrm{gs}}^{(\zeta)}(p)
    <
    c\frac{3^{p/2}}{p}-\varepsilon
\right)
=0.
\]
Also, by the same corollary, for every \(\delta>0\), for all sufficiently
large fixed $p$, and every \(\varepsilon>0\),
\[
\lim_{\zeta\to\infty}
\limsup_{n\to\infty}
\P(
    E_{n,\mathrm{prod}}^{(\zeta)}(p)
    >
    (1+\delta)\sqrt{2\log p}+\varepsilon
)
=0.
\]

By Theorem~\ref{thm:iterated-zeta-depth-product-comparison-hp}, for every
fixed \(D\ge1\) and every \(\varepsilon>0\),
\[
\lim_{\zeta\to\infty}
\limsup_{n\to\infty}
\sup_{r\in\mathbb Z_+}
\P(
    E_{n,r}^{(\zeta,D)}(p)
    >
    E_{n,\mathrm{prod}}^{(\zeta)}(p)+\varepsilon)
=0.
\]
Combining the last two estimates by a union bound gives
\[
\lim_{\zeta\to\infty}
\limsup_{n\to\infty}
\sup_{r\in\mathbb Z_+}
\P(
    E_{n,r}^{(\zeta,D)}(p)
    >
    (1+\delta)\sqrt{2\log p}+2\varepsilon )
=0.
\]

Therefore, another union bound gives
\[
\lim_{\zeta\to\infty}
\limsup_{n\to\infty}
\sup_{r\in\mathbb Z_+}
\P\left(
    E_{n,\mathrm{gs}}^{(\zeta)}(p)
    -
    E_{n,r}^{(\zeta,D)}(p)
    <
    c\frac{3^{p/2}}{p}
    -
    (1+\delta)\sqrt{2\log p}
    -
    3\varepsilon
\right)
=0.
\]
Since $3^{p/2}/p\gg \sqrt{\log p}$ as $p\to\infty$, there exists \(p_0\ge3\) such that, for every fixed \(p\ge p_0\), we may
choose \(\delta>0\) and \(\varepsilon>0\) sufficiently small so that
\[
    c\frac{3^{p/2}}{p}
    -
    (1+\delta)\sqrt{2\log p}
    -
    3\varepsilon
    \ge
    2c_p
\]
for some \(c_p>0\). Reducing \(c_p\) if necessary, we obtain the result.
\end{proof}

\begin{appendix}

\section{Technical results}

\subsection{A Bessel-type estimate for localized vectors}
\label{ssec:AuxTechnical}

We now derive a Bessel-type estimate for localized vectors. The point of the
estimate is that a family of vectors may fail to be orthogonal, but if each
vector is supported only on a small subset of qubits and no qubit belongs to too
many of these supports, then the family still satisfies a uniform Bessel
inequality. This is the form needed in the proof of
Proposition~\ref{prop:full-residual-variance}, where the supports are backward
light-cones and the overlap multiplicity is controlled by
Lemma~\ref{lem:backward-support-multiplicity}.

\begin{lemma}[Bessel-type estimate]
\label{lem:bessel-type-ineq}
Let \(\Omega\) be a finite set of qubits, let \(J\) and \(J_0\) be finite index sets, and
let \(L_j\subseteq\Omega\), \(j\in J\). Suppose that for every \(v\in\Omega\), $|\{j\in J:v\in L_j\} |\le \gamma$. Let \(w_j^b\in(\C^2)^{\otimes \Omega}\), \(j\in J\), \(b\in J_0\), be
vectors satisfying
\[
w_j^b\in
(\C^2)^{\otimes L_j}
\otimes
|0\rangle^{\otimes(\Omega\setminus L_j)},
\qquad
\langle 0^\Omega,w_j^b\rangle=0,
\qquad
\|w_j^b\|\le 1.
\]
Then, for every \(x\in(\C^2)^{\otimes \Omega}\),
\[
\sum_{j\in J}\sum_{b\in J_0}
|\langle x,w_j^b\rangle|^2
\le
|J_0|\gamma\|x\|^2.
\]
\end{lemma}

\begin{proof}
Define the linear operator
\[
A:(\C^2)^{\otimes \Omega}\to \C^{|J||J_0|},
\qquad
Ax:=\bigl(\langle x,w_j^b\rangle\bigr)_{j\in J, ~b\in J_0}.
\]
Then the desired bound is
\[
\|Ax\|_{\ell^2}^{2}
\le
|J_0|\gamma\|x\|^{2}
\qquad
\text{for all }x\in(\C^2)^{\otimes \Omega}.
\]
Equivalently, \(\|A\|\le \sqrt{|J_0|\gamma}\). Since \(\|A\|=\|A^*\|\), where \(A^*\) is the adjoint of \(A\), it suffices to show that
\[
\|A^*c\|^{2}
\le
|J_0|\gamma\|c\|_{\ell^2}^{2}
\qquad
\text{for all }c=(c_{j,b})\in\C^{|J||J_0|}.
\]
Since
\[
A^*c
=
\sum_{j\in J}\sum_{b\in J_0} c_{j,b}w_j^b,
\]
it suffices to show
\[
\bigg\|
\sum_{j\in J}\sum_{b\in J_0} c_{j,b}w_j^b
\bigg\|^2
\le
|J_0|\gamma
\sum_{j\in J}\sum_{b\in J_0} |c_{j,b}|^2.
\]

To that end, for \(T\subseteq\Omega\), let \(e_T\in(\C^2)^{\otimes \Omega}\) denote the
computational basis vector
\[
e_T
:=
\bigotimes_{i\in\Omega}
\begin{cases}
|1\rangle, & i\in T,\\
|0\rangle, & i\notin T.
\end{cases}
\]
Thus \(e_\emptyset=|0^\Omega\rangle\) is the all-zero computational basis vector.

Since \(w_j^b\in(\C^2)^{\otimes L_j}
\otimes |0\rangle^{\otimes(\Omega\setminus L_j)}\), it is a linear combination
of \(e_T\)'s with \(T\subseteq L_j\). The additional condition
\(\langle 0^\Omega,w_j^b\rangle=0\) removes the \(T=\emptyset\) coefficient.
Hence
\[
w_j^b
=
\sum_{\emptyset\neq T\subseteq L_j}
\alpha_{j,b,T} e_T, \qquad \alpha_{j,b,T} := \langle e_T, w_j^b \rangle.
\]

Furthermore, by orthonormality and the norm assumption on \(w_j^b,\) we have
\[
\sum_{\emptyset\neq T\subseteq L_j}
|\alpha_{j,b,T}|^2
=
\|w_j^b\|^2
\le 1.
\]

Let \(c_{j,b}\in\C\).
Then
\[
\sum_{j\in J}\sum_{b\in J_0} c_{j,b}w_j^b
=
\sum_{\emptyset\neq T\subseteq\Omega}
\bigg(
\sum_{\substack{j\in J,\,b\in J_0\\ T\subseteq L_j}}
c_{j,b}\alpha_{j,b,T}
\bigg)e_T.
\]
Therefore
\[
\bigg \|
\sum_{j\in J}\sum_{b\in J_0} c_{j,b}w_j^b
\bigg \|^2
=
\sum_{\emptyset\neq T\subseteq\Omega}
\bigg |
\sum_{\substack{j\in J,\,b\in J_0\\ T\subseteq L_j}}
c_{j,b}\alpha_{j,b,T}
\bigg |^2.
\]
Fix a nonempty \(T\subseteq\Omega\), and choose any \(v\in T\). Define
\[
\mcA_T
:=
\{(j,b):j\in J,\ b\in J_0,\ T\subseteq L_j\}.
\]
If \(T\subseteq L_j\), then \(v\in L_j\). By assumption, there are at most
\(\gamma\) choices of \(j\in J\) such that \(v\in L_j\), and for each such \(j\) there
are \(|J_0|\) choices of \(b\). Hence \(|\mcA_T|\le |J_0|\gamma\). Therefore, by Cauchy--Schwarz,
\begin{align*}
\big |
\sum_{\substack{j\in J,\,b\in J_0\\ T\subseteq L_j}}
c_{j,b}\alpha_{j,b,T}
\big |^2
&=
\big |
\sum_{(j,b)\in\mcA_T}
c_{j,b}\alpha_{j,b,T}
\big |^2  \\
&\le
|\mcA_T|
\sum_{(j,b)\in\mcA_T}
|c_{j,b}\alpha_{j,b,T}|^2  \\
&\le
|J_0|\gamma
\sum_{\substack{j\in J,\,b\in J_0\\ T\subseteq L_j}}
|c_{j,b}|^2|\alpha_{j,b,T}|^2 .
\end{align*}
Summing this bound over all nonempty \(T\subseteq\Omega\), we obtain
\begin{align*}
\bigg\|
\sum_{j\in J}\sum_{b\in J_0} c_{j,b}w_j^b
\bigg\|^2
&\le
|J_0|\gamma
\sum_{\emptyset\neq T\subseteq\Omega}
\sum_{\substack{j\in J,\,b\in J_0\\ T\subseteq L_j}}
|c_{j,b}|^2|\alpha_{j,b,T}|^2  \\
&=
|J_0|\gamma
\sum_{j\in J}\sum_{b\in J_0} |c_{j,b}|^2
\sum_{\emptyset\neq T\subseteq L_j}
|\alpha_{j,b,T}|^2  \\
&\le
|J_0|\gamma
\sum_{j\in J}\sum_{b\in J_0} |c_{j,b}|^2,
\end{align*}
which completes the proof.
\end{proof}

\subsection{Universality transfer lemmas}
\label{app:universality-transfer}

In several places we need to transfer ground-state and product-state estimates from the mean-field Gaussian model to its diluted analogues. The main result we use for this purpose is \cite[Theorem~6]{anschuetz2025bounds}, which compares Gaussian disorder with more general disorder through a Lindeberg replacement argument.

There is one difference from the setting in which
\cite[Theorem~6]{anschuetz2025bounds} is stated. Our sparsification is
hyperedge-wise rather than coefficient-wise. For a fixed hyperedge \(I\), the
coefficients
\[
    \alpha_{I,a}
    :=
    \kappa^{-1/2}\xi_I g_{I,a},
    \qquad
    a\in\{1,2,3\}^p,
\]
share the same Bernoulli variable \(\xi_I\). Thus the disorder variables are
not independent coefficient by coefficient. They are, however, independent
block by block, with one block for each hyperedge \(I\), and each block has
the fixed dimension \(3^p\).

The proof of \cite[Theorem~6]{anschuetz2025bounds} extends to this setting by
replacing one block at a time and applying the same third-order Taylor
expansion. Matching the mean and covariance of each block cancels the linear
and quadratic terms. The free-energy derivative estimates used in
\cite[Theorem~6]{anschuetz2025bounds} also extend to block derivatives.
Indeed, perturbing the \(I\)-th disorder block in a direction
\(u\in\R^{3^p}\) produces the Hamiltonian perturbation
\(B_I(u):=\sum_a u_aP_I^a\). Since
\(\|B_I(u)\|_{\mathrm{op}}\le 3^{p/2}\|u\|_2\), the derivative estimates
in \cite[Lemma~26]{anschuetz2025bounds} remain valid, up to constants
depending only on \(p\).

\begin{lemma}[Blockwise Lindeberg replacement]
\label{lem:blockwise-lindeberg}
Let \(\mathcal I\) be a finite index set, and let
\[
    X=(X_I)_{I\in\mathcal I},
    \qquad
    Y=(Y_I)_{I\in\mathcal I}
\]
be arrays of random vectors in \(\R^d\). Assume that the pairs $\{(X_I,Y_I):I\in\mathcal I\}$ are independent across \(I\). For every \(I\in\mathcal I\), assume that
\[
    \E X_I=\E Y_I=0,
    \qquad
    \E[X_IX_I^\mathsf T]
    =
    \E[Y_IY_I^\mathsf T].
\]
Let \(F:(\R^d)^{\mathcal I}\to\R\) be three times continuously
differentiable, and suppose that
\[
    \sup_{x}
    \sup_{I\in\mathcal I}
    \sup_{\substack{u,v,w\in\R^d\\
                    \|u\|_2,\|v\|_2,\|w\|_2\le1}}
    \left|
        D_I^3F(x)[u,v,w]
    \right|
    \le L_3.
\]
Then
\[
    \left|
        \E F(X)-\E F(Y)
    \right|
    \le
    \frac{L_3}{6}
    \sum_{I\in\mathcal I}
    \left(
        \E\|X_I\|_2^3
        +
        \E\|Y_I\|_2^3
    \right).
\]
\end{lemma}

\begin{proof}
Enumerate
\(\mathcal I=\{I_1,\ldots,I_N\}\), and define the hybrid arrays
\[
    Z^{(k)}
    :=
    (X_{I_1},\ldots,X_{I_k},
      Y_{I_{k+1}},\ldots,Y_{I_N}),
    \qquad
    k=0,\ldots,N.
\]
Then
\[
    \E F(X)-\E F(Y)
    =
    \sum_{k=1}^N
    \E\left[
        F(Z^{(k)})-F(Z^{(k-1)})
    \right].
\]

Fix \(k\), and condition on all blocks except the \(I_k\)-th block.
Writing the resulting function of that block as \(f:\R^d\to\R\), Taylor expanding about the origin gives
\[
    f(z)
    =
    f(0)
    +
    Df(0)[z]
    +
    \frac12D^2f(0)[z,z]
    +
    \mathcal R_f(z),
\]
where
\[
    |\mathcal R_f(z)|
    \le
    \frac{L_3}{6}\|z\|_2^3.
\]
Because \(X_{I_k}\) and \(Y_{I_k}\) have the same mean and covariance,
the expectations of the linear and quadratic Taylor terms agree.
Consequently,
\[
    \left|
        \E\left[
            f(X_{I_k})-f(Y_{I_k})
        \right]
    \right|
    \le
    \frac{L_3}{6}
    \left(
        \E\|X_{I_k}\|_2^3
        +
        \E\|Y_{I_k}\|_2^3
    \right).
\]
Summing over \(k=1,\ldots,N\) proves the claim.
\end{proof}

\begin{lemma}[Mean-field-limit transfer for diluted optima]
\label{lem:mean-field-limit-transfer}
Fix \(p\ge2\).

\begin{enumerate}
\item[\textup{(i)}] \textup{Growing-average-degree diluted regime.}
Let \(\kappa_n=n^{-\beta}\), with \(0<\beta<p-1\). Then
\[
\big|
\E E_{n,\mathrm{gs}}^{(\beta)}(p)
-
\E E_{n,\mathrm{gs}}(p)
\big|
\to0 .
\]
The same conclusion holds for the product-state optima:
\[
\big|
\E E_{n,\mathrm{prod}}^{(\beta)}(p)
-
\E E_{n,\mathrm{prod}}(p)
\big|
\to0 .
\]

\item[\textup{(ii)}] \textup{Bounded-average-degree diluted regime.}
Let \(\kappa_{n,\zeta}=\zeta n^{-(p-1)}\). Then
\[
\lim_{\zeta\to\infty}
\limsup_{n\to\infty}
\big|
\E E_{n,\mathrm{gs}}^{(\zeta)}(p)
-
\E E_{n,\mathrm{gs}}(p)
\big|
=0 .
\]
The same conclusion holds for the product-state optima:
\[
\lim_{\zeta\to\infty}
\limsup_{n\to\infty}
\big|
\E E_{n,\mathrm{prod}}^{(\zeta)}(p)
-
\E E_{n,\mathrm{prod}}(p)
\big|
=0 .
\]
\end{enumerate}
\end{lemma}

\begin{proof}
Write \(N_p:=\binom np\). In either diluted regime, let \(\kappa\) denote
the corresponding sampling probability and let \(m=\kappa N_p\). We use
\(H_{n,p}^{(\kappa)}\) as shorthand for \(H_{n,p}^{(\beta)}\) or
\(H_{n,p}^{(\zeta)}\), according to the regime. Using the mean-field
normalization, we may write
\[
    H_{n,p}^{(\kappa)}
    =
    N_p^{-1/2}
    \sum_{I\in\mcI_p^n}
    \sum_{a\in\{1,2,3\}^p}
    \alpha_{I,a}P_I^a,
    \qquad
    \alpha_{I,a}:=\kappa^{-1/2}\xi_I g_{I,a}.
\]

For each \(I\in\mcI_p^n\), define the disorder blocks
\[
    \alpha_I
    :=
    \bigl(\alpha_{I,a}\bigr)_{a\in\{1,2,3\}^p},
    \qquad
    G_I
    :=
    \bigl(g_{I,a}\bigr)_{a\in\{1,2,3\}^p}.
\]
The blocks \(\{\alpha_I\}_{I\in\mcI_p^n}\) are independent. Moreover,
\(\E\alpha_I=0\), and, for \(a,b\in\{1,2,3\}^p\),
\[
    \E[\alpha_{I,a}\alpha_{I,b}]
    =
    \kappa^{-1}\E[\xi_Ig_{I,a}g_{I,b}]
    =
    \delta_{a,b}.
\]
Thus \(\alpha_I\) and \(G_I\) have the same mean and covariance matrix.

We first consider the unrestricted optimum. For a block array
\(x=(x_I)_{I\in\mcI_p^n}\), with
\(x_I=(x_{I,a})_{a\in\{1,2,3\}^p}\), set
\[
    H_n(x)
    :=
    N_p^{-1/2}
    \sum_{I\in\mcI_p^n}
    \sum_{a\in\{1,2,3\}^p}
    x_{I,a}P_I^a.
\]
For \(t\ge1\), define
\[
    \Phi_t(x)
    :=
    \frac1t
    \log\operatorname{Tr}
    \exp\left(
        \frac{t}{\sqrt n}H_n(x)
    \right).
\]
The standard free-energy approximation
\cite[Eq.~(150)--(151)]{anschuetz2025bounds} gives
\[
    \frac1{\sqrt n}\lambda_{\max}(H_n(x))
    \le
    \Phi_t(x)
    \le
    \frac1{\sqrt n}\lambda_{\max}(H_n(x))
    +
    \frac{n\log2}{t}.
\]

We next verify the derivative estimate needed in the block replacement.
For \(I\in\mcI_p^n\), \(u\in\R^{3^p}\), and \(s\in\R\), define the
perturbed block array \(x^{(I,s,u)}\) by
\[
    x_J^{(I,s,u)}
    :=
    \begin{cases}
        x_I+su, & J=I,\\
        x_J, & J\neq I.
    \end{cases}
\]
Then
\[
    H_n\bigl(x^{(I,s,u)}\bigr)-H_n(x)
    =
    sN_p^{-1/2}B_I(u),
    \qquad
    B_I(u):=\sum_{a\in\{1,2,3\}^p}u_aP_I^a.
\]
Since every \(P_I^a\) is unitary,
\[
    \|B_I(u)\|_{\mathrm{op}}
    \le
    \sum_a|u_a|
    \le
    3^{p/2}\|u\|_2.
\]
For fixed \(x\) and \(I\), the map $(u,v,w)\longmapsto D_I^3\Phi_t(x)[u,v,w]$ is a symmetric trilinear form. The derivative argument of
\cite[Lemma~26]{anschuetz2025bounds} bounds its diagonal values
\(D_I^3\Phi_t(x)[u,u,u]\). The polarization identity for symmetric
trilinear forms therefore yields
\[
    \sup_{\substack{u,v,w\in\R^{3^p}\\
                    \|u\|_2,\|v\|_2,\|w\|_2\le1}}
    \left|
        D_I^3\Phi_t(x)[u,v,w]
    \right|
    \le
    C_p\frac{t^2}{(nN_p)^{3/2}}.
\]
The constant \(C_p\) absorbs the fixed block dimension \(3^p\).

We may therefore apply
Lemma~\ref{lem:blockwise-lindeberg} to the independent block arrays
\((\alpha_I)_{I\in\mcI_p^n}\) and
\((G_I)_{I\in\mcI_p^n}\). Since
\[
    \E\|\alpha_I\|_2^3
    =
    \kappa^{-1/2}\E\|G_I\|_2^3
    \le
    C_p\kappa^{-1/2},
\]
and \(\E\|G_I\|_2^3\le C_p\), we obtain
\[
    | \E\Phi_t(\alpha)-\E\Phi_t(G) |
    \le
    C_pN_p
    \frac{t^2}{(nN_p)^{3/2}}
    \left(1+\kappa^{-1/2}\right)
    \le
    C_p\frac{t^2}
    {n^{3/2}N_p^{1/2}\kappa^{1/2}},
\]
where we used \(\kappa\le1\). Taking \(t=An\) and using the free-energy
approximation on both sides gives
\begin{equation}
\label{eq:ground-state-expectation-transfer-bound}
\left|
    \E E_{n,\mathrm{gs}}^{(\kappa)}(p)
    -
    \E E_{n,\mathrm{gs}}(p)
\right|
\le
\frac{2\log2}{A}
+
C_pA^2
\sqrt{\frac{n}{\kappa N_p}}.
\end{equation}

We now turn to the product-state optimum. An additional discretization
is needed because the matrix free energy approximates the maximum over
all states rather than the maximum over product states.

Fix \(0<\eta<1/2\). There exists a finite set
\(\mathcal Q_\eta\subset\mcS^2\) such that
\[
    |\mathcal Q_\eta|
    \le
    \left(\frac C\eta\right)^2,
    \qquad
    (1-\eta)\B_2^3
    \subseteq
    \operatorname{conv}(\mathcal Q_\eta).
\]
Let \(\mathsf S_{n,\eta}\subset\mathsf S_{n,\mathrm{prod}}\) be the set
of product states whose one-site Bloch vectors belong to
\(\mathcal Q_\eta\). Then
\(\log|\mathsf S_{n,\eta}|\le Cn\log(C/\eta)\).

For a Hamiltonian \(H\) of the form considered here, write
\[
    \mathcal E_{\mathrm{prod}}(H)
    :=
    \frac1{\sqrt n}
    \sup_{\psi\in\mathsf S_{n,\mathrm{prod}}}
    \langle\psi|H|\psi\rangle,
    \qquad
    \mathcal E_\eta(H)
    :=
    \frac1{\sqrt n}
    \max_{\psi\in\mathsf S_{n,\eta}}
    \langle\psi|H|\psi\rangle.
\]
Since the product-state energy is multilinear in the one-site Bloch
vectors,
\[
    (1-\eta)^p\mathcal E_{\mathrm{prod}}(H)
    \le
    \mathcal E_\eta(H)
    \le
    \mathcal E_{\mathrm{prod}}(H).
\]
To see the first inequality, let \(x_1,\ldots,x_n\in\mcS^2\) be
arbitrary. For every \(i\), choose an independent random vector
\(Y_i\in\mathcal Q_\eta\) satisfying
\(\E Y_i=(1-\eta)x_i\). Since every interaction contains \(p\) distinct
sites, multilinearity gives
\[
    \E X_0(Y_1,\ldots,Y_n)
    =
    (1-\eta)^pX_0(x_1,\ldots,x_n).
\]
Hence some deterministic realization of
\((Y_1,\ldots,Y_n)\) attains at least the right-hand side.

We further use the uniform estimate
\(\E[n^{-1/2}\|H\|_{\mathrm{op}}]\le C_p\) for the mean-field and
diluted Hamiltonians. For the diluted model, conditionally on the
retained hypergraph, the Gaussian matrix-series inequality
\cite[Theorem~1.2]{tropp2012User} gives
\[
    \E_g\|H_{n,p}^{(\kappa)}\|_{\mathrm{op}}
    \le
    C_p\sqrt n
    \left(\frac{|\mcE_p|}{m}\right)^{1/2}.
\]
Taking expectation over \(\xi\), using
\(\E_\xi|\mcE_p|=m\), and applying Jensen's inequality proves the claim.
It follows that
\[
    \E\left|
        \mathcal E_{\mathrm{prod}}(H)-\mathcal E_\eta(H)
    \right|
    \le
    C_p\eta.
\]

For a block array \(x\), define $h_\psi(x)
    :=
    \frac1{\sqrt n}
    \langle\psi|H_n(x)|\psi\rangle$ and the soft maximum
\[
    \Phi_{t,\eta}(x)
    :=
    \frac1t
    \log
    \sum_{\psi\in\mathsf S_{n,\eta}}
    \exp\bigl(t h_\psi(x)\bigr).
\]
Then
\[
    \max_{\psi\in\mathsf S_{n,\eta}}h_\psi(x)
    \le
    \Phi_{t,\eta}(x)
    \le
    \max_{\psi\in\mathsf S_{n,\eta}}h_\psi(x)
    +
    \frac{\log|\mathsf S_{n,\eta}|}{t}.
\]

For \(I\in\mcI_p^n\), \(u\in\R^{3^p}\), and \(s\in\R\), replacing \(x\)
by \(x^{(I,s,u)}\) changes \(h_\psi(x)\) by $\frac{s}{\sqrt{nN_p}}
    \langle\psi|B_I(u)|\psi\rangle$. Differentiating the log-sum-exp expression shows that its third
derivative is \(t^2\) times the corresponding centered third moment.
Since
\[
    \bigg|
        \frac{1}{\sqrt{nN_p}}
        \langle\psi|B_I(u)|\psi\rangle
    \bigg|
    \le
    \frac{3^{p/2}}{\sqrt{nN_p}}\|u\|_2,
\]
polarization again gives
\[
    \sup_x
    \sup_{I\in\mcI_p^n}
    \sup_{\substack{u,v,w\in\R^{3^p}\\
                    \|u\|_2,\|v\|_2,\|w\|_2\le1}}
    \left|
        D_I^3\Phi_{t,\eta}(x)[u,v,w]
    \right|
    \le
    C_p\frac{t^2}{(nN_p)^{3/2}}.
\]
Applying Lemma~\ref{lem:blockwise-lindeberg} and taking \(t=An\)
therefore gives
\begin{align}
\label{eq:product-expectation-transfer-bound}
|
    \E E_{n,\mathrm{prod}}^{(\kappa)}(p)
    -
    \E E_{n,\mathrm{prod}}(p)
|
&\le
C_p\eta
+
\frac{C\log(C/\eta)}{A}
+
C_pA^2
\sqrt{\frac{n}{\kappa N_p}}.
\end{align}

Suppose first that \(\kappa=n^{-\beta}\), with
\(0<\beta<p-1\). Then $\kappa N_p/n
    =
    \Theta_p(n^{p-1-\beta})
    \to\infty$. Taking \(n\to\infty\) in
\eqref{eq:ground-state-expectation-transfer-bound} and then letting
\(A\to\infty\) proves the ground-state statement. Applying the same
limits to \eqref{eq:product-expectation-transfer-bound}, and then
letting \(\eta\downarrow0\), proves the product-state statement.

In the bounded-average-degree regime,
\(\kappa=\zeta n^{-(p-1)}\), and $\kappa N_p/n
    =
    \frac{\zeta}{p!}(1+o_n(1))$. Consequently, for every fixed \(A\),
\[
    \lim_{\zeta\to\infty}
    \limsup_{n\to\infty}
    A^2\sqrt{\frac{n}{\kappa N_p}}
    =0.
\]
The corresponding iterated-limit conclusions follow from
\eqref{eq:ground-state-expectation-transfer-bound} by first taking
\(n\to\infty\), then \(\zeta\to\infty\), and finally \(A\to\infty\).
For the product-state optimum, apply the same limits to
\eqref{eq:product-expectation-transfer-bound} and then let
\(\eta\downarrow0\).
\end{proof}

To pass from expectation transfer to high-probability transfer, we use
the following concentration estimate.

\begin{lemma}[Concentration of diluted optima]
\label{lem:diluted-optima-concentration}
There exists \(C_p<\infty\) such that, in either diluted regime,
\[
    \operatorname{Var}\left(
        E_{n,\mathrm{gs}}^{(\cdot)}(p)
    \right)
    \le
    \frac{C_p}{n},
    \qquad
    \operatorname{Var}\left(
        E_{n,\mathrm{prod}}^{(\cdot)}(p)
    \right)
    \le
    \frac{C_p}{n}.
\]
The same estimates hold for the mean-field optima.
\end{lemma}

\begin{proof}
Write \(N_p:=\binom np\) and \(m:=\kappa N_p\). For each
\(I\in\mcI_p^n\), define the independent disorder block
\[
    W_I
    :=
    \xi_I
    \bigl(g_{I,a}\bigr)_{a\in\{1,2,3\}^p}
    \in\R^{3^p},
\]
and let \(W:=(W_I)_{I\in\mcI_p^n}\). In this notation, the diluted
Hamiltonian is
\[
    H(W)
    :=
    \frac1{\sqrt m}
    \sum_{I\in\mcI_p^n}
    \sum_{a\in\{1,2,3\}^p}
    W_{I,a}P_I^a.
\]

Define
\[
    F_{\mathrm{gs}}(W)
    :=
    \frac1{\sqrt n}\lambda_{\max}(H(W)),
    \quad
    F_{\mathrm{prod}}(W)
    :=
    \frac1{\sqrt n}
    \sup_{\psi\in\mathsf S_{n,\mathrm{prod}}}
    \langle\psi|H(W)|\psi\rangle.
\]
We prove the variance bound for either of these functions, denoted
generically by \(F(W)\).

For each \(I\in\mcI_p^n\), let \(W_I'\) be an independent copy of
\(W_I\), and let \(W^{(I)}\) be the block array obtained from \(W\) by
replacing \(W_I\) with \(W_I'\), while leaving all other blocks unchanged.
Since the blocks \(\{W_I\}\) are independent, the Efron--Stein inequality
gives
\[
    \operatorname{Var}(F(W))
    \le
    \frac12
    \sum_{I\in\mcI_p^n}
    \E[
        (F(W)-F(W^{(I)}))^2
    ].
\]

We next bound the effect of replacing one block. For the ground-state
optimum, the Rayleigh--Ritz formula gives
\[
    |
    F_{\mathrm{gs}}(W)-F_{\mathrm{gs}}(W^{(I)})
    |
    \le
    \frac1{\sqrt n}
    \|
        H(W)-H(W^{(I)})
    \|_{\mathrm{op}}.
\]
For the product-state optimum, the same bound follows directly by taking
the supremum over product states. Since
\[
    H(W)-H(W^{(I)})
    =
    \frac1{\sqrt m}
    \sum_{a\in\{1,2,3\}^p}
    (W_{I,a}-W_{I,a}')P_I^a,
\]
we obtain, for either choice of \(F\),
\[
    |F(W)-F(W^{(I)})|
    \le
    \frac1{\sqrt{nm}}
    \sum_{a\in\{1,2,3\}^p}
    |W_{I,a}-W_{I,a}'|
    \le
    \frac{3^{p/2}}{\sqrt{nm}}
    \|W_I-W_I'\|_2.
\]

Moreover, \(W_I\) is centered and
\[
    \E\|W_I\|_2^2
    =
    \E\xi_I
    \sum_{a\in\{1,2,3\}^p}
    \E g_{I,a}^2
    =
    \kappa 3^p.
\]
Therefore $\E\|W_I-W_I'\|_2^2
    =
    2\kappa 3^p$. Substituting into the Efron--Stein bound and using \(m=\kappa N_p\), we obtain
\[
    \operatorname{Var}(F(W))
    \le
    \frac12 N_p
    \frac{3^p}{nm}
    \bigl(2\kappa 3^p\bigr) 
    =
    \frac{3^{2p}}{n}.
\]
This proves both diluted-model variance bounds. The mean-field case is the same argument with \(\xi_I\equiv1\), \(\kappa=1\), and \(m=N_p\).
\end{proof}

\begin{corollary}[High-probability transfer of mean-field benchmark bounds]
\label{cor:sparse-mean-field-transfer-hp}
There exist a universal constant \(c>0\) and a deterministic sequence
\(\varepsilon_p\to0\) such that the following hold.

\begin{enumerate}
\item[\textup{(i)}] \textup{Growing-average-degree diluted regime.}
For every fixed \(p\ge3\) and every \(0<\beta<p-1\), with high
probability as \(n\to\infty\),
\begin{align*}
    c\frac{3^{p/2}}{p}
    \le
    E_{n,\mathrm{gs}}^{(\beta)}(p)
    &\le
    3^{p/2}\sqrt{2\log p}, \\
    E_{n,\mathrm{prod}}^{(\beta)}(p)
    &\le
    (1+\varepsilon_p)\sqrt{2\log p}.
\end{align*}

\item[\textup{(ii)}] \textup{Bounded-average-degree diluted regime.}
For every fixed \(p\ge3\), there exists
\(\zeta_0=\zeta_0(p)<\infty\) such that, for every fixed
\(\zeta\ge\zeta_0\), with high probability as \(n\to\infty\),
\begin{align*}
    c\frac{3^{p/2}}{p}
    \le
    E_{n,\mathrm{gs}}^{(\zeta)}(p)
    &\le
    3^{p/2}\sqrt{2\log p}, \\
    E_{n,\mathrm{prod}}^{(\zeta)}(p)
    &\le
    (1+\varepsilon_p)\sqrt{2\log p}.
\end{align*}
\end{enumerate}
\end{corollary}

\begin{proof}
Let \(c_0>0\) be a universal constant such that, for every fixed
\(p\ge3\),
\[
    \P\left(
        E_{n,\mathrm{gs}}(p)
        \ge
        c_0\frac{3^{p/2}}{p}
    \right)
    \to 1.
\]
Also, let \(\eta_p\ge0\) be a deterministic sequence with
\(\eta_p\to0\) such that
\[
    \P(
        E_{n,\mathrm{prod}}(p)
        \le
        (1+\eta_p)\sqrt{2\log p}
    )
    \to 1.
\]
These are the mean-field benchmark bounds from
\cite{anschuetz2025bounds,anschuetz2025strongly}. By enlarging finitely
many terms of \(\eta_p\), if necessary, we may take the second statement
to hold for every \(p\ge3\).

Set
\[
    c:=\frac{c_0}{2},
    \qquad
    A_p:=c\frac{3^{p/2}}{p},
    \qquad
    B_p:=(1+\eta_p)\sqrt{2\log p}.
\]
Choose any deterministic sequence \(\rho_p>0\) such that
\(\rho_p\to0\), and set $\varepsilon_p:=\eta_p+\rho_p$. Thus the mean-field ground-state lower bound is \(2A_p\).

We first record a consequence of the mean-field concentration estimate.
If a sequence \(Z_n\) satisfies
\(\operatorname{Var}(Z_n)\to0\) and
\(\P(Z_n\ge a)\to1\), then
\(\liminf_n\E Z_n\ge a\); otherwise, there would exist \(\rho>0\) and a
subsequence along which
\(\E Z_n\le a-\rho\), and Chebyshev's inequality would give
\[
    \P(Z_n\ge a)
    \le
    \frac{\operatorname{Var}(Z_n)}{\rho^2}
    \to0,
\]
which would be a contradiction. Similarly, if
\(\P(Z_n\le b)\to1\), then
\(\limsup_n\E Z_n\le b\).

Applying this observation and
Lemma~\ref{lem:diluted-optima-concentration} to the mean-field optima
gives
\begin{equation}
    \liminf_{n\to\infty}
    \E E_{n,\mathrm{gs}}(p)
    \ge
    2A_p,
    \qquad
    \limsup_{n\to\infty}
    \E E_{n,\mathrm{prod}}(p)
    \le
    B_p.
\label{eq:mean-field-expectation-benchmarks}
\end{equation}

We now prove part \textup{(i)}. Suppose that
\(\kappa_n=n^{-\beta}\), where \(0<\beta<p-1\).
Lemma~\ref{lem:mean-field-limit-transfer} and
\eqref{eq:mean-field-expectation-benchmarks} imply $\liminf_{n\to\infty}
    \E E_{n,\mathrm{gs}}^{(\beta)}(p)
    \ge
    2A_p$. Consequently, for every \(\varepsilon>0\), the distance between
\(\E E_{n,\mathrm{gs}}^{(\beta)}(p)\) and
\(A_p-\varepsilon\) is bounded below by a positive constant for all
sufficiently large \(n\). Hence, by
Lemma~\ref{lem:diluted-optima-concentration}, $\P(
        E_{n,\mathrm{gs}}^{(\beta)}(p)
        <
        A_p)
    \to0$.
    
For the product-state bound,
Lemma~\ref{lem:mean-field-limit-transfer} and
\eqref{eq:mean-field-expectation-benchmarks} imply
\[
    \limsup_{n\to\infty}
    \E E_{n,\mathrm{prod}}^{(\beta)}(p)
    \le
    (1+\eta_p)\sqrt{2\log p}.
\]
Since \(\varepsilon_p-\eta_p=\rho_p>0\), the threshold
\((1+\varepsilon_p)\sqrt{2\log p}\) lies a positive distance above the
diluted expectation for all sufficiently large \(n\). Therefore,
Lemma~\ref{lem:diluted-optima-concentration} gives
\[
    \P(
        E_{n,\mathrm{prod}}^{(\beta)}(p)
        >
        (1+\varepsilon_p)\sqrt{2\log p}
    )
    \to0.
\]
This proves part \textup{(i)}.

We next prove part \textup{(ii)}. Fix \(p\ge3\). By
Lemma~\ref{lem:mean-field-limit-transfer},
\[
    \lim_{\zeta\to\infty}
    \limsup_{n\to\infty}
    |
        \E E_{n,\mathrm{gs}}^{(\zeta)}(p)
        -
        \E E_{n,\mathrm{gs}}(p)
    |
    =0.
\]
Hence there exists
\(\zeta_{\mathrm{gs}}(p)<\infty\) such that, for every
\(\zeta\ge\zeta_{\mathrm{gs}}(p)\),
\[
    \limsup_{n\to\infty}
    |
        \E E_{n,\mathrm{gs}}^{(\zeta)}(p)
        -
        \E E_{n,\mathrm{gs}}(p)
    |
    \le
    \frac{A_p}{2}.
\]
Together with \eqref{eq:mean-field-expectation-benchmarks}, this gives
\[
    \liminf_{n\to\infty}
    \E E_{n,\mathrm{gs}}^{(\zeta)}(p)
    \ge
    \frac{3A_p}{2}
\]
for every fixed \(\zeta\ge\zeta_{\mathrm{gs}}(p)\). Therefore
Lemma~\ref{lem:diluted-optima-concentration} implies $\P\left(
        E_{n,\mathrm{gs}}^{(\zeta)}(p)
        <
        A_p
    \right)
    \to0$ for every such fixed \(\zeta\).

For the product-state bound, recall that \(\varepsilon_p=\eta_p+\rho_p\).
Again by Lemma~\ref{lem:mean-field-limit-transfer}, there exists
\(\zeta_{\mathrm{prod}}(p)<\infty\) such that, for every
\(\zeta\ge\zeta_{\mathrm{prod}}(p)\),
\[
    \limsup_{n\to\infty}
    |
        \E E_{n,\mathrm{prod}}^{(\zeta)}(p)
        -
        \E E_{n,\mathrm{prod}}(p)
    |
    \le
    \frac{\rho_p}{4}\sqrt{2\log p}.
\]
Using \eqref{eq:mean-field-expectation-benchmarks}, we obtain
\[
    \limsup_{n\to\infty}
    \E E_{n,\mathrm{prod}}^{(\zeta)}(p)
    \le
    \left(
        1+\eta_p+\frac{\rho_p}{4}
    \right)
    \sqrt{2\log p}.
\]
Since $\eta_p+\rho_p/4
    <
    \varepsilon_p$, the threshold
\((1+\varepsilon_p)\sqrt{2\log p}\) lies a positive distance above the
expectation. Thus
Lemma~\ref{lem:diluted-optima-concentration} gives
\[
    \P\left(
        E_{n,\mathrm{prod}}^{(\zeta)}(p)
        >
        (1+\varepsilon_p)\sqrt{2\log p}
    \right)
    \to0
\]
for every fixed
\(\zeta\ge\zeta_{\mathrm{prod}}(p)\).

Set $\zeta_0(p)
    :=
    \max\{
        \zeta_{\mathrm{gs}}(p),
        \zeta_{\mathrm{prod}}(p)
    \}$.
Then, for every fixed \(\zeta\ge\zeta_0(p)\), both claimed bounds hold
with high probability as \(n\to\infty\). Since
\(\eta_p\to0\) and \(\rho_p\to0\), we also have
\(\varepsilon_p\to0\).

It remains to prove the ground-state upper bounds. Conditionally on the
retained hypergraph, the Gaussian matrix-series
inequality~\cite[Theorem~1.2]{tropp2012User} gives
\[
    \E_g\|H_{n,p}^{(\zeta)}\|_{\mathrm{op}}
    \le
    \sqrt{2(n+1)\log2}
    \left(
        \frac{3^p|\mcE_p|}{m_{n,\zeta}}
    \right)^{1/2}.
\]
After dividing by \(\sqrt n\), taking expectation over the retained
hypergraph, and using Jensen's inequality and
\(\E|\mcE_p|=m_{n,\zeta}\), we obtain
\[
    \E E_{n,\mathrm{gs}}^{(\zeta)}(p)
    \le
    3^{p/2}
    \sqrt{2\left(1+\frac1n\right)\log2}.
\]
For \(p\ge3\), the right-hand side is strictly smaller than
\(3^{p/2}\sqrt{2\log p}\) for all sufficiently large \(n\). The
concentration estimate therefore gives
\[
    \P(
        E_{n,\mathrm{gs}}^{(\zeta)}(p)
        >
        3^{p/2}\sqrt{2\log p}
    )
    \to0.
\]
\end{proof}

\section{Background and additional context}
\label{sec:background}

This section collects background material for readers who may be less familiar with quantum spin glasses or with the NLTS literature. In Section~\ref{ssec:classical-quantum-p-spin-background}, we recall the classical \(p\)-spin model and explain how the quantum \(p\)-spin Hamiltonian studied in this paper can be viewed as its noncommutative analogue. In Section~\ref{ssec:ABN}, we provide additional context on the NLTS construction of Anshu, Breuckmann, and Nirkhe~\cite{anshu2023nlts}, emphasizing the ways in which
their code-based Hamiltonian differs from the random spin-glass Hamiltonians considered here. The material in this section is not needed for the proofs, but it is intended to situate our results within the broader landscape of spin-glass models, quantum Hamiltonian complexity, and state-preparation lower
bounds.

\subsection{Classical and quantum \(p\)-spin models}
\label{ssec:classical-quantum-p-spin-background}
The classical \(p\)-spin model is a central mean-field model in spin glass theory \cite{MezardMontanariBook, TalagrandBook, panchenko2013Sherrington}. It assigns a random energy to each spin configuration \(\eta\in\{\pm1\}^n\), where the randomness comes from Gaussian couplings among \(p\)-tuples of spins. Concretely, the classical Hamiltonian is the random
function
\[
H_{n,p}^{\mathrm{cl}}(\eta)
=
\binom{n}{p}^{-1/2}
\sum_{1\le i_1<\cdots<i_p\le n}
g_{i_1,\ldots,i_p}
\eta_{i_1}\cdots \eta_{i_p},
\qquad
\eta\in\{\pm1\}^n,
\]
where the coefficients \(g_{i_1,\ldots,i_p}\) are independent standard
Gaussian random variables. Thus \(H_{n,p}^{\mathrm{cl}}\) is a Gaussian process
indexed by the hypercube. Its covariance is determined by the overlap between
spin configurations: for \(\eta,\tau\in\{\pm1\}^n\),
\[
    \mathbb E\bigl[
        H_{n,p}^{\mathrm{cl}}(\eta)
        H_{n,p}^{\mathrm{cl}}(\tau)
    \bigr]
    =
    \binom np^{-1}
    \sum_{1\le i_1<\cdots<i_p\le n}
    \prod_{s=1}^p \eta_{i_s}\tau_{i_s}.
\]
If $R(\eta,\tau):=\frac1n\sum_{i=1}^n \eta_i\tau_i$ denotes the \textit{overlap}, then this covariance is \(R(\eta,\tau)^p+O_p(n^{-1})\). Thus configurations with large overlap have
strongly correlated energies, while configurations with small overlap have
much weaker energy correlation.

A central problem is to understand the extremal energy of this random
landscape, for instance
\[
    \max_{\eta\in\{\pm1\}^n} H_{n,p}^{\mathrm{cl}}(\eta).
\]
Extremal energies describe the most favorable configurations of the system,
up to the sign convention for the Hamiltonian, and determine the low-temperature
behavior of the associated spin glass. From the perspective of random
optimization, this is the problem of maximizing a strongly correlated random
objective over exponentially many configurations.

The quantum \(p\)-spin model \cite{das2008colloquium, sachdev1999quantum, anschuetz2025bounds} is the natural quantum analogue of the classical
\(p\)-spin model. Whereas the classical model describes disordered systems of classical spins, the quantum model describes disordered many-body systems of interacting qubits, where
noncommutativity and entanglement play an essential role. Such models arise as
idealized mean-field models of quantum spin glasses and have been studied in
connection with quantum annealing, quantum optimization, and the statistical
mechanics of disordered quantum systems.

In the quantum model, the energy is no longer a random function on classical
configurations, but a random self-adjoint operator acting on the Hilbert space $\mcH_n=(\C^2)^{\otimes n}$. The mean-field Gaussian quantum \(p\)-spin Hamiltonian is defined in \eqref{eq:mean-field-Hamiltonian}. For a normalized quantum state \(|\psi\rangle\in\mcH_n\), its energy is $\langle \psi|H_{n,p}|\psi\rangle$. Thus the unrestricted quantum optimization problem is
\[
    \sup_{\psi\in\mcH_n:\,\|\psi\|=1}
    \langle \psi|H_{n,p}|\psi\rangle .
\]
This formulation contains the classical \(p\)-spin model as a special
commuting case. To see this, suppose that only the \(Z\)-Pauli strings are
retained, namely consider
\[
    H_{n,p}^{Z}
    :=
    \binom np^{-1/2}
    \sum_{I=(i_1,\dots,i_p)\in\mcI_p^n}
    g_I
    \prod_{s=1}^p \sigma_{i_s}^3 .
\]
Then all summands commute and \(H_{n,p}^{Z}\) is diagonal in the computational
basis. For a basis vector \(|x\rangle\), \(x\in\{0,1\}^n\), set $\eta_i:=(-1)^{x_i}\in\{\pm1\}$. Since \(\sigma_i^3|x\rangle=\eta_i|x\rangle\), we have
\[
    \langle x|H_{n,p}^{Z}|x\rangle
    =
    \binom np^{-1/2}
    \sum_{I=(i_1,\dots,i_p)\in\mcI_p^n}
    g_I
    \prod_{s=1}^p \eta_{i_s}.
\]
Thus the diagonal entries of \(H_{n,p}^{Z}\) are exactly the classical
\(p\)-spin Hamiltonian values. Classical spin glasses are already random and frustrated: there need not exist a spin configuration that simultaneously optimizes all summands of the Hamiltonian. The genuinely quantum feature of
\eqref{eq:mean-field-Hamiltonian} is therefore not randomness or frustration
alone, but the presence of many Pauli strings which generally do not commute.

Since the Pauli-string summands in \(H_{n,p}\) generally do not commute, the
Hamiltonian is not diagonal in any fixed product basis and cannot be identified
with a random function on classical spin configurations. The optimization is
therefore over the full quantum state space, and the comparison between
unrestricted states, product states, and low-depth circuit states measures the
role of entanglement and circuit complexity in attaining large energy.

\subsection{Comparison with the ABN construction}
\label{ssec:ABN}
In this section, we compare our setting with the code-based NLTS construction of Anshu,
Breuckmann, and Nirkhe~\cite{anshu2023nlts}. 

A stabilizer Hamiltonian is built from mutually commuting Hermitian Pauli strings. Each such Pauli string \(S\)
defines a stabilizer check: the desired states are those satisfying the
constraint \(S|\psi\rangle=|\psi\rangle\). Equivalently, they lie in the
\(+1\)-eigenspace of \(S\). Since \(S^2=\mathrm{Id}_n\), the only possible eigenvalues of
\(S\) are \(\pm1\). The operator \((\mathrm{Id}_n-S)/2\) is therefore the projector onto
the \((-1)\)-eigenspace of \(S\). Indeed, if
\(S|\psi\rangle=-|\psi\rangle\), then
\[
    \frac{\mathrm{Id}_n-S}{2}|\psi\rangle=|\psi\rangle,
\]
whereas if \(S|\psi\rangle=|\psi\rangle\), then
\[
    \frac{\mathrm{Id}_n-S}{2}|\psi\rangle=0.
\]
Thus \((\mathrm{Id}_n-S)/2\) penalizes violation of the stabilizer check, and a sum of such
projectors penalizes states that violate one or more of the checks. When the
checks commute and have a common \(+1\)-eigenspace, the Hamiltonian is
frustration-free: its ground states simultaneously minimize every local term.

The construction of~\cite{anshu2023nlts} is based on Calderbank--Shor--Steane (CSS) quantum low-density parity-check (LDPC) codes, a class of quantum error-correcting codes. In this setting, CSS means that the stabilizer checks split into two families: \(\sigma^1\)-type checks and \(\sigma^3\)-type checks.
The LDPC condition means that each check acts on only \(O(1)\) qubits and that
each qubit participates in only \(O(1)\) checks. For binary
vectors \(v,w\in\mathbb F_2^n\), write
\[
    \sigma^{1,v}:=\prod_{i:v_i=1}\sigma_i^1,
    \qquad
    \sigma^{3,w}:=\prod_{i:w_i=1}\sigma_i^3 .
\]
The CSS compatibility condition ensures that the \(\sigma^1\)-type and
\(\sigma^3\)-type checks commute. The associated stabilizer Hamiltonian has
the form
\[
    H_{\mathrm{ABN}}
    =
    \sum_{v\in \mathcal R_1} \frac{\mathrm{Id}_n-\sigma^{1,v}}{2}
    +
    \sum_{w\in \mathcal R_3} \frac{\mathrm{Id}_n-\sigma^{3,w}}{2},
\]
where \(\mathcal R_1\) and \(\mathcal R_3\) are the sets of
\(\sigma^1\)- and \(\sigma^3\)-checks. The LDPC condition means that the
checks have bounded weight and that each qubit participates in only a bounded number of checks. Hence \(H_{\mathrm{ABN}}\) is an \(O(1)\)-local, bounded-degree, commuting, frustration-free Hamiltonian with \(\Theta(n)\) terms and ground energy zero. The result of~\cite{anshu2023nlts} shows that, for a suitable family of such Hamiltonians, every state of energy at most \(\epsilon n\), for some constant \(\epsilon>0\), requires circuit depth \(\Omega(\log n)\).

The Hamiltonians studied in this paper are structurally different. Unlike the code Hamiltonians above, they are random, disordered, and generally
noncommuting. Indeed, if two Pauli strings overlap on a qubit where their Pauli labels anticommute, then the two strings themselves anticommute; for example, $\sigma_i^1\sigma_i^3=-\sigma_i^3\sigma_i^1$. Thus these Hamiltonians do not have a commuting stabilizer-check description or a planted code space.

This distinction is not merely formal. In the code-based setting, low-energy states are constrained by the geometry of the underlying quantum code. The proof in~\cite{anshu2023nlts} exploits this structure by relating low-energy states to approximate codewords obtained from measurements in the
\(\sigma^1\)- and \(\sigma^3\)-bases, and then using clustering properties of those approximate codewords to rule out shallow-circuit preparation. Our model has no analogous codeword structure. The obstruction instead comes from the interaction between spin-glass disorder and the limited light-cones of shallow circuits.

In this sense, the quantum \(p\)-spin model provides a complementary testing approach for NLTS-type questions. It is much less structured than the known code-based constructions, but it is more directly tied to random quantum
optimization and disordered many-body physics.

\end{appendix}

\section*{Acknowledgements}
OAG was supported by funding from the Eric and Wendy Schmidt Center at the Broad Institute of MIT and Harvard. DG was funded by ONR grant N000142512545.

\bibliographystyle{amsplain}
\bibliography{bibliography-05.2026}

\end{document}